\documentclass[prl,letterpaper,english,reprint,nofootinbib,aps,superscriptaddress,showpacs,showkeys]{revtex4-1}

\usepackage{babel,calc,amsmath,amsthm,amssymb,graphicx,subfigure,xcolor,comment}
\usepackage{mathdots}
\usepackage{dsfont}
\usepackage[T1]{fontenc}
\usepackage[unicode=true]{hyperref}
\hypersetup{
	colorlinks=true,       		% false: boxed links; true: colored links
	linkcolor=red,          	% color of internal links
	citecolor=blue,             % color of links to bibliography
	urlcolor=blue,          % color of external links
}

\usepackage{booktabs}

\usepackage{makecell}
\usepackage{braket}
\usepackage[normalem]{ulem}

\newif\ifdebug

\debugtrue

\ifdebug
\definecolor{zhliu}{rgb}{0.48, 0.12, 0}

\newcommand{\note}[1]{\textcolor{zhliu}{#1}}
\newcommand\delete{\bgroup\markoverwith{\textcolor{zhliu}{\rule[0.5ex]{2pt}{0.8pt}}}\ULon}

\else

\newcommand{\note}[1]{\ignorespaces}
\newcommand{\delete}[1]{\ignorespaces}

\fi

\begin{document}

	\renewcommand{\figurename}{Fig.}
	
	\title{Generalized Einstein-Podolsky-Rosen Steering Paradox}

	\author{Zhi-Jie Liu}
	%\email{1120220019@mail.nankai.edu.cn}
	\affiliation{Theoretical Physics Division, Chern Institute of Mathematics and LPMC, Nankai University, Tianjin 300071, People's Republic of China}

	\author{Xing-Yan Fan}
	\affiliation{Theoretical Physics Division, Chern Institute of Mathematics and LPMC, Nankai University, Tianjin 300071, People's Republic of China}

	\author{Jie Zhou}
	\affiliation{College of Physics and Materials Science, Tianjin Normal University, Tianjin 300382, People's Republic of China}

	\author{Mi Xie}
	%\email{xiemi@tju.edu.cn}
	\affiliation{Department of Physics, School of Science, Tianjin University, Tianjin 300072,  People's Republic of China}

	\author{Jing-Ling~Chen}
	\email{chenjl@nankai.edu.cn}
	\affiliation{Theoretical Physics Division, Chern Institute of Mathematics and LPMC, Nankai University, Tianjin 300071, People's Republic of China}

	\date{\today}

\begin{abstract}
Quantum paradoxes are essential means to reveal the incompatibility between
quantum and classical theories, among which the Einstein-Podolsky-Rosen (EPR)
steering paradox offers a sharper criterion for the contradiction between
local-hidden-state model and quantum mechanics than the usual inequality-based
method. In this work, we present a generalized EPR steering paradox, which predicts
a contradictory equality ``$2_{Q}=\left(  1+\delta\right)
_{C}$''\ ($0\leq\delta<1$) given by the quantum ($Q$) and
classical ($C$) theories. For any $N$-qubit state in which the conditional
state of the steered party is pure, we test the paradox through a two-setting
steering protocol, and find that the state is steerable if some specific
measurement requirements are satisfied. Moreover, our
construction also enlightens the building of EPR steering inequality, which
may contribute to some schemes for typical quantum teleportation and quantum key
distributions.
\end{abstract}
	
%	\pacs{03.65.Ud, 03.67.Mn}
	
\maketitle

\section{INTRODUCTION}

The concept of quantum nonlocality araised from
the seminal paper of Einstein, Podolsky, and Rosen (EPR) \cite{EPR1935}, which
questioned the completeness of quantum mechanical description for physical
reality under the assumptions of locality and reality. As a peculiar
nonlocality, EPR steering, also known as `` the spooky action
at a distance''\ \cite{Schr1935,Schr1936}, was initially
imported by Schr\"{o}dinger in his reply to EPR's paper \cite{EPR1935}. After
that, the notion of steering remained undiscussed for a long time, until
Vujicic and  Herbut extended Schr\"{o}dinger's results to the system
of continuous variable \cite{VM1988}. Subsequently, Reid proved the
possibility of the EPR paradox by performing orthogonal phase measurements on
the two output beams of a non-simplex parametric amplifier, which led to an
EPR-type argument \cite{RMD1989}. In addition, Verstraete pointed out the
connection of Schr\"{o}dinger's idea with quantum teleportation and
entanglement transformation \cite{VFPhd2002}.

Mathematically, quantum nonlocality \cite{Q-nonlocality1994} is defined by the
corresponding classical model. For example, entanglement \cite{Quantum
entanglement2009} is defined by the separable state model
\cite{separable2002,separable2008}; the classical model corresponding to Bell
nonlocality \cite{Bell nonlocality2014} is the local-hidden-variable (LHV)
model \cite{LHV1976,LHV2014,LHV2006}. In order to depict EPR steering
rigorously, Wiseman \emph{et al.} proposed the local-hidden-state (LHS) model
and thus gave an operational definition of steering using a task with two
parties (say Alice and Bob) sharing a quantum state
\cite{Wiseman2007,jones2007}; the immediate effect of this nonlocality is that
Alice can steer the particles of Bob, according to the measurement postulation
of quantum mechanics. However, the LHS model assumes that the particles in
Bob's hand is specific, though Bob does not know it. Therefore, Alice's
ability to steer Bob's state is an illusion that cannot be observed
experimentally. If Alice convince Bob that the two particles in their hands
are entangled and the LHS model cannot describe the particles in Bob's hand,
then Alice is capable of (EPR) steering Bob. These works
\cite{Wiseman2007,jones2007} have stimulated interest in the research of EPR
steering, for examples, how to

\begin{itemize}
\item[(i).] determine the steerability of a given state
\cite{DJ2010,MHX2018,MHX20182,AVN2013,chenjl2016,Feng2021,Liu2021,EPR
paradox2024};

\item[(ii).] characterize the border of quantum steering with the other two
types of nonlocality \cite{CD2011,AA2006};

\item[(iii).] apply steering to realistic quantum information tasks
\cite{QKD2016,QSS2015,QT2015,RM2018,SC2015,HMP2021}.
\end{itemize}

For the research type-(i) mentioned above, there are generally two
effective methods to determine whether a quantum state is steerable. The first
one is based on the EPR steering inequality, for instance, the linear EPR
steering inequality \cite{DJ2010}, and the chained EPR steering inequalities
\cite{MHX2018,MHX20182}, which were constructed compared to the Bell chain
inequality \cite{Bellchain}, optimizing the visibility of the Werner state
\cite{werner}. The other is the approach without inequalities, also dubbed the
all-versus-nothing (AVN) proof, which reveals the contradiction between
classical theory and quantum mechanics more intuitively than the inequality
approach. The AVN proof \cite{AVN1} was first built as an elegant argument for
the nonexistence of the LHV model. Analogous to the AVN argument for Bell
nonlocality without inequalities, Chen \emph{et al.} presented the AVN proof
of EPR steering for any two-qubit entangled state based on a two-setting
steering protocol \cite{AVN2013}. They then extracted the method for EPR
steering paradox as the contradictory equality ``$2_{Q}=1_{C}%
$''\ (``$2_{Q}$''\ is the
quantum prediction, and ``$1_{C}$''\ is the
result of the LHS model) for any two-qubit pure state in the two-setting
steering protocol \cite{chenjl2016}. Subsequently, the EPR steering paradox
approach has been further investigated, e.g., ``$2_{Q}=1_{C}%
$''\ for a particular four-qubit mixed state \cite{Liu2021},
``$k_{Q}=1_{C}$''\ for an arbitrary two-qubit
pure state in the $k$-setting steering protocol {\cite{Feng2021}}, and
``$2_{Q}=1_{C}$''\ for an arbitrary $N$-qubit
entangled state {\cite{EPR paradox2024}}.

In this work, we propose a general EPR steering paradox, expressed as
``$2_{Q}=\left(  1+\delta\right)  _{C}$''%
\ ($0\leq\delta<1$), for an arbitrary $N$-qubit entangled state, both pure and
mixed. Considering that the conditional states of the steered party (Bob) are
all pure, we propose a theorem for any $N$-qubit quantum steerable states in the
two-setting protocol. We analyze all possible four cases of Bob's conditional
states, corresponding to $\delta=0$ and $0<\delta<1$ respectively. It is noteworthy to
mention that we find all the steerable states when Bob's conditional states
are pure.

\vspace{12pt}

\section{RESULTS}

\subsection{A theorem on the EPR steering paradox}

In order to investigate the existence of steering trait for a quantum state,
we propose the generalized EPR steering paradox; the contradiction between the
LHS model and EPR steering is used to determine whether one side, Alice, can
steer the other side, Bob. The main result is summarized into a theorem:

\emph{Theorem 1.}---In the two-setting steering protocol $\left\{  \hat{n}%
_{1},\hat{n}_{2}\right\}  $, Alice and Bob share an $N$-qubit state $\rho
_{AB}$. Assume that Alice measures along $\hat{n}_{1}$ and $\hat{n}_{2}$,
corresponding to the results $a$, and $a^{\prime}$; then Bob obtains the
conditional states $\tilde{\rho}_{a}^{\hat{n}_{1}}$ and $\tilde{\rho
}_{a^{\prime}}^{\hat{n}_{2}}$, respectively. Considering the scenario where
both sets of Bob's conditional states $\{\tilde{\rho}_{a}^{\hat{n}_{1}}\}$ and
$\{\tilde{\rho}_{a^{\prime}}^{\hat{n}_{2}}\}$ are all pure, there will be a
contradiction of ``$2_{Q}=\left(  1+\delta\right)  _{C}%
$''\ $\left(  \text{with }0\leq\delta<1\right)  $\ if
$\rho_{AB}$ satisfies `` the measurement
requirement'', viz.,

\begin{itemize}
\item Bob obtains two sets of results $\{\tilde{\rho}_{a}^{\hat{n}_{1}}\}$ and
$\{\tilde{\rho}_{a^{\prime}}^{\hat{n}_{2}}\}$ that are not exactly the same,
i.e., $\{\tilde{\rho}_{a}^{\hat{n}_{1}}\}\neq\{\tilde{\rho}_{a^{\prime}}%
^{\hat{n}_{2}}\}$.
\end{itemize}

%\fxy{Note that our theorem merely applys to the situation} that Bob's conditional states \fxy{are} pure \fxy{and} not identical after Alice's measurement.
The complete proof of the theorem is given in Appendix A. Notice that
the EPR steering paradox ``$2_{Q}=\left(  1+\delta\right)
_{C}$''\ mentioned in the theorem can be categorized into
$\delta=0$ and $\delta\neq0$, based on which some analysis and examples of the
theorem are given below.

\subsection{Some cases of the theorem}

\emph{Case. 1}---The sets of Bob's conditional states $\{\tilde{\rho}
_{a}^{\hat{n}_{1}}\}$ and $\{\tilde{\rho}_{a^{\prime}}^{\hat{n}_{2}}\}$ are
all different.

\emph{Example 1.}---We consider the example of an arbitrary two-qubit pure
state presented by Chen \emph{et al.} \cite{chenjl2016}. Alice prepares an
arbitrary two-qubit pure state $\rho_{AB}=\left\vert \Psi\left(
\theta\right)  \right\rangle \left\langle \Psi\left(  \theta\right)
\right\vert $, where
\begin{equation}
\left\vert \Psi\left(  \theta\right)  \right\rangle =\cos\theta\left\vert
0\right\rangle \left\vert 0\right\rangle +\sin\theta\left\vert 1\right\rangle
\left\vert 1\right\rangle , \label{eq2.01}%
\end{equation}
with $\theta\in\left(  0,\pi/2\right)  $; she keeps one of the particle in Eq.
(\ref{eq2.01}), and sent the other to Bob. In the two-setting steering
protocol $\left\{  \hat{z},\hat{x}\right\}  $, Alice performs four possible
projective measurements
\begin{equation}%
\begin{array}
[c]{l}%
{{P}_{0}^{\hat{z}}=}\left\vert 0\right\rangle \left\langle 0\right\vert ,\\
{{P}_{1}^{\hat{z}}=}\left\vert 1\right\rangle \left\langle 1\right\vert ,
\end{array}%
\begin{array}
[c]{l}%
{{P}_{0}^{\hat{x}}=}\left\vert +\right\rangle \left\langle +\right\vert ,\\
{{P}_{1}^{\hat{x}}=}\left\vert -\right\rangle \left\langle -\right\vert ,
\end{array}
\label{eq2.02}%
\end{equation}
where $\left\vert \pm\right\rangle =1/\sqrt{2}\left(  \left\vert
0\right\rangle \pm\left\vert 1\right\rangle \right)  $. After Alice's
measurements, the four unnormalized conditional states corresponding to Bob
are
\begin{equation}%
\begin{array}
[c]{l}%
\tilde{\rho}{_{0}^{\hat{z}}=\cos}^{2}\theta\left\vert 0\right\rangle
\left\langle 0\right\vert ,\\
\tilde{\rho}{_{1}^{\hat{z}}=\sin}^{2}{\theta}\left\vert 1\right\rangle
\left\langle 1\right\vert ,
\end{array}%
\begin{array}
[c]{l}%
\tilde{\rho}_{0}^{\hat{x}}=\frac{1}{2}\left\vert \lambda_{+}\right\rangle
\left\langle \lambda_{\newline+}\right\vert ,\\
\tilde{\rho}_{1}^{\hat{x}}=\frac{1}{2}\left\vert \lambda_{-}\right\rangle
\left\langle \lambda_{\newline-}\right\vert ,
\end{array}
\label{eq2.03}%
\end{equation}
where $\left\vert \lambda_{\pm}\right\rangle =\cos\theta\left\vert
0\right\rangle \pm{\sin\theta}\left\vert 1\right\rangle $. It is worth noting
that Bob's conditional states are all pure states and none of them are
identical, which relate to the measurement requirement mentioned in Theorem 1.

The mathematical definition of the LHS model is given by Wiseman \emph{et al.}
\cite{Wiseman2007}. If the LHS model can describe the state assemblage of
Bob's unnormalized conditional states, then Bob's state can be described by
\begin{equation}
\tilde{\rho}_{a}^{\hat{n}}=%
%TCIMACRO{\dsum \limits_{\xi}}%
%BeginExpansion
{\displaystyle\sum\limits_{\xi}}
%EndExpansion
\wp\left(  a|\hat{n},\xi\right)  \wp_{\xi}\rho_{\xi}, \label{eq2.1}%
\end{equation}
and
\begin{equation}%
%TCIMACRO{\dsum \limits_{\xi}}%
%BeginExpansion
{\displaystyle\sum\limits_{\xi}}
%EndExpansion
\wp_{\xi}\rho_{\xi}=\rho_{B}. \label{eq2.2}%
\end{equation}
Here $\wp_{\xi}$ and $\wp(a|\hat{n},\xi)$ are probabilities satisfying
${\sum_{\xi}\wp_{\xi}}=1$, and ${\sum_{a}}\wp(a|\hat{n},\xi)=1$, for a fixed
$\xi$, and $\rho_{B}=\mathrm{tr}_{A}\left(  \rho_{AB}\right)  $ is Bob's
reduced density matrix. Because Bob's unnormalized conditional states in Eq.
(\ref{eq2.03}) are four different pure states, the state ensemble can be taken
as $\{\wp_{\xi}\rho_{\xi}\}=\{\wp_{1}\rho_{1},\wp_{2}\rho_{2},\wp_{3}\rho
_{3},\wp_{4}\rho_{4}\}$. And it is a well-established tenet that a pure state
cannot be obtained as a convex combination of other distinct states, which
suggests
\begin{equation}
\left\{
\begin{array}
[c]{l}%
\tilde{\rho}_{0}^{\hat{z}}=\wp_{1}\rho_{1},\\
\tilde{\rho}_{1}^{\hat{z}}=\wp_{2}\rho_{2},\\
\tilde{\rho}_{0}^{\hat{x}}=\wp_{3}\rho_{3},\\
\tilde{\rho}_{1}^{\hat{x}}=\wp_{4}\rho_{4}.
\end{array}
\right.  \label{eq2.04}%
\end{equation}
Here $\wp\left(  0|\hat{z},1\right)  =\wp\left(  1|\hat{z},2\right)
=\wp\left(  0|\hat{x},3\right)  =\wp\left(  1|\hat{x},4\right)  =1$, and other
$\wp\left(  a|\hat{n},\xi\right)  =0$. Summing and taking the trace of Eq.
(\ref{eq2.04}) gives
\begin{equation}
\mathrm{tr}\left(  \tilde{\rho}{^{\hat{z}}}+\tilde{\rho}{^{\hat{x}}}\right)
=2\mathrm{tr}\left(  \rho_{B}\right)  =2_{Q}%
\end{equation}
on the left, and
\begin{equation}
\mathrm{tr}\left(  \wp_{1}\rho_{1}+\wp_{2}\rho_{2}+\wp_{3}\rho_{3}+\wp_{4}%
\rho_{4}\right)  =\mathrm{tr}\left(  \rho_{B}\right)  =1_{C}%
\end{equation}
on the right. Then the contradiction ``$2_{Q}=1_{C}%
$''\ is obtained. In this case EPR steering paradox is
formulated as ``$2_{Q}=1_{C}$''\ with
$\delta=0$.

\emph{Case. 2}---Bob has the same conditional states in the same set of measurements.

\emph{Example 2.}---We consider the example of a four-qubit mixed entangled
state presented by Liu \emph{et al.} \cite{Liu2021}. Alice prepares the state
$\rho_{AB}=\cos^{2}\theta\left\vert LC_{4}\right\rangle \left\langle
LC_{4}\right\vert +\sin^{2}\theta\left\vert LC_{4}^{\prime}\right\rangle
\left\langle LC_{4}^{\prime}\right\vert $, where
\begin{equation}%
\begin{array}
[c]{c}%
\left\vert LC_{4}\right\rangle =\dfrac{1}{2}\left(  \left\vert
0000\right\rangle +\left\vert 1100\right\rangle +\left\vert 0011\right\rangle
-\left\vert 1111\right\rangle \right)  ,\\[9pt]%
\left\vert LC_{4}^{\prime}\right\rangle =\dfrac{1}{2}\left(  \left\vert
0100\right\rangle +\left\vert 1000\right\rangle +\left\vert 0111\right\rangle
-\left\vert 1011\right\rangle \right)  ,
\end{array}
\label{eq2.06}%
\end{equation}
are the linear cluster states \cite{clus2005}. She keeps the particles 1, 2
and sends the particles 3, 4 to Bob. In the two-setting steering protocol
$\left\{  \hat{n}_{1},\hat{n}_{2}\right\}  $ $\left(  \hat{n}_{1}\neq\hat
{n}_{2}\right)  $,
\begin{equation}%
\begin{array}
[c]{c}%
\hat{n}_{1}=\sigma_{z}\sigma_{z}\equiv zz,\\
\hat{n}_{2}=\sigma_{y}\sigma_{x}\equiv yx.
\end{array}
\label{eq2.07}%
\end{equation}
After Alice's measurements, Bob's unnormalized conditional states are
\begin{equation}%
\begin{array}
[c]{l}%
\tilde{\rho}_{00}^{\hat{n}_{1}}=\frac{1}{4}\cos^{2}\theta\left(  \left\vert
00\right\rangle +\left\vert 11\right\rangle \right)  \left(  \left\langle
00\right\vert +\left\langle 11\right\vert \right)  ,\\[9pt]%
\tilde{\rho}_{01}^{\hat{n}_{1}}=\frac{1}{4}\sin^{2}\theta\left(  \left\vert
00\right\rangle +\left\vert 11\right\rangle \right)  \left(  \left\langle
00\right\vert +\left\langle 11\right\vert \right)  ,\\[9pt]%
\tilde{\rho}_{10}^{\hat{n}_{1}}=\frac{1}{4}\sin^{2}\theta\left(  \left\vert
00\right\rangle -\left\vert 11\right\rangle \right)  \left(  \left\langle
00\right\vert -\left\langle 11\right\vert \right)  ,\\[9pt]%
\tilde{\rho}_{11}^{\hat{n}_{1}}=\frac{1}{4}\cos^{2}\theta\left(  \left\vert
00\right\rangle -\left\vert 11\right\rangle \right)  \left(  \left\langle
00\right\vert -\left\langle 11\right\vert \right)  ,\\[9pt]%
\tilde{\rho}_{00}^{\hat{n}_{2}}=\frac{1}{8}\left(  \left\vert 00\right\rangle
+\mathrm{i}\left\vert 11\right\rangle \right)  \left(  \left\langle
00\right\vert -\mathrm{i}\left\langle 11\right\vert \right)  ,\\[9pt]%
\tilde{\rho}_{01}^{\hat{n}_{2}}=\frac{1}{8}\left(  \left\vert 00\right\rangle
-\mathrm{i}\left\vert 11\right\rangle \right)  \left(  \left\langle
00\right\vert +\mathrm{i}\left\langle 11\right\vert \right)  ,\\[9pt]%
\tilde{\rho}_{10}^{\hat{n}_{2}}=\frac{1}{8}\left(  \left\vert 00\right\rangle
-\mathrm{i}\left\vert 11\right\rangle \right)  \left(  \left\langle
00\right\vert +\mathrm{i}\left\langle 11\right\vert \right)  ,\\[9pt]%
\tilde{\rho}_{11}^{\hat{n}_{2}}=\frac{1}{8}\left(  \left\vert 00\right\rangle
+\mathrm{i}\left\vert 11\right\rangle \right)  \left(  \left\langle
00\right\vert -\mathrm{i}\left\langle 11\right\vert \right)  ,
\end{array}
\label{2.09}%
\end{equation}
which are not all the same, and satisfy the requirement in Theorem 1.

If Bob's states have a LHS description, they must satisfy Eqs. (\ref{eq2.1})
and (\ref{eq2.2}). Because the eight states in the set of Bob's conditional
states (\ref{2.09}) are pure states, and there are only four different states
therein; it is sufficient to take $\xi$ from $1$ to $4$. Similarly, a pure
state cannot be obtained by the convex combination of other pure states, and
thus one has
\begin{equation}%
\begin{array}
[c]{c}%
\tilde{\rho}_{00}^{\hat{n}_{1}}=\wp\left(  00|\hat{n}_{1},1\right)  \wp
_{1}\rho_{1},\\
\tilde{\rho}_{01}^{\hat{n}_{1}}=\wp\left(  01|\hat{n}_{1},1\right)  \wp
_{1}\rho_{1},\\
\tilde{\rho}_{10}^{\hat{n}_{1}}=\wp\left(  10|\hat{n}_{1},2\right)  \wp
_{2}\rho_{2},\\
\tilde{\rho}_{11}^{\hat{n}_{1}}=\wp\left(  11|\hat{n}_{1},2\right)  \wp
_{2}\rho_{2},
\end{array}
\begin{array}
[c]{c}%
\tilde{\rho}_{00}^{\hat{n}_{2}}=\wp\left(  00|\hat{n}_{2},3\right)  \wp
_{3}\rho_{3},\\
\tilde{\rho}_{01}^{\hat{n}_{2}}=\wp\left(  01|\hat{n}_{2},4\right)  \wp
_{4}\rho_{4},\\
\tilde{\rho}_{10}^{\hat{n}_{2}}=\wp\left(  10|\hat{n}_{2},4\right)  \wp
_{4}\rho_{4},\\
\tilde{\rho}_{11}^{\hat{n}_{2}}=\wp\left(  11|\hat{n}_{2},3\right)  \wp
_{3}\rho_{3}.
\end{array}
\label{eq2.10}%
\end{equation}
Since in the set Eq. (\ref{2.09}), $\tilde{\rho}_{00}^{\hat{n}_{1}}
=\tilde{\rho}_{01}^{\hat{n}_{1}}$, $\tilde{\rho}_{10}^{\hat{n}_{1}}
=\tilde{\rho}_{11}^{\hat{n}_{1}}$, $\tilde{\rho}_{00}^{\hat{n}_{2}}
=\tilde{\rho}_{11}^{\hat{n}_{2}}$, and $\tilde{\rho}_{01}^{\hat{n}_{2}}
=\tilde{\rho}_{10}^{\hat{n}_{2}}$, in the LHS description we describe the
identical terms by the same hidden state. And ${\sum_{a}}\wp\left(  a|\hat
{n},\xi\right)  =1$, so we have $\wp\left(  00|\hat{n}_{1},1\right)
+\wp\left(  01|\hat{n}_{1},1\right)  =1$, $\wp\left(  10|\hat{n}_{1},2\right)
+\wp\left(  11|\hat{n}_{1},2\right)  =1$, $\wp\left(  00|\hat{n}_{2},3\right)
+\wp\left(  11|\hat{n}_{2},3\right)  =1$, $\wp\left(  01|\hat{n}_{2},4\right)
+\wp\left(  10|\hat{n}_{2},4\right)  =1$, and other $\wp\left(  a|\hat{n}
,\xi\right)  =0$. After taking the sum of Eq. (\ref{eq2.10}) and taking the
trace subsequently, we get ``$2_{Q}=1_{C}$''.
In the same way, we get the EPR steering paradox ``$2_{Q}
=1_{C}$''\ and $\delta=0$.

\emph{Case. 3}---Bob has the same conditional states in different sets of measurements.

\emph{Example 3.}---We consider the example of a 3-qubit pure entangled state
$\rho_{AB}=\left\vert \Psi^{\prime}\right\rangle \left\langle \Psi^{\prime
}\right\vert $, where
\begin{align}
\left\vert \Psi^{\prime}\right\rangle =\dfrac{1}{\sqrt{6}}\Big[  &  \left\vert
001\right\rangle +\left\vert 01\right\rangle \big(\left\vert 0\right\rangle
+\left\vert 1\right\rangle \big)\nonumber\\
&  +\left\vert 10\right\rangle \big(\left\vert 0\right\rangle -\left\vert
1\right\rangle \big)-\left\vert 110\right\rangle \Big]. \label{eq3.1}%
\end{align}
Alice prepares the state $\rho_{AB}$ as in Eq. (\ref{eq3.1}). She keeps the
particles 1, 2 and sends the particle 3 to Bob. In the two-setting steering
protocol $\left\{  \hat{n}_{1},\hat{n}_{2}\right\}  $ $\left(  \hat{n}_{1}%
\neq\hat{n}_{2}\right)  $, with
\begin{equation}%
\begin{array}
[c]{c}%
\hat{n}_{1}=\sigma_{z}\sigma_{z}\equiv zz,\\
\hat{n}_{2}=\sigma_{x}\sigma_{x}\equiv xx,
\end{array}
\label{3.12}%
\end{equation}
after Alice's measurements, Bob's unnormalized conditional states are
\begin{equation}%
\begin{array}
[c]{l}%
\tilde{\rho}_{00}^{\hat{n}_{1}}=\frac{1}{6}\left\vert 1\right\rangle
\left\langle 1\right\vert ,\\[9pt]%
\tilde{\rho}_{01}^{\hat{n}_{1}}=\frac{1}{6}\left(  \left\vert 0\right\rangle
+\left\vert 1\right\rangle \right)  \left(  \left\langle 0\right\vert
+\left\langle 1\right\vert \right)  ,\\[9pt]%
\tilde{\rho}_{10}^{\hat{n}_{1}}=\frac{1}{6}\left(  \left\vert 0\right\rangle
-\left\vert 1\right\rangle \right)  \left(  \left\langle 0\right\vert
-\left\langle 1\right\vert \right)  ,\\[9pt]%
\tilde{\rho}_{11}^{\hat{n}_{1}}=\frac{1}{6}\left\vert 0\right\rangle
\left\langle 0\right\vert ,\\[9pt]%
\tilde{\rho}_{00}^{\hat{n}_{2}}=\frac{1}{24}\left(  \left\vert 0\right\rangle
+\left\vert 1\right\rangle \right)  \left(  \left\langle 0\right\vert
+\left\langle 1\right\vert \right)  ,\\[9pt]%
\tilde{\rho}_{01}^{\hat{n}_{2}}=\frac{1}{24}\left(  \left\vert 0\right\rangle
-\left\vert 1\right\rangle \right)  \left(  \left\langle 0\right\vert
-\left\langle 1\right\vert \right)  ,\\[9pt]%
\tilde{\rho}_{10}^{\hat{n}_{2}}=\frac{1}{24}\left(  \left\vert 0\right\rangle
+3\left\vert 1\right\rangle \right)  \left(  \left\langle 0\right\vert
+3\left\langle 1\right\vert \right)  ,\\[9pt]%
\tilde{\rho}_{11}^{\hat{n}_{2}}=\frac{1}{24}\left(  3\left\vert 0\right\rangle
-\left\vert 1\right\rangle \right)  \left(  3\left\langle 0\right\vert
-\left\langle 1\right\vert \right)  .
\end{array}
\label{eq3.2}%
\end{equation}
Notice that all of Bob's states are pure, while there are identical results in
$\{\tilde{\rho}_{a}^{\hat{n}_{1}}\}$ and $\{\tilde{\rho}_{a^{\prime}}^{\hat
{n}_{2}}\}$, which implies that Bob's two sets of results are not identical,
satisfying the requirement of Theorem 1.

Suppose Bob's states have a LHS description, and then they must satisfy Eqs.
(\ref{eq2.1}) and (\ref{eq2.2}). Because the eight states of Eq. (\ref{eq3.2})
are pure states, and there are only six different states in the measurement
result Eq. (\ref{eq3.2}), it is sufficient to take $\xi$ from $1$ to $6$.
Analogously, one has
\begin{equation}%
\begin{array}
[c]{c}%
\tilde{\rho}_{00}^{\hat{n}_{1}}=\wp_{1}\rho_{1},\\
\tilde{\rho}_{01}^{\hat{n}_{1}}=\wp_{2}\rho_{2},\\
\tilde{\rho}_{10}^{\hat{n}_{1}}=\wp_{3}\rho_{3},\\
\tilde{\rho}_{11}^{\hat{n}_{1}}=\wp_{4}\rho_{4},
\end{array}%
\begin{array}
[c]{c}%
\tilde{\rho}_{00}^{\hat{n}_{2}}=\wp_{2}\rho_{2},\\
\tilde{\rho}_{01}^{\hat{n}_{2}}=\wp_{3}\rho_{3},\\
\tilde{\rho}_{10}^{\hat{n}_{2}}=\wp_{5}\rho_{5},\\
\tilde{\rho}_{11}^{\hat{n}_{2}}=\wp_{6}\rho_{6}.
\end{array}
\label{eq3.3}%
\end{equation}
Here we describe the identical terms by the same hidden states. $\wp\left(
00|\hat{n}_{1},1\right)  =\wp\left(  01|\hat{n}_{1},2\right)  =\wp\left(
10|\hat{n}_{1},3\right)  =\wp\left(  11|\hat{n}_{1},4\right)  =\wp\left(
00|\hat{n}_{2},2\right)  =\wp\left(  01|\hat{n}_{2},3\right)  =\wp\left(
10|\hat{n}_{2},5\right)  =\wp\left(  11|\hat{n}_{2},6\right)  =1$, and other
$\wp\left(  a|\hat{n},\xi\right)  =0$. Summing Eq. (\ref{eq3.3}), the quantum
result is $2\rho_{B}$, while the classical result is $\rho_{B}+\wp_{2}\rho
_{2}+\wp_{3}\rho_{3}$. Then taking the trace, we obtain ``%
$2_{Q}=\left(  1+\wp_{2}+\wp_{3}\right)  _{C}$''. And
$\sum\nolimits_{\xi=1}^{6}\wp_{\xi}=1$, $\wp_{1}$, $\wp_{2}$, $\wp_{3}$,
$\wp_{4}$, $\wp_{5}$, together with $\wp_{6}$ are non-zero, thus $0<$ $\wp
_{2}+\wp_{3}<1$. In this case, we get the EPR steering paradox
``$2_{Q}=\left(  1+\delta\right)  _{C}$''\ and
$\delta=\wp_{2}+\wp_{3}$, with $0<\delta<1$.

\emph{Case. 4}---Bob has the same conditional states in the sets corresponding
to both the same and different directions of Alice's measurements.

\emph{Example 4.}---We consider the W state \cite{W2000} shared by Alice and
Bob. Suppose Alice prepares the W state $\rho_{AB}=\left\vert W\right\rangle
\left\langle W\right\vert $, where
\[
\left\vert W\right\rangle =\frac{1}{\sqrt{3}}\Big[\left\vert 100\right\rangle
+\left\vert 010\right\rangle +\left\vert 001\right\rangle \Big].
\]
She keeps the particles 1, 2, and sends the other to Bob. In the two-setting
steering protocol $\left\{  \hat{n}_{1},\hat{n}_{2}\right\}  $, where $\hat
{n}_{1}$, $\hat{n}_{2}$ are shown in Eq. (\ref{3.12}), after Alice's
measurements, Bob's unnormalized conditional states are
\begin{equation}%
\begin{array}
[c]{l}%
\tilde{\rho}_{00}^{zz}=\frac{1}{3}\left\vert 1\right\rangle \left\langle
1\right\vert ,\\
\tilde{\rho}_{01}^{zz}=\frac{1}{3}\left\vert 0\right\rangle \left\langle
0\right\vert ,\\
\tilde{\rho}_{10}^{zz}=\frac{1}{3}\left\vert 0\right\rangle \left\langle
0\right\vert ,\\
\tilde{\rho}_{11}^{zz}=0,\\
\tilde{\rho}_{00}^{xx}=\frac{1}{12}\left(  2\left\vert 0\right\rangle
+\left\vert 1\right\rangle \right)  \left(  2\left\langle 0\right\vert
+\left\langle 1\right\vert \right)  ,\\
\tilde{\rho}_{01}^{xx}=\frac{1}{12}\left\vert 1\right\rangle \left\langle
1\right\vert ,\\
\tilde{\rho}_{10}^{xx}=\frac{1}{12}\left\vert 1\right\rangle \left\langle
1\right\vert ,\\
\tilde{\rho}_{11}^{xx}=\frac{1}{12}\left(  2\left\vert 0\right\rangle
-\left\vert 1\right\rangle \right)  \left(  2\left\langle 0\right\vert
-\left\langle 1\right\vert \right)  ,
\end{array}
\label{d1}%
\end{equation}
namely, Bob obtains seven pure states. This satisfies the requirement in
Theorem 1. Summing all the equations in Eq. (\ref{d1}) together, and then
taking the trace, we get $\mathrm{tr}\left(  2\rho_{B}\right)  =2_{Q}$.

If Bob's states have a LHS description, they satisfy Eqs. (\ref{eq2.1}) and
(\ref{eq2.2}). In Eq. (\ref{d1}), there are four pure states, which means that
it is sufficient to take $\xi$ from $1$ to $4$. Because the conditional states
obtained by Bob are all pure states, and a pure state can only be expanded by
itself, and then one selects
\begin{equation}%
\begin{array}
[c]{l}%
\tilde{\rho}_{00}^{zz}=\wp_{1}\rho_{1},\\
\tilde{\rho}_{01}^{zz}=\wp\left(  01|\hat{z}\hat{z},2\right)  \wp_{2}\rho
_{2},\\
\tilde{\rho}_{10}^{zz}=\wp\left(  10|\hat{z}\hat{z},2\right)  \wp_{2}\rho
_{2},\\
\tilde{\rho}_{11}^{zz}=0,\\
\tilde{\rho}_{00}^{xx}=\wp_{3}\rho_{3},\\
\tilde{\rho}_{01}^{xx}=\wp\left(  01|\hat{x}\hat{x},1\right)  \wp_{1}\rho
_{1},\\
\tilde{\rho}_{10}^{xx}=\wp\left(  10|\hat{x}\hat{x},1\right)  \wp_{1}\rho
_{1},\\
\tilde{\rho}_{11}^{xx}=\wp_{4}\rho_{4}.
\end{array}
\label{d3}%
\end{equation}
Perceive that in Eq. (\ref{d1}), $\tilde{\rho}_{00}^{zz}=\tilde{\rho}%
_{01}^{xx}=\tilde{\rho}_{10}^{xx}$, so we take the same hide state $\wp
_{1}\rho_{1}$ to describe $\tilde{\rho}_{00}^{zz}$, $\tilde{\rho}_{01}^{xx}$,
and $\tilde{\rho}_{10}^{xx}$. Similarly, $\tilde{\rho}_{01}^{zz}=\tilde{\rho
}_{10}^{zz}$, and we take $\wp_{2}\rho_{2}$ to describe them. Because
$\sum\nolimits_{a}\wp\left(  a|\hat{n},\xi\right)  =1$, we have $\wp\left(
00|zz,1\right)  =\wp\left(  00|xx,3\right)  =\wp\left(  11|xx,4\right)  =1$,
$\wp\left(  01|\hat{z}\hat{z},2\right)  +\wp\left(  10|\hat{z}\hat
{z},2\right)  =1$, $\wp\left(  01|\hat{x}\hat{x},1\right)  +\wp\left(
10|\hat{x}\hat{x},1\right)  =1$, and other $\wp\left(  a|\hat{n},\xi\right)
=0$. We sum up Eq. (\ref{d3}) and take the trace, which leads the right-hand
side to $\left(  1+\wp_{1}\right)  _{C}$. Since $\tilde{\rho}_{00}^{zz}$,
$\tilde{\rho}_{01}^{zz}$, $\tilde{\rho}_{10}^{zz}$, $\tilde{\rho}_{00}^{xx}$,
$\tilde{\rho}_{01}^{xx}$, $\tilde{\rho}_{10}^{xx}$, and $\tilde{\rho}%
_{11}^{xx}$ are nonzero in Eq. (\ref{d1}), $\wp_{1}$, $\wp_{2}$, $\wp_{3}$,
and $\wp_{4}$ are all nonzero, as well as $\sum\nolimits_{\xi}\wp_{\xi}=1$,
which imply that $0<\wp_{1}<1$. However, the result of the left-hand side is
$2$. Consequently, we obtain the paradox ``$2_{Q}=\left(
1+\wp_{1}\right)  _{C}$''. In the case of W state, we get the
contradiction ``$2_{Q}=\left(  1+\delta\right)  _{C}%
$'', where $0<\delta<1$.

In Cases. 1 and 2, Bob's two sets of results are completely different, with
the contradiction ``$2_{Q}=1_{C}$'', where
$\delta=0$, consistent with the results in \cite{EPR paradox2024}.
Notwithstanding, in Cases. 3 and 4, Bob's two sets of results are not exactly
the same, with the contradiction ``$2_{Q}=\left(
1+\delta\right)  _{C}$'', where $0<\delta<1$. Generalized
proof of each case and detailed analysis of the examples are provided in the
Supplementary Information (SI) (See Supplementary Information for detail
analyses of the theorem.).

\vspace{12pt}

\section{Conclusion and Discussion}

This work advances the study of the EPR steering paradox. In Theorem 1, we
present the generalized EPR steering paradox ``$2_{Q}=\left(
1+\delta\right)  _{C}$'', where $0\leq\delta<1$. The argument
holds for any $N$-qubit entangled state, either pure or mixed. We show that
some quantum states do not satisfy the previous EPR steering paradox
``$2_{Q}=1_{C}$'' \cite{chenjl2016}, but the quantum results
contradict the LHS model as well, expressed as ``$2_{Q}=\left(
1+\delta\right)  _{C}$''. We find a universal EPR steering
paradox and show that our result contains the previous EPR steering paradox
``$2_{Q}=1_{C}$''\ as special cases
\cite{chenjl2016,Liu2021}. It is worth noting that we have found all the steerable
states when Bob's conditional states are pure.

Moreover, if one considers the $k$-setting EPR steering scenario, then a
contradiction ``$k_{Q}=\left(  1+\delta_{k}\right) _{C}
$'', where $0\leq\delta_{k}<k-1$, can be derived from a similar
procedure. Of course, our conclusion applies to any $N$-qudit state naturally.
Besides, as there exists a close connection between steering paradox and
inequality \cite{AVN1}, our conclusion also provides a new way to structure
the $N$-qubit EPR steering inequality from the perspective of the generalized
steering paradox. However, for any quantum state, if the result
``$2_{Q}=2_{C}$''\ happens, i.e., there is no
contradiction between the quantum result and the LHS model, in such a case one cannot judge whether the quantum state is steerable based on the paradox approach. Finally, if Bob's
conditional states are all mixed, is it possible to find a sufficient and
necessary condition for steering? We shall investigate this problem subsequently.

\vspace{12pt}

\begin{acknowledgments}
		% \vspace{12pt}
		% \noindent\add{\textbf{ACKNOWLEDGEMENTS}}
		
		Z.J.L. thanks Wei-Min Shang and Hao-Nan Qiang for insightful discussion.
		J.L.C. is supported by the National Natural Science Foundation of China (Grants No. 12275136 and 12075001) and the 111 Project of B23045.
		Z.J.L. is supported by the Nankai Zhide Foundation.
	
\end{acknowledgments}

\vspace{12pt}
Z.J.L., X.Y.F. and J.Z. contributed equally to this work.

		\vspace{12pt}
\noindent\textbf{COMPETING INTERESTS}
		
		The authors declare no competing interests.
\appendix
\onecolumngrid

\section{Proof of Theorem 1}

We consider that Alice and Bob share an $N$-qubit entangled state
\begin{equation}
\rho_{AB}=\sum_{\alpha}p_{\alpha}\left\vert \psi_{AB}^{\left(  \alpha\right)
}\right\rangle \left\langle \psi_{AB}^{\left(  \alpha\right)  }\right\vert ,
\label{3.1}%
\end{equation}
with
\begin{equation}
\left\vert \psi_{AB}^{\left(  \alpha\right)  }\right\rangle =\sum_{i}\left(
s_{i}^{\left(  \alpha\right)  }\left\vert \phi_{i}\right\rangle \left\vert
\eta_{i}^{\left(  \alpha\right)  }\right\rangle \right)  , \label{3.2}%
\end{equation}
or
\begin{equation}
\left\vert \psi_{AB}^{\left(  \alpha\right)  }\right\rangle =\sum_{j}\left(
t_{j}^{\left(  \alpha\right)  }\left\vert \varphi_{j}\right\rangle \left\vert
\varepsilon_{j}^{\left(  \alpha\right)  }\right\rangle \right)  . \label{3.3}%
\end{equation}
Eqs. (\ref{3.2}) and (\ref{3.3}) are the expressions of $|\psi_{AB}^{(\alpha
)}\rangle$ under different representations $\left\vert \phi_{i}\right\rangle
\otimes\left\vert \eta_{i}\right\rangle $ and $\left\vert \varphi
_{j}\right\rangle \otimes\left\vert \varepsilon_{j}\right\rangle $, where
$i,j=1,2,\cdots,2^{M}$. Alice keeps $M\left(  M<N\right)  $ particles, and Bob
keeps $(N-M)$ particles.

In the two-setting steering protocol $\left\{  \hat{n}_{1},\hat{n}%
_{2}\right\}  $, Alice performs $2^{M+1}$ projective measurements, and then
Bob obtains the according unnormalized conditional states $\tilde{\rho}%
_{a}^{\hat{n}_{\ell}}=\mathrm{tr}_{A}\left[  \left(  {{P}_{a}^{\hat{n}_{\ell}%
}\otimes\mathds{1}}\right)  \rho_{AB}\right]  $, where $\hat{n}_{\ell}$
($\ell=1,2$) is the measurement direction, $a$ is the result of Alice, and
${\mathds{1}}$ is a $2^{N-M}\times2^{N-M}$ identity matrix. Set ${{P}_{a_{i}%
}^{\hat{n}_{1}}=}\left\vert \phi_{i}\right\rangle \left\langle \phi
_{i}\right\vert $ and ${{P}_{a_{j}^{\prime}}^{\hat{n}_{2}}=}\left\vert
\varphi_{j}\right\rangle \left\langle \varphi_{j}\right\vert $. After Alice's
measurement, Bob obtains
\begin{equation}%
\begin{array}
[c]{c}%
\tilde{\rho}_{a_{i}}^{\hat{n}_{1}}=\sum_{\alpha}p_{\alpha}|s_{i}^{(\alpha
)}|^{2}|\eta_{i}^{(\alpha)}\rangle\langle\eta_{i}^{(\alpha)}|,
\end{array}
\label{s1}%
\end{equation}
or
\begin{equation}
\tilde{\rho}_{a_{j}^{\prime}}^{\hat{n}_{2}}=\sum_{\alpha}p_{\alpha}
|t_{j}^{(\alpha)}|^{2}|\varepsilon_{j}^{(\alpha)}\rangle\langle\varepsilon
_{j}^{(\alpha)}|.
\label{s1-a}%
\end{equation}

If Bob's states have a LHS description, they satisfy Eqs. (\ref{eq2.1}) and
(\ref{eq2.2}), which lead to $2^{M+1}$ equations:
\begin{equation}%
\begin{array}
[c]{c}%
\tilde{\rho}_{a}^{\hat{n}_{1}}=%
%TCIMACRO{\dsum \limits_{\xi}}%
%BeginExpansion
{\displaystyle\sum\limits_{\xi}}
%EndExpansion
\wp\left(  a|\hat{n}_{1},\xi\right)  \wp_{\xi}\rho_{\xi},\\
\tilde{\rho}_{a^{\prime}}^{\hat{n}_{2}}=%
%TCIMACRO{\dsum \limits_{\xi}}%
%BeginExpansion
{\displaystyle\sum\limits_{\xi}}
%EndExpansion
\wp\left(  a^{\prime}|\hat{n}_{2},\xi\right)  \wp_{\xi}\rho_{\xi}.
\end{array}
\label{s2}%
\end{equation}
The quantum value $2$ can be attained through summing Eq. (\ref{s1}) and Eq. (\ref{s1-a}) and then
taking the trace. However the maximum value of the LHS result is $2$ if and
only if
\begin{align*}
\sum_{a}\sum_{\xi}\wp\left(  a|\hat{n}_{1},\xi\right)  \wp_{\xi}\rho_{\xi}  &
=\sum_{a^{\prime}}\sum_{\xi}\wp\left(  a^{\prime}|\hat{n}_{2},\xi\right)
\wp_{\xi}\rho_{\xi}\\
&  =\sum_{\xi}\wp_{\xi}\rho_{\xi}.
\end{align*}
But the equation ``$2_{Q}=2_{C}$''\ is not
contradictory, and under this circumstance we cannot determine whether Alice
can steer Bob, since there may exist other measurement strategies for Alice to
steer Bob, or produce the contradiction ``$2_{Q}=\left(
1+\delta\right)  _{C}$'', with $\delta\in\left[  0,1\right)  $.

Considering that the conditional states of Bob are all pure, which implies
that $\left\vert \eta_{i}\right\rangle $ and $\left\vert \varepsilon
_{j}\right\rangle $ are independent of $\alpha$, i.e.,
\begin{equation}
\left\vert \psi_{AB}^{\left(  \alpha\right)  }\right\rangle =\sum_{i}%
s_{i}^{\left(  \alpha\right)  }\left\vert \phi_{i}\right\rangle \left\vert
\eta_{i}\right\rangle , \label{3.4}%
\end{equation}
or
\begin{equation}
\left\vert \psi_{AB}^{\left(  \alpha\right)  }\right\rangle =\sum_{j}%
t_{j}^{\left(  \alpha\right)  }\left\vert \varphi_{j}\right\rangle \left\vert
\varepsilon_{j}\right\rangle . \label{3.5}%
\end{equation}
Then the conditional states of Bob are
\begin{equation}
\left\{
\begin{array}
[c]{l}%
\tilde{\rho}_{a_{1}}^{\hat{n}_{1}}=\sum_{\alpha}p_{\alpha}\left\vert
s_{1}^{\left(  \alpha\right)  }\right\vert ^{2}\left\vert \eta_{1}%
\right\rangle \left\langle \eta_{1}\right\vert ,\\
\cdots\\
\tilde{\rho}_{a_{2^{M}}}^{\hat{n}_{1}}=\sum_{\alpha}p_{\alpha}\left\vert
s_{2^{M}}^{\left(  \alpha\right)  }\right\vert ^{2}\left\vert \eta_{2^{M}%
}\right\rangle \left\langle \eta_{2^{M}}\right\vert ,\\
\tilde{\rho}_{a_{1}^{\prime}}^{\hat{n}_{2}}=\sum_{\alpha}p_{\alpha}\left\vert
t_{1}^{\left(  \alpha\right)  }\right\vert ^{2}\left\vert \varepsilon
_{1}\right\rangle \left\langle \varepsilon_{1}\right\vert ,\\
\cdots\\
\tilde{\rho}_{a_{2^{M}}^{\prime}}^{\hat{n}_{2}}=\sum_{\alpha}p_{\alpha
}\left\vert t_{2^{M}}^{\left(  \alpha\right)  }\right\vert ^{2}\left\vert
\varepsilon_{2^{M}}\right\rangle \left\langle \varepsilon_{2^{M}}\right\vert .
\end{array}
\right.  \label{3.6}%
\end{equation}

\emph{Sufficiency}---``$2_{Q}=2_{C}$''%
$\Longrightarrow\{\tilde{\rho}_{a}^{\hat{n}_{1}}\}=\{\tilde{\rho}_{a^{\prime}%
}^{\hat{n}_{2}}\}$. To prove the proposition $\{\tilde{\rho}_{a}^{\hat{n}_{1}%
}\}\neq\{\tilde{\rho}_{a^{\prime}}^{\hat{n}_{2}}\}\Longrightarrow
$``$2_{Q}=\left(  1+\delta\right)  _{C}$'', we
simplify the process by demonstrating the accuracy of the contrapositive.

Since the density matrix of a pure state can only be expanded by itself, that
is each $\tilde{\rho}_{a}^{\hat{n}_{1}}$ and $\tilde{\rho}_{a^{\prime}}%
^{\hat{n}_{2}}$ in Eq. (\ref{s2}) can only be described by a definite hidden
state. Without loss of generality, suppose that Bob's conditional states are
completely different in the set involving the same measurement direction
$\{\tilde{\rho}_{a}^{\hat{n}_{1}}\}$ or $\{\tilde{\rho}_{a^{\prime}}^{\hat
{n}_{2}}\}$. Further Eq. (\ref{s2}) can be written as
\begin{equation}
\left\{
\begin{array}
[c]{l}%
\tilde{\rho}_{a_{1}}^{\hat{n}_{1}}=\wp_{1}\rho_{1},\\
\tilde{\rho}_{a_{2}}^{\hat{n}_{1}}=\wp_{2}\rho_{2},\\
\cdots\\
\tilde{\rho}_{a_{2^{M}}}^{\hat{n}_{1}}=\wp_{2^{M}}\rho_{2^{M}},
\end{array}
\right.  \left\{
\begin{array}
[c]{l}%
\tilde{\rho}_{a_{1}^{\prime}}^{\hat{n}_{2}}=\wp_{1}^{\prime}\rho_{1}^{\prime
},\\
\tilde{\rho}_{a_{2}^{\prime}}^{\hat{n}_{2}}=\wp_{2}^{\prime}\rho_{2}^{\prime
},\\
\cdots\\
\tilde{\rho}_{a_{2^{M}}^{\prime}}^{\hat{n}_{2}}=\wp_{2^{M}}^{\prime}%
\rho_{2^{M}}^{\prime}.
\end{array}
\right.  \label{s3}%
\end{equation}
Summing the Eq. (\ref{s3}), we have
\[
\tilde{\rho}_{a}^{\hat{n}_{1}}+\tilde{\rho}_{a^{\prime}}^{\hat{n}_{2}}%
=\sum_{i=1}^{2^{M}}\wp_{i}\rho_{i}+\sum_{i=1}^{2^{M}}\wp_{i}^{\prime}\rho
_{i}^{\prime}.
\]
The equation ``$2_{Q}=2_{C}$''\ requests
\begin{equation}
\sum_{i=1}^{2^{M}}\wp_{i}^{\prime}\rho_{i}^{\prime}=\sum_{i=1}^{2^{M}}\wp
_{i}\rho_{i}=\sum_{\xi}\wp_{\xi}\rho_{\xi}.
\end{equation}
It means that $\{\tilde{\rho}_{a}^{\hat{n}_{1}}\}=\{\tilde{\rho}_{a^{\prime}%
}^{\hat{n}_{2}}\}$, i.e., the measurement requirement is not satisfied. Note
that there is also no contradiction appearing if Bob has the same conditional
state in the same direction measurement result $\{\tilde{\rho}_{a}^{\hat
{n}_{1}}\}$ or $\{\tilde{\rho}_{a^{\prime}}^{\hat{n}_{2}}\}$. As a result, if
the measurement requirement is met, the paradoxical equality ``%
$2_{Q}=\left(  1+\delta\right)  _{C}$'',\ with $\delta
\in\left[  0,1\right)  $, can be achieved.

\emph{Necessity}---$\{\tilde{\rho}_{a}^{\hat{n}_{1}}\}=\{\tilde{\rho
}_{a^{\prime}}^{\hat{n}_{2}}\}\Longrightarrow$``$2_{Q}=2_{C}%
$''. Similarly, the proposition $\text{``}%
2_{Q}=\left(  1+\delta\right)  _{C}\text{''}\Longrightarrow
\{\tilde{\rho}_{a}^{\hat{n}_{1}}\}\neq\{\tilde{\rho}_{a^{\prime}}^{\hat{n}%
_{2}}\}$ can be proven by showing the accuracy of its contrapositive.

Suppose that there are $2^{M}$ distinct states in $\{\tilde{\rho}_{a}^{\hat
{n}_{1}}\}$ and that they are the same as the corresponding $2^{M}$ elements
in $\{\tilde{\rho}_{a}^{\hat{n}_{2}}\}$. Bob's conditional states are
\begin{equation}
\left\{
\begin{array}
[c]{l}%
\tilde{\rho}_{a_{1}}^{\hat{n}_{1}}=\sum_{\alpha}p_{\alpha}\left\vert
s_{1}^{\left(  \alpha\right)  }\right\vert ^{2}\left\vert \eta_{1}%
\right\rangle \left\langle \eta_{1}\right\vert ,\\
\cdots\\
\tilde{\rho}_{a_{2^{M}}}^{\hat{n}_{1}}=\sum_{\alpha}p_{\alpha}\left\vert
s_{2^{M}}^{\left(  \alpha\right)  }\right\vert ^{2}\left\vert \eta_{2^{M}%
}\right\rangle \left\langle \eta_{2^{M}}\right\vert ,\\
\tilde{\rho}_{a_{1}^{\prime}}^{\hat{n}_{2}}=\sum_{\alpha}p_{\alpha}\left\vert
t_{1}^{\left(  \alpha\right)  }\right\vert ^{2}\left\vert \eta_{1}%
\right\rangle \left\langle \eta_{1}\right\vert ,\\
\cdots\\
\tilde{\rho}_{a_{2^{M}}^{\prime}}^{\hat{n}_{2}}=\sum_{\alpha}p_{\alpha
}\left\vert t_{2^{M}}^{\left(  \alpha\right)  }\right\vert ^{2}\left\vert
\eta_{2^{M}}\right\rangle \left\langle \eta_{2^{M}}\right\vert .
\end{array}
\right.  \label{c5.1}%
\end{equation}
There are only $2^{M}$ different\ pure states in the quantum result Eq.
(\ref{c5.1}). It is sufficient to take $\xi$ from $1$ to $2^{M}$ in the LHS
description, e.g.,
\begin{equation}
\left\{
\begin{array}
[c]{l}%
\tilde{\rho}_{a}^{\hat{n}_{1}}=\sum\limits_{\xi=1}^{2^{M}}\wp\left(  a|\hat
{n}_{1},\xi\right)  \wp_{\xi}\rho_{\xi},\\
\tilde{\rho}_{a^{\prime}}^{\hat{n}_{2}}=\sum\limits_{\xi=1}^{2^{M}}\wp\left(
a^{\prime}|\hat{n}_{2},\xi\right)  \wp_{\xi}\rho_{\xi}.
\end{array}
\right.  \label{c5.2}%
\end{equation}
Since Bob's states are all pure, each $\tilde{\rho}_{a}^{\hat{n}_{\ell}}$ in
Eq. (\ref{c5.2}) contains only one term. Eq. (\ref{c5.2}) can be written as
\begin{equation}
\left\{
\begin{array}
[c]{l}%
\tilde{\rho}_{a_{1}}^{\hat{n}_{1}}=\wp_{1}\rho_{1},\\
\cdots\\
\tilde{\rho}_{a_{2^{M}}}^{\hat{n}_{1}}=\wp_{2^{M}}\rho_{2^{M}},\\
\tilde{\rho}_{a_{1}^{\prime}}^{\hat{n}_{2}}=\wp_{1}\rho_{1},\\
\cdots\\
\tilde{\rho}_{a_{2^{M}}^{\prime}}^{\hat{n}_{2}}=\wp_{2^{M}}\rho_{2^{M}}.
\end{array}
\right.  \label{c5.3}%
\end{equation}
Here we describe the same term with the same hidden state. And ${\sum
\limits_{a}}\wp\left(  a|\hat{n},\xi\right)  =1$, so we have $\wp\left(
a_{1}|\hat{n}_{1},1\right)  =\cdots=\wp\left(  a_{2^{M}}|\hat{n}_{1}%
,2^{M}\right)  =\wp\left(  a_{1}^{\prime}|\hat{n}_{2},1\right)  \cdots
=\wp\left(  a_{2^{M}}^{\prime}|\hat{n}_{2},2^{M}\right)  =1$, and the others
$\wp\left(  a|\hat{n},\xi\right)  =0$. Then we take the sum of Eq.
(\ref{c5.3}) and take the trace, and get ``$2_{Q}=2_{C}%
$''. That means that there is no contradiction if
$\{\tilde{\rho}_{a}^{\hat{n}_{1}}\}=\{\tilde{\rho}_{a^{\prime}}^{\hat{n}_{2}%
}\}$, i.e., the measurement requirement is not satisfied, and we cannot
conclude whether Alice can steer Bob. Consequently, if there is a
contradiction ``$2_{Q}=\left(  1+\delta\right)  _{C}%
$'', then the measurement condition must be satisfied.

In summary, we demonstrate that the measurement requirement is both a
sufficient and necessary condition for the contradiction ``%
$2_{Q}=\left(  1+\delta\right)  _{C}$'',\ on the premise that
Bob's conditional states are pure.

	%--------------------

	% \vspace{12pt}
	% \noindent\add{\textbf{REFERENCES}}

   \twocolumngrid

	% 	\noindent\add{\textbf{AUTHOR CONTRIBUTIONS}}
		
	% 	J.L.C. supervised the project. Z.J.L., J.Z., X.Y.F. and M.X. contributed to the theoretical analysis.
	% 	All authors read the paper and discussed the results. Z.J.L., X.Y.F. and J.Z. contributed equally to this work.
		
	% \vspace{12pt}
	% \noindent\add{\textbf{DATA AVAILABILITY}}
	
	% %The data that support the findings of this study are available from the authors upon request.
	
    % All data that support the findings of this study are included within the article (and any supplementary files).

\end{document}

\renewcommand{\figurename}{Fig.}

\title{Supplemental Information for \textquotedblleft Generalized
Einstein-Podolsky-Rosen Steering Paradox\textquotedblright}

	\author{Zhi-Jie Liu}
	%\email{1120220019@mail.nankai.edu.cn}
	\affiliation{Theoretical Physics Division, Chern Institute of Mathematics and LPMC, Nankai University, Tianjin 300071, People's Republic of China}

	\author{Xing-Yan Fan}
	\affiliation{Theoretical Physics Division, Chern Institute of Mathematics and LPMC, Nankai University, Tianjin 300071, People's Republic of China}

	\author{Jie Zhou}
	\affiliation{College of Physics and Materials Science, Tianjin Normal University, Tianjin 300382, People's Republic of China}

	\author{Mi Xie}
	%\email{xiemi@tju.edu.cn}
	\affiliation{Department of Physics, School of Science, Tianjin University, Tianjin 300072,  People's Republic of China}

	\author{Jing-Ling~Chen}
	\email{chenjl@nankai.edu.cn}
	\affiliation{Theoretical Physics Division, Chern Institute of Mathematics and LPMC, Nankai University, Tianjin 300071, People's Republic of China}

	\date{\today}

\maketitle
\tableofcontents

\section{Detail analyses of the theorem}

We consider Alice and Bob share an $N$-qubit entangled state
\begin{equation}
\rho_{AB}=\sum_{\alpha}p_{\alpha}\left\vert \psi_{AB}^{\left(  \alpha\right)
}\right\rangle \left\langle \psi_{AB}^{\left(  \alpha\right)  }\right\vert .
\end{equation}
$\left\vert \psi_{AB}^{\left(  \alpha\right)  }\right\rangle $ can be written
as
\begin{equation}
\left\vert \psi_{AB}^{\left(  \alpha\right)  }\right\rangle =\sum_{i}\left(
s_{i}^{\left(  \alpha\right)  }\left\vert \phi_{i}\right\rangle \left\vert
\eta_{i}^{\left(  \alpha\right)  }\right\rangle \right)  , \label{3.2}%
\end{equation}
or%
\begin{equation}
\left\vert \psi_{AB}^{\left(  \alpha\right)  }\right\rangle =\sum_{j}\left(
t_{j}^{\left(  \alpha\right)  }\left\vert \varphi_{j}\right\rangle \left\vert
\varepsilon_{j}^{\left(  \alpha\right)  }\right\rangle \right)  , \label{3.3}%
\end{equation}
Eqs. (\ref{3.2}) and (\ref{3.3}) are representations of $\left\vert \psi
_{AB}^{\left(  \alpha\right)  }\right\rangle $ under different representations
of $\left\vert \phi_{i}\right\rangle \otimes\left\vert \eta_{i}\right\rangle $
and $\left\vert \varphi_{j}\right\rangle \otimes\left\vert \varepsilon
_{j}\right\rangle $, where $i,j=1,2,\cdots,2^{M}$. Alice keeps $M\left(
M<N\right)  $ particles, and Bob keeps $(N-M)$.

In the 2-setting steering protocol $\left\{  \hat{n}_{1},\hat{n}_{2}\right\}
$, Alice performs $2^{M+1}$ projective measurements, then Bob obtains
unnormalized conditional states $\tilde{\rho}_{a}^{\hat{n}_{\ell}}%
=\mathrm{tr}_{A}\left[  \left(  {{P}_{a}^{\hat{n}_{\ell}}\otimes
\mathds{1}}\right)  \rho_{AB}\right]  $, where $\hat{n}_{\ell}$ ($\ell=1,2$)
is the measurement direction, $a$ is the rusult of Alice, and ${\mathds{1}}$
is a $2^{N-M}\times2^{N-M}$ identity matrix. We consider Alice's measurements
${{P}_{a_{i}}^{\hat{n}_{1}}=}\left\vert \phi_{i}\right\rangle \left\langle
\phi_{i}\right\vert $ and ${{P}_{a_{j}^{\prime}}^{\hat{n}_{2}}=}\left\vert
\varphi_{j}\right\rangle \left\langle \varphi_{j}\right\vert $. After Alice's
measurement, Bob obtains
\[
\tilde{\rho}_{a_{i}}^{\hat{n}_{1}}=\sum_{\alpha}p_{\alpha}\left\vert
s_{i}^{\left(  \alpha\right)  }\right\vert ^{2}\left\vert \eta_{i}^{\left(
\alpha\right)  }\right\rangle \left\langle \eta_{i}^{\left(  \alpha\right)
}\right\vert ,
\]
and%
\[
\tilde{\rho}_{a_{j}^{\prime}}^{\hat{n}_{2}}=\sum_{\alpha}p_{\alpha}\left\vert
t_{j}^{(\alpha)}\right\vert ^{2}\left\vert \varepsilon_{j}^{(\alpha
)}\right\rangle \left\langle \varepsilon_{j}^{(\alpha)}\right\vert .
\]

The pure state requirement requests the conditional states of Bob are all pure
states. It's means that $\left\vert \eta_{i}\right\rangle $ and $\left\vert
\varepsilon_{j}\right\rangle $ are independent of $\alpha$. So the state
function can be reexpressed as%
\begin{equation}
\left\vert \psi_{AB}^{\left(  \alpha\right)  }\right\rangle =\sum_{i}\left(
s_{i}^{\left(  \alpha\right)  }\left\vert \phi_{i}\right\rangle \left\vert
\eta_{i}\right\rangle \right)  , \label{3.4}%
\end{equation}
or%
\begin{equation}
\left\vert \psi_{AB}^{\left(  \alpha\right)  }\right\rangle =\sum_{j}\left(
t_{j}^{\left(  \alpha\right)  }\left\vert \varphi_{j}\right\rangle \left\vert
\varepsilon_{j}\right\rangle \right)  , \label{3.5}%
\end{equation}
and the conditional states of Bob can be rewritten as%
\begin{equation}
\left\{
\begin{array}
[c]{l}%
\tilde{\rho}_{a_{1}}^{\hat{n}_{1}}=\sum_{\alpha}p_{\alpha}\left\vert
s_{1}^{\left(  \alpha\right)  }\right\vert ^{2}\left\vert \eta_{1}%
\right\rangle \left\langle \eta_{1}\right\vert ,\\
\cdots\\
\tilde{\rho}_{a_{2^{M}}}^{\hat{n}_{1}}=\sum_{\alpha}p_{\alpha}\left\vert
s_{2^{M}}^{\left(  \alpha\right)  }\right\vert ^{2}\left\vert \eta_{2^{M}%
}\right\rangle \left\langle \eta_{2^{M}}\right\vert ,\\
\tilde{\rho}_{a_{1}^{\prime}}^{\hat{n}_{2}}=\sum_{\alpha}p_{\alpha}\left\vert
t_{1}^{\left(  \alpha\right)  }\right\vert ^{2}\left\vert \varepsilon
_{1}\right\rangle \left\langle \varepsilon_{1}\right\vert ,\\
\cdots\\
\tilde{\rho}_{a_{2^{M}}^{\prime}}^{\hat{n}_{2}}=\sum_{\alpha}p_{\alpha
}\left\vert t_{2^{M}}^{\left(  \alpha\right)  }\right\vert ^{2}\left\vert
\varepsilon_{2^{M}}\right\rangle \left\langle \varepsilon_{2^{M}}\right\vert .
\end{array}
\right.  \label{3.6}%
\end{equation}

If Bob's states have a LHS discription, they satisfy
\begin{equation}
\tilde{\rho}_{a}^{\hat{n}}=%
%TCIMACRO{\dsum \limits_{\xi}}%
%BeginExpansion
{\displaystyle\sum\limits_{\xi}}
%EndExpansion
\wp\left(  a|\hat{n},\xi\right)  \wp_{\xi}\rho_{\xi}, \label{eq2.1}%
\end{equation}
and
\begin{equation}%
%TCIMACRO{\dsum \limits_{\xi}}%
%BeginExpansion
{\displaystyle\sum\limits_{\xi}}
%EndExpansion
\wp_{\xi}\rho_{\xi}=\rho_{B}. \label{eq2.2}%
\end{equation}
Here $\wp_{\xi}$ and $\wp\left(  a|\hat{n},\xi\right)  $ are probabilities
satisfying%
\begin{equation}
{\sum\limits_{\xi}}\wp_{\xi}=1,
\end{equation}
and$\;$%
\begin{equation}
{\sum\limits_{a}}\wp\left(  a|\hat{n},\xi\right)  =1,
\end{equation}
for a fixed $\xi$, and $\rho_{B}=\mathrm{tr}_{A}\left(  \rho_{AB}\right)  $ is
Bob's reduced density matrix. Then there are $2^{M+1}$ equations:
\begin{equation}
\tilde{\rho}_{a}^{\hat{n}_{\ell}}=%
%TCIMACRO{\dsum \limits_{\xi}}%
%BeginExpansion
{\displaystyle\sum\limits_{\xi}}
%EndExpansion
\wp\left(  a|\hat{n}_{\ell},\xi\right)  \wp_{\xi}\rho_{\xi}. \label{3.7}%
\end{equation}
Since Bob's conditional states are all pure states, and the density matrix of
a pure state cannot be obtained by superposition of other states, each
$\tilde{\rho}_{a}^{\hat{n}_{\ell}}$ can only be described by one hidden state
$\wp_{\xi}\rho_{\xi}$.

\subsection{ The sets of Bob's conditional states are all different}

Firstly, we consider Bob's $2^{M+1}$ conditional states are all different in
the quantum result. So it is enough to take $\xi$ from 1 to $2^{M+1}$ in the
LHS discription. Then we have%
\begin{equation}
\left\{
\begin{array}
[c]{l}%
\tilde{\rho}_{a_{1}}^{\hat{n}_{1}}=%
%TCIMACRO{\dsum \limits_{\xi=1}^{2^{M+1}}}%
%BeginExpansion
{\displaystyle\sum\limits_{\xi=1}^{2^{M+1}}}
%EndExpansion
\wp\left(  a_{1}|\hat{n}_{1},\xi\right)  \wp_{\xi}\rho_{\xi},\\
\cdots\\
\tilde{\rho}_{a_{2^{M}}}^{\hat{n}_{1}}=%
%TCIMACRO{\dsum \limits_{\xi=1}^{2^{M+1}}}%
%BeginExpansion
{\displaystyle\sum\limits_{\xi=1}^{2^{M+1}}}
%EndExpansion
\wp\left(  a_{2^{M}}|\hat{n}_{1},\xi\right)  \wp_{\xi}\rho_{\xi},\\
\tilde{\rho}_{a_{1}^{\prime}}^{\hat{n}_{2}}=%
%TCIMACRO{\dsum \limits_{\xi=1}^{2^{M+1}}}%
%BeginExpansion
{\displaystyle\sum\limits_{\xi=1}^{2^{M+1}}}
%EndExpansion
\wp\left(  a_{1}^{\prime}|\hat{n}_{2},\xi\right)  \wp_{\xi}\rho_{\xi},\\
\cdots\\
\tilde{\rho}_{a_{2^{M}}^{\prime}}^{\hat{n}_{2}}=%
%TCIMACRO{\dsum \limits_{\xi=1}^{2^{M+1}}}%
%BeginExpansion
{\displaystyle\sum\limits_{\xi=1}^{2^{M+1}}}
%EndExpansion
\wp\left(  a_{2^{M}}^{\prime}|\hat{n}_{2},\xi\right)  \wp_{\xi}\rho_{\xi}.
\end{array}
\right.  \label{c1.1}%
\end{equation}
Each $\tilde{\rho}_{a}^{\hat{n}_{\ell}}$ in Eq. (\ref{c1.1}) contains only one
term since Bob's conditional states are all pure, and the density matrix of a
pure state can only be expanded by itself. And they're all different from each
other. Eq. (\ref{c1.1}) can be written as%
\begin{equation}
\left\{
\begin{array}
[c]{l}%
\tilde{\rho}_{a_{1}}^{\hat{n}_{1}}=\wp_{1}\rho_{1},\\
\cdots\\
\tilde{\rho}_{a_{2^{M}}}^{\hat{n}_{1}}=\wp_{2^{M}}\rho_{2^{M}},\\
\tilde{\rho}_{a_{1}^{\prime}}^{\hat{n}_{2}}=\wp_{2^{M}+1}\rho_{2^{M}+1},\\
\cdots\\
\tilde{\rho}_{a_{2^{M}}^{\prime}}^{\hat{n}_{2}}=\wp_{2^{M+1}}\rho_{2^{M+1}}.
\end{array}
\right.  \label{c1.2}%
\end{equation}
Because ${\sum\limits_{a}}\wp\left(  a|\hat{n},\xi\right)  =1$, $\wp\left(
a_{1}|\hat{n}_{1},1\right)  =\wp\left(  a_{2}|\hat{n}_{1},2\right)
=\cdots=\wp\left(  a_{2^{M}}|\hat{n}_{1},2^{M}\right)  =\wp\left(
a_{1}^{\prime}|\hat{n}_{2},2^{M}+1\right)  \cdots=\wp\left(  a_{2^{M}}%
^{\prime}|\hat{n}_{2},2^{M+1}\right)  =1$ and all other $\wp\left(  a|\hat
{n},\xi\right)  =0$. Summing the Eq. (\ref{c1.2}) and taking the trace, we get
the paradox \textquotedblleft$2_{Q}=1_{C}$\textquotedblright.
\textquotedblleft$2_{Q}$\textquotedblright\ is the quantum result, and
\textquotedblleft$1_{C}$\textquotedblright\ is the classical result. This also
means that in this case we can obtain the EPR steering paradox
\textquotedblleft$2_{Q}=1_{C}$\textquotedblright, and at this point $\delta=0$.

\subsubsection{Example 1}

We consider the example of an arbitrary 2-qubit pure state presented by Chen
\emph{et al}. (2016) \cite{chenjl2016}. For arbitrary 2-qubit pure state
$\rho_{AB}=\left\vert \Psi\left(  \theta\right)  \right\rangle \left\langle
\Psi\left(  \theta\right)  \right\vert $, where
\begin{equation}
\left\vert \Psi\left(  \theta\right)  \right\rangle =\cos\theta\left\vert
0\right\rangle \left\vert 0\right\rangle +\sin\theta\left\vert 1\right\rangle
\left\vert 1\right\rangle ,\label{eq2.01}%
\end{equation}
or%
\begin{align}
\left\vert \Psi\left(  \theta\right)  \right\rangle  &  =\frac{\cos\theta
}{\sqrt{2}}\left(  \left\vert +\right\rangle +\left\vert -\right\rangle
\right)  \left\vert 0\right\rangle +\frac{\sin\theta}{\sqrt{2}}\left(
\left\vert +\right\rangle -\left\vert -\right\rangle \right)  \left\vert
1\right\rangle \label{eq2.012}\\
&  =\frac{\left\vert +\right\rangle }{\sqrt{2}}\left(  \cos\theta\left\vert
0\right\rangle +\sin\theta\left\vert 1\right\rangle \right)  +\frac{\left\vert
-\right\rangle }{\sqrt{2}}\left(  \cos\theta\left\vert 0\right\rangle
-\sin\theta\left\vert 1\right\rangle \right)  ,\nonumber
\end{align}
with $\theta\in\left(  0,\pi/2\right)  $. Eqs. (\ref{eq2.01}) and
(\ref{eq2.012}) are representations of $\left\vert \Psi\left(  \theta\right)
\right\rangle $ under different representations. Alice prepares this 2-qubit
pure state, one of which she keeps, and the other is sent to Bob. In the
two-setting steering protocol $\left\{  \hat{z},\hat{x}\right\}  $, Alice has
four possible projective measurements%
\begin{equation}
\left\{
\begin{array}
[c]{l}%
{{P}_{0}^{\hat{z}}=}\left\vert 0\right\rangle \left\langle 0\right\vert ,\\
{{P}_{1}^{\hat{z}}=}\left\vert 1\right\rangle \left\langle 1\right\vert ,\\
{{P}_{0}^{\hat{x}}=}\left\vert +\right\rangle \left\langle +\right\vert ,\\
{{P}_{1}^{\hat{x}}=}\left\vert -\right\rangle \left\langle -\right\vert ,
\end{array}
\right.  \label{eq2.02}%
\end{equation}
where $\left\vert \pm\right\rangle =\frac{1}{\sqrt{2}}\left(  \left\vert
0\right\rangle \pm\left\vert 1\right\rangle \right)  $. After Alice's
measurement, the four non-normalized conditional states corresponding to Bob
are
\begin{equation}
\left\{
\begin{array}
[c]{l}%
\tilde{\rho}{_{0}^{\hat{z}}=\cos}^{2}\theta\left\vert 0\right\rangle
\left\langle 0\right\vert ,\\
\tilde{\rho}{_{1}^{\hat{z}}=\sin}^{2}{\theta}\left\vert 1\right\rangle
\left\langle 1\right\vert ,\\
\tilde{\rho}{_{0}^{\hat{x}}=}\frac{1}{2}\left\vert \lambda_{+}\right\rangle
\left\langle \lambda_{\newline+}\right\vert ,\\
\tilde{\rho}{_{1}^{\hat{x}}=}\frac{1}{2}\left\vert \lambda_{-}\right\rangle
\left\langle \lambda_{\newline-}\right\vert ,
\end{array}
\right.  \label{eq2.03}%
\end{equation}
where $\left\vert \lambda_{\pm}\right\rangle =\cos\theta\left\vert
0\right\rangle \pm{\sin\theta}\left\vert 1\right\rangle $. This is the quantum
result. And it is worth noting that Bob's conditional states are all pure
states and that none of the four conditional states are identical. This
satisfies the pure state requirement and the measurement requirement mentioned
in the theorem.

If Bob's states have a LHS description, they must satisfy Eqs. (\ref{eq2.1})
and (\ref{eq2.2}). Then Bob's four unnormalized conditional states can be
written as%
\begin{equation}
\left\{
\begin{array}
[c]{c}%
\widetilde{\rho}_{0}^{\hat{z}}=%
%TCIMACRO{\dsum \limits_{\xi}}%
%BeginExpansion
{\displaystyle\sum\limits_{\xi}}
%EndExpansion
\wp\left(  0|\hat{z},\xi\right)  \wp_{\xi}\rho_{\xi},\\
\widetilde{\rho}_{1}^{\hat{z}}=%
%TCIMACRO{\dsum \limits_{\xi}}%
%BeginExpansion
{\displaystyle\sum\limits_{\xi}}
%EndExpansion
\wp\left(  1|\hat{z},\xi\right)  \wp_{\xi}\rho_{\xi},\\
\widetilde{\rho}_{0}^{\hat{x}}=%
%TCIMACRO{\dsum \limits_{\xi}}%
%BeginExpansion
{\displaystyle\sum\limits_{\xi}}
%EndExpansion
\wp\left(  0|\hat{x},\xi\right)  \wp_{\xi}\rho_{\xi},\\
\widetilde{\rho}_{1}^{\hat{x}}=%
%TCIMACRO{\dsum \limits_{\xi}}%
%BeginExpansion
{\displaystyle\sum\limits_{\xi}}
%EndExpansion
\wp\left(  1|\hat{x},\xi\right)  \wp_{\xi}\rho_{\xi}.
\end{array}
\right.  \label{eq2.003}%
\end{equation}
Because the four states of Eq. (\ref{eq2.03}) are pure states, it is
sufficient to take $\xi$ from $1$ to $4$. Eq. (\ref{eq2.003}) can be write as%

\begin{equation}
\left\{
\begin{array}
[c]{c}%
\widetilde{\rho}_{0}^{\hat{z}}=\wp\left(  0|\hat{z},1\right)  \wp_{1}\rho
_{1}+\wp\left(  0|\hat{z},2\right)  \wp_{2}\rho_{2}+\wp\left(  0|\hat
{z},3\right)  \wp_{3}\rho_{3}+\wp\left(  0|\hat{z},4\right)  \wp_{4}\rho
_{4},\\
\widetilde{\rho}_{1}^{\hat{z}}=\wp\left(  1|\hat{z},1\right)  \wp_{1}\rho
_{1}+\wp\left(  1|\hat{z},2\right)  \wp_{2}\rho_{2}+\wp\left(  1|\hat
{z},3\right)  \wp_{3}\rho_{3}+\wp\left(  1|\hat{z},4\right)  \wp_{4}\rho
_{4},\\
\widetilde{\rho}_{0}^{\hat{x}}=\wp\left(  0|\hat{x},1\right)  \wp_{1}\rho
_{1}+\wp\left(  0|\hat{x},2\right)  \wp_{2}\rho_{2}+\wp\left(  0|\hat
{x},3\right)  \wp_{3}\rho_{3}+\wp\left(  0|\hat{x},4\right)  \wp_{4}\rho
_{4},\\
\widetilde{\rho}_{1}^{\hat{x}}=\wp\left(  1|\hat{x},1\right)  \wp_{1}\rho
_{1}+\wp\left(  1|\hat{x},2\right)  \wp_{2}\rho_{2}+\wp\left(  1|\hat
{x},3\right)  \wp_{3}\rho_{3}+\wp\left(  1|\hat{x},4\right)  \wp_{4}\rho_{4}.
\end{array}
\right.  \label{eq2.005}%
\end{equation}
And owing to the fact that a pure state cannot be obtained by convex
combination of other pure states, one has
\begin{equation}
\left\{
\begin{array}
[c]{l}%
\tilde{\rho}{_{0}^{\hat{z}}=}\wp_{1}\rho_{1},\\
\tilde{\rho}{_{1}^{\hat{z}}=}\wp_{2}\rho_{2},\\
\tilde{\rho}{_{0}^{\hat{x}}=}\wp_{3}\rho_{3},\\
\tilde{\rho}{_{1}^{\hat{x}}=}\wp_{4}\rho_{4}.
\end{array}
\right.  \label{eq2.04}%
\end{equation}
Here%
\[
\wp\left(  0|\hat{z},1\right)  =\wp\left(  1|\hat{z},2\right)  =\wp\left(
0|\hat{x},3\right)  =\wp\left(  1|\hat{x},4\right)  =1,
\]
and other $\wp\left(  a|\hat{n},\xi\right)  =0$. Summing and taking the trace
of Eq. (\ref{eq2.04}) gives $\mathrm{tr}\left(  \tilde{\rho}{^{\hat{z}}%
}+\tilde{\rho}{^{\hat{x}}}\right)  =2\mathrm{tr}\left(  \rho_{B}\right)  =2$
on the left and $\mathrm{tr}\left(  \wp_{1}\rho_{1}+\wp_{2}\rho_{2}+\wp
_{3}\rho_{3}+\wp_{4}\rho_{4}\right)  =\mathrm{tr}\left(  \rho_{B}\right)  =1$
on the right. Then the contradiction \textquotedblleft$2_{Q}=1_{C}%
$\textquotedblright\ is obtained. In this case EPR steering paradox is
formulated as \textquotedblleft$2_{Q}=1_{C}$\textquotedblright\ with
$\delta=0$.

\subsection{There are identical states in the same set of measurements}

Suppose two elements in $\left\{  \left\vert \eta_{i}\right\rangle \right\}  $
(or $\left\{  \left\vert \varepsilon_{j}\right\rangle \right\}  $) are
identical, for example, $\left\vert \eta_{m}\right\rangle =\left\vert \eta
_{n}\right\rangle $, $\left(  \text{where }m\neq n\text{, }m,n=1,2,\cdots
,2^{M}\right)  $, and the others are different. In this case, Bob's
conditional states are%
\begin{equation}
\left\{
\begin{array}
[c]{l}%
\tilde{\rho}_{a_{1}}^{\hat{n}_{1}}=\sum_{\alpha}p_{\alpha}\left\vert
s_{1}^{\left(  \alpha\right)  }\right\vert ^{2}\left\vert \eta_{1}%
\right\rangle \left\langle \eta_{1}\right\vert ,\\
\cdots\\
\tilde{\rho}_{a_{m}}^{\hat{n}_{1}}=\sum_{\alpha}p_{\alpha}\left\vert
s_{m}^{\left(  \alpha\right)  }\right\vert ^{2}\left\vert \eta_{m}%
\right\rangle \left\langle \eta_{m}\right\vert ,\\
\cdots\\
\tilde{\rho}_{a_{n}}^{\hat{n}_{1}}=\sum_{\alpha}p_{\alpha}\left\vert
s_{n}^{\left(  \alpha\right)  }\right\vert ^{2}\left\vert \eta_{m}%
\right\rangle \left\langle \eta_{m}\right\vert ,\\
\cdots\\
\tilde{\rho}_{a_{2^{M}}}^{\hat{n}_{1}}=\sum_{\alpha}p_{\alpha}\left\vert
s_{2^{M}}^{\left(  \alpha\right)  }\right\vert ^{2}\left\vert \eta_{2^{M}%
}\right\rangle \left\langle \eta_{2^{M}}\right\vert ,\\
\tilde{\rho}_{a_{1}^{\prime}}^{\hat{n}_{2}}=\sum_{\alpha}p_{\alpha}\left\vert
t_{1}^{\left(  \alpha\right)  }\right\vert ^{2}\left\vert \varepsilon
_{1}\right\rangle \left\langle \varepsilon_{1}\right\vert ,\\
\cdots\\
\tilde{\rho}_{a_{2^{M}}^{\prime}}^{\hat{n}_{2}}=\sum_{\alpha}p_{\alpha
}\left\vert t_{2^{M}}^{\left(  \alpha\right)  }\right\vert ^{2}\left\vert
\varepsilon_{2^{M}}\right\rangle \left\langle \varepsilon_{2^{M}}\right\vert .
\end{array}
\right.  \label{c2.1}%
\end{equation}
In this case, there are $2^{M+1}-1$ different\ pure states in the quantum
result Eq. (\ref{c2.1}). For the LHS description, it is sufficient to take
$\xi$ from 1 to $2^{M+1}-1$, e.g.,
\begin{equation}
\left\{
\begin{array}
[c]{l}%
\tilde{\rho}_{a_{1}}^{\hat{n}_{1}}=%
%TCIMACRO{\dsum \limits_{\xi=1}^{2^{M+1}-1}}%
%BeginExpansion
{\displaystyle\sum\limits_{\xi=1}^{2^{M+1}-1}}
%EndExpansion
\wp\left(  a_{1}|\hat{n}_{1},\xi\right)  \wp_{\xi}\rho_{\xi},\\
\cdots\\
\tilde{\rho}_{a_{m}}^{\hat{n}_{1}}=%
%TCIMACRO{\dsum \limits_{\xi=1}^{2^{M+1}-1}}%
%BeginExpansion
{\displaystyle\sum\limits_{\xi=1}^{2^{M+1}-1}}
%EndExpansion
\wp\left(  a_{m}|\hat{n}_{1},\xi\right)  \wp_{\xi}\rho_{\xi},\\
\cdots\\
\tilde{\rho}_{a_{n}}^{\hat{n}_{1}}=%
%TCIMACRO{\dsum \limits_{\xi=1}^{2^{M+1}-1}}%
%BeginExpansion
{\displaystyle\sum\limits_{\xi=1}^{2^{M+1}-1}}
%EndExpansion
\wp\left(  a_{n}|\hat{n}_{1},\xi\right)  \wp_{\xi}\rho_{\xi},\\
\cdots\\
\tilde{\rho}_{a_{2^{M}}}^{\hat{n}_{1}}=%
%TCIMACRO{\dsum \limits_{\xi=1}^{2^{M+1}-1}}%
%BeginExpansion
{\displaystyle\sum\limits_{\xi=1}^{2^{M+1}-1}}
%EndExpansion
\wp\left(  a_{2^{M}}|\hat{n}_{1},\xi\right)  \wp_{\xi}\rho_{\xi},\\
\tilde{\rho}_{a_{1}^{\prime}}^{\hat{n}_{2}}=%
%TCIMACRO{\dsum \limits_{\xi=1}^{2^{M+1}-1}}%
%BeginExpansion
{\displaystyle\sum\limits_{\xi=1}^{2^{M+1}-1}}
%EndExpansion
\wp\left(  a_{1}^{\prime}|\hat{n}_{2},\xi\right)  \wp_{\xi}\rho_{\xi},\\
\cdots\\
\tilde{\rho}_{a_{2^{M}}^{\prime}}^{\hat{n}_{2}}=%
%TCIMACRO{\dsum \limits_{\xi=1}^{2^{M+1}-1}}%
%BeginExpansion
{\displaystyle\sum\limits_{\xi=1}^{2^{M+1}-1}}
%EndExpansion
\wp\left(  a_{2^{M}}^{\prime}|\hat{n}_{2},\xi\right)  \wp_{\xi}\rho_{\xi}.
\end{array}
\right.  \label{c2.2}%
\end{equation}
Similarly, because Bob's states are all pure, each $\tilde{\rho}_{a}^{\hat
{n}_{\ell}}$ in Eq. (\ref{c2.2}) contains only one term. Eq. (\ref{c2.2}) can
be written as%
\begin{equation}
\left\{
\begin{array}
[c]{l}%
\tilde{\rho}_{a_{1}}^{\hat{n}_{1}}=\wp_{1}\rho_{1},\\
\cdots\\
\tilde{\rho}_{a_{m}}^{\hat{n}_{1}}=\wp\left(  a_{m}|\hat{n}_{1},m\right)
\wp_{m}\rho_{m},\\
\cdots\\
\tilde{\rho}_{a_{n}}^{\hat{n}_{1}}=\wp\left(  a_{n}|\hat{n}_{1},m\right)
\wp_{m}\rho_{m},\\
\cdots\\
\tilde{\rho}_{a_{2^{M}}}^{\hat{n}_{1}}=\wp_{2^{M}-1}\rho_{2^{M}-1},\\
\tilde{\rho}_{a_{1}^{\prime}}^{\hat{n}_{2}}=\wp_{2^{M}}\rho_{2^{M}},\\
\cdots\\
\tilde{\rho}_{a_{2^{M}}^{\prime}}^{\hat{n}_{2}}=\wp_{2^{M+1}-1}\rho
_{2^{M+1}-1}.
\end{array}
\right.  \label{c2.3}%
\end{equation}
Since in quantum result we assume that $\left\vert \eta_{m}\right\rangle
=\left\vert \eta_{n}\right\rangle $, in the LHS description we describe
$\tilde{\rho}_{a_{m}}^{\hat{n}_{1}}$ and $\tilde{\rho}_{a_{n}}^{\hat{n}_{1}}$
by the same hidden state $\wp_{m}\rho_{m}$. And ${\sum\limits_{a}}\wp\left(
a|\hat{n},\xi\right)  =1$, so we have $\wp\left(  a_{1}|\hat{n}_{1},1\right)
=\cdots=\wp\left(  a_{m-1}|\hat{n}_{1},m-1\right)  =\wp\left(  a_{m+1}|\hat
{n}_{1},m+1\right)  =\cdots=\wp\left(  a_{n-1}|\hat{n}_{1},n-1\right)
=\wp\left(  a_{n+1}|\hat{n}_{1},n\right)  =\cdots=\wp\left(  a_{2^{M}}|\hat
{n}_{1},2^{M}-1\right)  =\wp\left(  a_{1}^{\prime}|\hat{n}_{2},2^{M}\right)
\cdots=\wp\left(  a_{2^{M}}^{\prime}|\hat{n}_{2},2^{M+1}-1\right)  =1$,
$\wp\left(  a_{m}|\hat{n}_{1},m\right)  +\wp\left(  a_{n}|\hat{n}%
_{1},m\right)  =1$, and the others $\wp\left(  a|\hat{n},\xi\right)  =0$. Then
we take the sum of Eq. (\ref{c2.3}) and take the trace, and we get
\textquotedblleft$2_{Q}=1_{C}$\textquotedblright. In the same way, in this
case we get the EPR steering paradox \textquotedblleft$2_{Q}=1_{C}%
$\textquotedblright\ and at this moment $\delta=0$.

\subsubsection{Example 2}

We consider the example of a 4-qubit mixed entangled state presented by Liu
\emph{et al.} (2021) \cite{Liu2021}. Here $\rho_{AB}=\cos^{2}\theta\left\vert
LC_{4}\right\rangle \left\langle LC_{4}\right\vert +\sin^{2}\theta\left\vert
LC_{4}^{\prime}\right\rangle \left\langle LC_{4}^{\prime}\right\vert $, where
\begin{equation}%
\begin{array}
[c]{c}%
\left\vert LC_{4}\right\rangle =\dfrac{1}{2}\left(  \left\vert
0000\right\rangle +\left\vert 1100\right\rangle +\left\vert 0011\right\rangle
-\left\vert 1111\right\rangle \right)  ,\\[9pt]%
\left\vert LC_{4}^{\prime}\right\rangle =\dfrac{1}{2}\left(  \left\vert
0100\right\rangle +\left\vert 1000\right\rangle +\left\vert 0111\right\rangle
-\left\vert 1011\right\rangle \right)  ,
\end{array}
\label{eq2.06}%
\end{equation}
are the linear cluster states. They can be written as
\begin{align*}
\left\vert LC_{4}\right\rangle  &  =\dfrac{1}{4}\left(  \left\vert
\uparrow\right\rangle +\left\vert \downarrow\right\rangle \right)  \left(
\left\vert +\right\rangle +\left\vert -\right\rangle \right)  \left\vert
00\right\rangle -\dfrac{\mathrm{i}}{4}\left(  \left\vert \uparrow\right\rangle
-\left\vert \downarrow\right\rangle \right)  \left(  \left\vert +\right\rangle
-\left\vert -\right\rangle \right)  \left\vert 00\right\rangle \\
&  +\dfrac{1}{4}\left(  \left\vert \uparrow\right\rangle +\left\vert
\downarrow\right\rangle \right)  \left(  \left\vert +\right\rangle +\left\vert
-\right\rangle \right)  \left\vert 11\right\rangle \newline+\dfrac{\mathrm{i}%
}{4}\left(  \left\vert \uparrow\right\rangle -\left\vert \downarrow
\right\rangle \right)  \left(  \left\vert +\right\rangle -\left\vert
-\right\rangle \right)  \left\vert 11\right\rangle \\
&  =\dfrac{1}{4}\left[  \left\vert \uparrow+\right\rangle \left(  \left\vert
00\right\rangle -\mathrm{i}\left\vert 00\right\rangle +\left\vert
11\right\rangle +\mathrm{i}\left\vert 11\right\rangle \newline\right)
+\left\vert \uparrow-\right\rangle \left(  \left\vert 00\right\rangle
+\mathrm{i}\left\vert 00\right\rangle +\left\vert 11\right\rangle
-\mathrm{i}\left\vert 11\right\rangle \newline\right)  \right.  \\
&  \left.  +\left\vert \downarrow+\right\rangle \left(  \left\vert
00\right\rangle +\mathrm{i}\left\vert 00\right\rangle +\left\vert
11\right\rangle -\mathrm{i}\left\vert 11\right\rangle \newline\right)
+\left\vert \downarrow-\right\rangle \left(  \left\vert 00\right\rangle
-\mathrm{i}\left\vert 00\right\rangle +\left\vert 11\right\rangle
+\mathrm{i}\left\vert 11\right\rangle \newline\right)  \right]  \\
&  =\dfrac{1}{4}\left\{  \left(  \left\vert \uparrow+\right\rangle +\left\vert
\downarrow-\right\rangle \right)  \left[  \left(  1-\mathrm{i}\right)
\left\vert 00\right\rangle +\left(  1+\mathrm{i}\right)  \left\vert
11\right\rangle \newline\right]  +\left(  \left\vert \uparrow-\right\rangle
+\left\vert \downarrow+\right\rangle \right)  \left[  \left(  1+\mathrm{i}%
\right)  \left\vert 00\right\rangle +\left(  1-\mathrm{i}\right)  \left\vert
11\right\rangle \newline\right]  \right\}  \\
&  =\dfrac{1}{4}\left\{  \left(  \left\vert \uparrow+\right\rangle +\left\vert
\downarrow-\right\rangle \right)  \left(  1-\mathrm{i}\right)  \left(
\left\vert 00\right\rangle +\mathrm{i}\left\vert 11\right\rangle
\newline\right)  +\left(  \left\vert \uparrow-\right\rangle +\left\vert
\downarrow+\right\rangle \right)  \left(  1+\mathrm{i}\right)  \left(
\left\vert 00\right\rangle -\mathrm{i}\left\vert 11\right\rangle \right)
\newline\right\}  ,
\end{align*}%
\begin{align*}
\left\vert LC_{4}^{\prime}\right\rangle  &  =\dfrac{1}{4}\left(  \left\vert
\uparrow\right\rangle +\left\vert \downarrow\right\rangle \right)  \left(
\left\vert +\right\rangle -\left\vert -\right\rangle \right)  \left\vert
00\right\rangle -\dfrac{\mathrm{i}}{4}\left(  \left\vert \uparrow\right\rangle
-\left\vert \downarrow\right\rangle \right)  \left(  \left\vert +\right\rangle
+\left\vert -\right\rangle \right)  \left\vert 00\right\rangle \\
&  +\dfrac{1}{4}\left(  \left\vert \uparrow\right\rangle +\left\vert
\downarrow\right\rangle \right)  \left(  \left\vert +\right\rangle -\left\vert
-\right\rangle \right)  \left\vert 11\right\rangle \newline+\dfrac{\mathrm{i}%
}{4}\left(  \left\vert \uparrow\right\rangle -\left\vert \downarrow
\right\rangle \right)  \left(  \left\vert +\right\rangle +\left\vert
-\right\rangle \right)  \left\vert 11\right\rangle \\
&  =\dfrac{1}{4}\left[  \left\vert \uparrow+\right\rangle \left(  \left\vert
00\right\rangle -\mathrm{i}\left\vert 00\right\rangle +\left\vert
11\right\rangle +\mathrm{i}\left\vert 11\right\rangle \newline\right)
+\left\vert \uparrow-\right\rangle \left(  -\left\vert 00\right\rangle
-\mathrm{i}\left\vert 00\right\rangle -\left\vert 11\right\rangle
+\mathrm{i}\left\vert 11\right\rangle \newline\right)  \right.  \\
&  \left.  +\left\vert \downarrow+\right\rangle \left(  \left\vert
00\right\rangle +\mathrm{i}\left\vert 00\right\rangle +\left\vert
11\right\rangle -\mathrm{i}\left\vert 11\right\rangle \newline\right)
+\left\vert \downarrow-\right\rangle \left(  -\left\vert 00\right\rangle
+\mathrm{i}\left\vert 00\right\rangle -\left\vert 11\right\rangle
-\mathrm{i}\left\vert 11\right\rangle \newline\right)  \right]  \\
&  \dfrac{1}{4}\left\{  \left(  \left\vert \uparrow+\right\rangle -\left\vert
\downarrow-\right\rangle \right)  \left[  \left(  1-\mathrm{i}\right)
\left\vert 00\right\rangle +\left(  1+\mathrm{i}\right)  \left\vert
11\right\rangle \newline\right]  +\left(  -\left\vert \uparrow-\right\rangle
+\left\vert \downarrow+\right\rangle \right)  \left[  \left(  1+\mathrm{i}%
\right)  \left\vert 00\right\rangle +\left(  1-\mathrm{i}\right)  \left\vert
11\right\rangle \newline\right]  \right\}  \\
&  =\dfrac{1}{4}\left\{  \left(  \left\vert \uparrow+\right\rangle -\left\vert
\downarrow-\right\rangle \right)  \left(  1-\mathrm{i}\right)  \left(
\left\vert 00\right\rangle +\mathrm{i}\left\vert 11\right\rangle
\newline\right)  +\left(  -\left\vert \uparrow-\right\rangle +\left\vert
\downarrow+\right\rangle \right)  \left(  1+\mathrm{i}\right)  \left(
\left\vert 00\right\rangle -\mathrm{i}\left\vert 11\right\rangle \right)
\newline\right\}  ,
\end{align*}
where $|\pm\rangle=(1/\sqrt{2})(\left\vert 0\right\rangle \pm\left\vert
1\right\rangle )$, $|\updownarrow\rangle=(1/\sqrt{2})(\left\vert
0\right\rangle \pm\mathrm{i}\left\vert 1\right\rangle )$. Alice prepares the
state $\rho_{AB}$. She keeps particles 1, 2, and sends particles 3, 4 to Bob.
In the 2-setting steering protocol $\left\{  \hat{n}_{1},\hat{n}_{2}\right\}
$ $\left(  \hat{n}_{1}\neq\hat{n}_{2}\right)  $, with%
\begin{equation}%
\begin{array}
[c]{c}%
\hat{n}_{1}=\sigma_{z}\sigma_{z}\equiv zz,\\
\hat{n}_{2}=\sigma_{y}\sigma_{x}\equiv yx.
\end{array}
\label{eq2.07}%
\end{equation}
In the protocol, Bob asks Alice to carry out either one of two possible
projective measurements on her qubits, i.e.,%
\begin{equation}
\left\{
\begin{array}
[c]{l}%
P_{00}^{\hat{n}_{1}}=\left\vert 00\right\rangle \left\langle 00\right\vert ,\\
P_{01}^{\hat{n}_{1}}=\left\vert 01\right\rangle \left\langle 01\right\vert ,\\
P_{10}^{\hat{n}_{1}}=\left\vert 10\right\rangle \left\langle 10\right\vert ,\\
P_{11}^{\hat{n}_{1}}=\left\vert 11\right\rangle \left\langle 11\right\vert ,\\
P_{00}^{\hat{n}_{2}}=\left\vert \uparrow+\right\rangle \left\langle
\uparrow+\right\vert ,\\
P_{01}^{\hat{n}_{2}}=\left\vert \uparrow-\right\rangle \left\langle
\uparrow-\right\vert ,\\
P_{10}^{\hat{n}_{2}}=\left\vert \downarrow+\right\rangle \left\langle
\downarrow+\right\vert ,\\
P_{11}^{\hat{n}_{2}}=\left\vert \downarrow-\right\rangle \left\langle
\downarrow-\right\vert .
\end{array}
\right.  \label{eq2.08}%
\end{equation}
After Alice's measurement, Bob's unnormalized conditional states are
\begin{equation}
\left\{
\begin{array}
[c]{l}%
\tilde{\rho}_{00}^{\hat{n}_{1}}=\dfrac{1}{4}\cos^{2}\theta\left(  \left\vert
00\right\rangle +\left\vert 11\right\rangle \right)  \left(  \left\langle
00\right\vert +\left\langle 11\right\vert \right)  ,\\[9pt]%
\tilde{\rho}_{01}^{\hat{n}_{1}}=\dfrac{1}{4}\sin^{2}\theta\left(  \left\vert
00\right\rangle +\left\vert 11\right\rangle \right)  \left(  \left\langle
00\right\vert +\left\langle 11\right\vert \right)  ,\\[9pt]%
\tilde{\rho}_{10}^{\hat{n}_{1}}=\dfrac{1}{4}\sin^{2}\theta\left(  \left\vert
00\right\rangle -\left\vert 11\right\rangle \right)  \left(  \left\langle
00\right\vert -\left\langle 11\right\vert \right)  ,\\[9pt]%
\tilde{\rho}_{11}^{\hat{n}_{1}}=\dfrac{1}{4}\cos^{2}\theta\left(  \left\vert
00\right\rangle -\left\vert 11\right\rangle \right)  \left(  \left\langle
00\right\vert -\left\langle 11\right\vert \right)  ,\\[9pt]%
\tilde{\rho}_{00}^{\hat{n}_{2}}=\dfrac{1}{8}\left(  \left\vert 00\right\rangle
+\mathrm{i}\left\vert 11\right\rangle \right)  \left(  \left\langle
00\right\vert -\mathrm{i}\left\langle 11\right\vert \right)  ,\\[9pt]%
\tilde{\rho}_{01}^{\hat{n}_{2}}=\dfrac{1}{8}\left(  \left\vert 00\right\rangle
-\mathrm{i}\left\vert 11\right\rangle \right)  \left(  \left\langle
00\right\vert +\mathrm{i}\left\langle 11\right\vert \right)  ,\\[9pt]%
\tilde{\rho}_{10}^{\hat{n}_{2}}=\dfrac{1}{8}\left(  \left\vert 00\right\rangle
-\mathrm{i}\left\vert 11\right\rangle \right)  \left(  \left\langle
00\right\vert +\mathrm{i}\left\langle 11\right\vert \right)  ,\\[9pt]%
\tilde{\rho}_{11}^{\hat{n}_{2}}=\dfrac{1}{8}\left(  \left\vert 00\right\rangle
+\mathrm{i}\left\vert 11\right\rangle \right)  \left(  \left\langle
00\right\vert -\mathrm{i}\left\langle 11\right\vert \right)  ,
\end{array}
\right.  \label{2.09}%
\end{equation}
where%

\begin{align*}
\tilde{\rho}_{00}^{\hat{n}_{1}} &  =\mathrm{tr}_{A}\left[  \left(
P_{00}^{\hat{n}_{1}}{\otimes\mathds{1}}\right)  \rho_{AB}\right]  \\
&  =\mathrm{tr}_{A}\left[  \left(  \left\vert 00\right\rangle \left\langle
00\right\vert {\otimes\mathds{1}}\right)  \left(  \cos^{2}\theta\left\vert
LC_{4}\right\rangle \left\langle LC_{4}\right\vert +\sin^{2}\theta\left\vert
LC_{4}^{\prime}\right\rangle \left\langle LC_{4}^{\prime}\right\vert \right)
\right]  \\
&  =\dfrac{1}{4}\mathrm{tr}_{A}\left[  \left(  \left\vert 00\right\rangle
_{A}{\otimes\mathds{1}}\right)  \left(  \cos^{2}\theta\left(  \left\vert
00\right\rangle +\left\vert 11\right\rangle \right)  \right)  \left(
\left\langle 0000\right\vert +\left\langle 1100\right\vert +\left\langle
0011\right\vert -\left\langle 1111\right\vert \right)  \right]  \\
&  =\dfrac{1}{4}\left(  \left\langle 00\right\vert _{A}{\otimes\mathds{1}}%
\right)  \left(  \left\vert 00\right\rangle _{A}{\otimes\mathds{1}}\right)
\left(  \cos^{2}\theta\left(  \left\vert 00\right\rangle +\left\vert
11\right\rangle \right)  \right)  \left[  \left\langle 00\right\vert
_{A}\left(  \left\langle 00\right\vert +\left\langle 11\right\vert \right)
+\left\langle 11\right\vert _{A}\left(  \left\langle 00\right\vert
-\left\langle 11\right\vert \right)  \right]  \left(  \left\vert
00\right\rangle _{A}{\otimes\mathds{1}}\right)  \\
&  =\dfrac{1}{4}\cos^{2}\theta\left(  \left\vert 00\right\rangle +\left\vert
11\right\rangle \right)  \left(  \left\langle 00\right\vert +\left\langle
11\right\vert \right)  ,
\end{align*}

\begin{align*}
\tilde{\rho}_{01}^{\hat{n}_{1}} &  =\mathrm{tr}_{A}\left[  \left(
P_{01}^{\hat{n}_{1}}{\otimes\mathds{1}}\right)  \rho_{AB}\right]  \\
&  =\mathrm{tr}_{A}\left[  \left(  \left\vert 01\right\rangle \left\langle
01\right\vert {\otimes\mathds{1}}\right)  \left(  \cos^{2}\theta\left\vert
LC_{4}\right\rangle \left\langle LC_{4}\right\vert +\sin^{2}\theta\left\vert
LC_{4}^{\prime}\right\rangle \left\langle LC_{4}^{\prime}\right\vert \right)
\right]  \\
&  =\dfrac{1}{4}\mathrm{tr}_{A}\left[  \left\vert 01\right\rangle _{A}\left(
\sin^{2}\theta\left(  \left\vert 00\right\rangle +\left\vert 11\right\rangle
\right)  \right)  \left(  \left\langle 0100\right\vert +\left\langle
1000\right\vert +\left\langle 0111\right\vert -\left\langle 1011\right\vert
\right)  \right]  \\
&  =\dfrac{1}{4}\left(  \left\langle 01\right\vert _{A}{\otimes\mathds{1}}%
\right)  \left(  \left\vert 01\right\rangle _{A}{\otimes\mathds{1}}\right)
\left(  \sin^{2}\theta\left(  \left\vert 00\right\rangle +\left\vert
11\right\rangle \right)  \right)  \left[  \left\langle 01\right\vert
_{A}\left(  \left\langle 00\right\vert +\left\langle 11\right\vert \right)
+\left\langle 10\right\vert _{A}\left(  \left\langle 00\right\vert
-\left\langle 11\right\vert \right)  \right]  \left(  \left\vert
01\right\rangle _{A}{\otimes\mathds{1}}\right)  \\
&  =\dfrac{1}{4}\sin^{2}\theta\left(  \left\vert 00\right\rangle +\left\vert
11\right\rangle \right)  \left(  \left\langle 00\right\vert +\left\langle
11\right\vert \right)  ,
\end{align*}

\begin{align*}
\tilde{\rho}_{10}^{\hat{n}_{1}} &  =\mathrm{tr}_{A}\left[  \left(
P_{10}^{\hat{n}_{1}}{\otimes\mathds{1}}\right)  \rho_{AB}\right]  \\
&  =\mathrm{tr}_{A}\left[  \left(  \left\vert 10\right\rangle \left\langle
10\right\vert {\otimes\mathds{1}}\right)  \left(  \cos^{2}\theta\left\vert
LC_{4}\right\rangle \left\langle LC_{4}\right\vert +\sin^{2}\theta\left\vert
LC_{4}^{\prime}\right\rangle \left\langle LC_{4}^{\prime}\right\vert \right)
\right]  \\
&  =\dfrac{1}{4}\mathrm{tr}_{A}\left[  \left\vert 10\right\rangle _{A}\left(
\sin^{2}\theta\left(  \left\vert 00\right\rangle -\left\vert 11\right\rangle
\right)  \right)  \left(  \left\langle 0100\right\vert +\left\langle
1000\right\vert +\left\langle 0111\right\vert -\left\langle 1011\right\vert
\right)  \right]  \\
&  =\dfrac{1}{4}\left(  \left\langle 10\right\vert _{A}{\otimes\mathds{1}}%
\right)  \left(  \left\vert 10\right\rangle _{A}{\otimes\mathds{1}}\right)
\left(  \sin^{2}\theta\left(  \left\vert 00\right\rangle -\left\vert
11\right\rangle \right)  \right)  \left[  \left\langle 01\right\vert
_{A}\left(  \left\langle 00\right\vert +\left\langle 11\right\vert \right)
+\left\langle 10\right\vert _{A}\left(  \left\langle 00\right\vert
-\left\langle 11\right\vert \right)  \right]  \left(  \left\vert
10\right\rangle _{A}{\otimes\mathds{1}}\right)  \\
&  =\dfrac{1}{4}\sin^{2}\theta\left(  \left\vert 00\right\rangle -\left\vert
11\right\rangle \right)  \left(  \left\langle 00\right\vert -\left\langle
11\right\vert \right)  ,
\end{align*}%
\begin{align*}
\tilde{\rho}_{11}^{\hat{n}_{1}} &  =\mathrm{tr}_{A}\left[  \left(
P_{11}^{\hat{n}_{1}}{\otimes\mathds{1}}\right)  \rho_{AB}\right]  \\
&  =\mathrm{tr}_{A}\left[  \left(  \left\vert 11\right\rangle \left\langle
11\right\vert {\otimes\mathds{1}}\right)  \left(  \cos^{2}\theta\left\vert
LC_{4}\right\rangle \left\langle LC_{4}\right\vert +\sin^{2}\theta\left\vert
LC_{4}^{\prime}\right\rangle \left\langle LC_{4}^{\prime}\right\vert \right)
\right]  \\
&  =\dfrac{1}{4}\mathrm{tr}_{A}\left[  \left\vert 11\right\rangle _{A}\left(
\cos^{2}\theta\left(  \left\vert 00\right\rangle -\left\vert 11\right\rangle
\right)  \right)  \left(  \left\langle 0000\right\vert +\left\langle
1100\right\vert +\left\langle 0011\right\vert -\left\langle 1111\right\vert
\right)  \right]  \\
&  =\dfrac{1}{4}\left(  \left\langle 11\right\vert _{A}{\otimes\mathds{1}}%
\right)  \left(  \left\vert 11\right\rangle _{A}{\otimes\mathds{1}}\right)
\left(  \cos^{2}\theta\left(  \left\vert 00\right\rangle -\left\vert
11\right\rangle \right)  \right)  \left[  \left\langle 00\right\vert
_{A}\left(  \left\langle 00\right\vert +\left\langle 11\right\vert \right)
+\left\langle 11\right\vert _{A}\left(  \left\langle 00\right\vert
-\left\langle 11\right\vert \right)  \right]  \left(  \left\vert
11\right\rangle _{A}{\otimes\mathds{1}}\right)  \\
&  =\dfrac{1}{4}\cos^{2}\theta\left(  \left\vert 00\right\rangle -\left\vert
11\right\rangle \right)  \left(  \left\langle 00\right\vert -\left\langle
11\right\vert \right)  ,
\end{align*}

\begin{align*}
\tilde{\rho}_{00}^{\hat{n}_{2}} &  =\mathrm{tr}_{A}\left[  \left(
P_{00}^{\hat{n}_{2}}{\otimes\mathds{1}}\right)  \rho_{AB}\right]  \\
&  =\mathrm{tr}_{A}\left[  \left(  \left\vert \uparrow+\right\rangle
\left\langle \uparrow+\right\vert {\otimes\mathds{1}}\right)  \left(  \cos
^{2}\theta\left\vert LC_{4}\right\rangle \left\langle LC_{4}\right\vert
+\sin^{2}\theta\left\vert LC_{4}^{\prime}\right\rangle \left\langle
LC_{4}^{\prime}\right\vert \right)  \right]  \\
&  =\dfrac{1}{4}\mathrm{tr}_{A}\left\{  \left(  \left\vert \uparrow
+\right\rangle _{A}{\otimes\mathds{1}}\right)  \left[  \cos^{2}\theta\left(
1-\mathrm{i}\right)  \left(  \left\vert 00\right\rangle +\mathrm{i}\left\vert
11\right\rangle \newline\right)  \left\langle LC_{4}\right\vert +\sin
^{2}\theta\left(  1-\mathrm{i}\right)  \left(  \left\vert 00\right\rangle
+\mathrm{i}\left\vert 11\right\rangle \newline\right)  \left\langle
LC_{4}^{\prime}\right\vert \right]  \right\}  \\
&  =\dfrac{1}{4}\left(  \left\langle \uparrow+\right\vert _{A}{\otimes
\mathds{1}}\right)  \left(  \left\vert \uparrow+\right\rangle _{A}%
{\otimes\mathds{1}}\right)  \left[  \cos^{2}\theta\left(  1-\mathrm{i}\right)
\left(  \left\vert 00\right\rangle +\mathrm{i}\left\vert 11\right\rangle
\newline\right)  \left\langle LC_{4}\right\vert +\sin^{2}\theta\left(
1-\mathrm{i}\right)  \left(  \left\vert 00\right\rangle +\mathrm{i}\left\vert
11\right\rangle \newline\right)  \left\langle LC_{4}^{\prime}\right\vert
\right]  \left(  \left\vert \uparrow+\right\rangle _{A}{\otimes\mathds{1}}%
\right)  \\
&  =\dfrac{1}{16}\left(  \left\langle \uparrow+\right\vert _{A}{\otimes
\mathds{1}}\right)  \left(  \left\vert \uparrow+\right\rangle _{A}%
{\otimes\mathds{1}}\right)  \cos^{2}\theta\left(  1-\mathrm{i}\right)  \left(
\left\vert 00\right\rangle +\mathrm{i}\left\vert 11\right\rangle
\newline\right)  \left[  \left(  \left\langle \uparrow+\right\vert
+\left\langle \downarrow-\right\vert \right)  \left(  1+\mathrm{i}\right)
\left(  \left\langle 00\right\vert -\mathrm{i}\left\langle 11\right\vert
\newline\right)  \right.  \\
&  \left.  +\left(  \left\langle \left\vert \uparrow-\right\rangle \right\vert
+\left\langle \downarrow+\right\vert \right)  \left(  1-\mathrm{i}\right)
\left(  \left\langle 00\right\vert +\mathrm{i}\left\langle 11\right\vert
\right)  \right]  \newline\left(  \left\vert \uparrow+\right\rangle
_{A}{\otimes\mathds{1}}\right)  \\
&  +\left(  \left\langle \uparrow+\right\vert _{A}{\otimes\mathds{1}}\right)
\left(  \left\vert \uparrow+\right\rangle _{A}{\otimes\mathds{1}}\right)
\sin^{2}\theta\left(  1-\mathrm{i}\right)  \left(  \left\vert 00\right\rangle
+\mathrm{i}\left\vert 11\right\rangle \newline\right)  \left[  \left(
\left\langle \uparrow+\right\vert -\left\langle \downarrow-\right\vert
\right)  \left(  1+\mathrm{i}\right)  \left(  \left\langle 00\right\vert
-\mathrm{i}\left\langle 11\right\vert \newline\right)  \right.  \\
&  \left.  +\left(  -\left\langle \uparrow-\right\vert +\left\langle
\downarrow+\right\vert \right)  \left(  1-\mathrm{i}\right)  \left(
\left\langle 00\right\vert +\mathrm{i}\left\langle 11\right\vert \right)
\right]  \newline\left(  \left\vert \uparrow+\right\rangle _{A}{\otimes
\mathds{1}}\right)  \\
&  =\dfrac{1}{16}\left(  1-\mathrm{i}\right)  \left(  \left\vert
00\right\rangle +\mathrm{i}\left\vert 11\right\rangle \newline\right)
\left\langle \uparrow+\right\vert _{A}\left\{  \left\vert \uparrow
+\right\rangle \left(  1+\mathrm{i}\right)  \left(  \left\langle 00\right\vert
-\mathrm{i}\left\langle 11\right\vert \newline\right)  \newline\right\}  \\
&  =\dfrac{1}{16}\left(  1-\mathrm{i}\right)  \left(  \left\vert
00\right\rangle +\mathrm{i}\left\vert 11\right\rangle \newline\right)  \left(
1+\mathrm{i}\right)  \left(  \left\langle 00\right\vert -\mathrm{i}%
\left\langle 11\right\vert \newline\right)  \newline\\
&  =\dfrac{1}{8}\left(  \left\vert 00\right\rangle +\mathrm{i}\left\vert
11\right\rangle \right)  \left(  \left\langle 00\right\vert -\mathrm{i}%
\left\langle 11\right\vert \right)  ,
\end{align*}

\begin{align*}
\tilde{\rho}_{01}^{\hat{n}_{2}} &  =\mathrm{tr}_{A}\left[  \left(
P_{01}^{\hat{n}_{2}}{\otimes\mathds{1}}\right)  \rho_{AB}\right]  \\
&  =\mathrm{tr}_{A}\left[  \left(  \left\vert \uparrow-\right\rangle
\left\langle \uparrow-\right\vert {\otimes\mathds{1}}\right)  \left(  \cos
^{2}\theta\left\vert LC_{4}\right\rangle \left\langle LC_{4}\right\vert
+\sin^{2}\theta\left\vert LC_{4}^{\prime}\right\rangle \left\langle
LC_{4}^{\prime}\right\vert \right)  \right]  \\
&  =\dfrac{1}{4}\mathrm{tr}_{A}\left\{  \left(  \left\vert \uparrow
-\right\rangle _{A}{\otimes\mathds{1}}\right)  \left[  \cos^{2}\theta\left(
1+\mathrm{i}\right)  \left(  \left\vert 00\right\rangle -\mathrm{i}\left\vert
11\right\rangle \right)  \left\langle LC_{4}\right\vert -\sin^{2}\theta\left(
1+\mathrm{i}\right)  \left(  \left\vert 00\right\rangle -\mathrm{i}\left\vert
11\right\rangle \right)  \left\langle LC_{4}^{\prime}\right\vert \right]
\right\}  ,\\
&  =\dfrac{1}{4}\left(  \left\langle \uparrow-\right\vert _{A}{\otimes
\mathds{1}}\right)  \left(  \left\vert \uparrow-\right\rangle _{A}%
{\otimes\mathds{1}}\right)  \left[  \cos^{2}\theta\left(  1+\mathrm{i}\right)
\left(  \left\vert 00\right\rangle -\mathrm{i}\left\vert 11\right\rangle
\right)  \left\langle LC_{4}\right\vert -\sin^{2}\theta\left(  1+\mathrm{i}%
\right)  \left(  \left\vert 00\right\rangle -\mathrm{i}\left\vert
11\right\rangle \right)  \left\langle LC_{4}^{\prime}\right\vert \right]
\left(  \left\vert \uparrow-\right\rangle _{A}{\otimes\mathds{1}}\right)  ,\\
&  =\dfrac{1}{16}\left(  \left\langle \uparrow-\right\vert _{A}{\otimes
\mathds{1}}\right)  \left(  \left\vert \uparrow-\right\rangle _{A}%
{\otimes\mathds{1}}\right)  \cos^{2}\theta\left(  1+\mathrm{i}\right)  \left(
\left\vert 00\right\rangle -\mathrm{i}\left\vert 11\right\rangle \right)
\left[  \left(  \left\langle \uparrow+\right\vert +\left\langle \downarrow
-\right\vert \right)  \left(  1+\mathrm{i}\right)  \left(  \left\langle
00\right\vert -\mathrm{i}\left\langle 11\right\vert \newline\right)  \right.
\\
&  \left.  +\left(  \left\langle \left\vert \uparrow-\right\rangle \right\vert
+\left\langle \downarrow+\right\vert \right)  \left(  1-\mathrm{i}\right)
\left(  \left\langle 00\right\vert +\mathrm{i}\left\langle 11\right\vert
\right)  \right]  \newline\left(  \left\vert \uparrow-\right\rangle
_{A}{\otimes\mathds{1}}\right)  \\
&  -\left(  \left\langle \uparrow-\right\vert _{A}{\otimes\mathds{1}}\right)
\left(  \left\vert \uparrow-\right\rangle _{A}{\otimes\mathds{1}}\right)
\sin^{2}\theta\left(  1+\mathrm{i}\right)  \left(  \left\vert 00\right\rangle
-\mathrm{i}\left\vert 11\right\rangle \right)  \left[  \left(  \left\langle
\uparrow+\right\vert -\left\langle \downarrow-\right\vert \right)  \left(
1+\mathrm{i}\right)  \left(  \left\langle 00\right\vert -\mathrm{i}%
\left\langle 11\right\vert \newline\right)  \right.  \\
&  \left.  +\left(  -\left\langle \uparrow-\right\vert +\left\langle
\downarrow+\right\vert \right)  \left(  1-\mathrm{i}\right)  \left(
\left\langle 00\right\vert +\mathrm{i}\left\langle 11\right\vert \right)
\right]  \newline\left(  \left\vert \uparrow-\right\rangle _{A}{\otimes
\mathds{1}}\right)  ,\\
&  =\dfrac{1}{16}\cos^{2}\theta\left(  1+\mathrm{i}\right)  \left(  \left\vert
00\right\rangle -\mathrm{i}\left\vert 11\right\rangle \right)  \left(
1-\mathrm{i}\right)  \left(  \left\langle 00\right\vert +\mathrm{i}%
\left\langle 11\right\vert \right)  +\sin^{2}\theta\left(  1+\mathrm{i}%
\right)  \left(  \left\vert 00\right\rangle -\mathrm{i}\left\vert
11\right\rangle \right)  \left(  1-\mathrm{i}\right)  \left(  \left\langle
00\right\vert +\mathrm{i}\left\langle 11\right\vert \right)  \newline%
\newline,\\
&  =\dfrac{1}{16}\left(  1+\mathrm{i}\right)  \left(  1-\mathrm{i}\right)
\left(  \left\vert 00\right\rangle -\mathrm{i}\left\vert 11\right\rangle
\right)  \left(  \left\langle 00\right\vert +\mathrm{i}\left\langle
11\right\vert \right)  \newline,\\
&  =\dfrac{1}{8}\left(  \left\vert 00\right\rangle -\mathrm{i}\left\vert
11\right\rangle \right)  \left(  \left\langle 00\right\vert +\mathrm{i}%
\left\langle 11\right\vert \right)  ,
\end{align*}

\begin{align*}
\tilde{\rho}_{10}^{\hat{n}_{2}} &  =\mathrm{tr}_{A}\left[  \left(
P_{10}^{\hat{n}_{2}}{\otimes\mathds{1}}\right)  \rho_{AB}\right]  ,\\
&  =\mathrm{tr}_{A}\left[  \left(  \left\vert \downarrow+\right\rangle
\left\langle \downarrow+\right\vert {\otimes\mathds{1}}\right)  \left(
\cos^{2}\theta\left\vert LC_{4}\right\rangle \left\langle LC_{4}\right\vert
+\sin^{2}\theta\left\vert LC_{4}^{\prime}\right\rangle \left\langle
LC_{4}^{\prime}\right\vert \right)  \right]  ,\\
&  =\dfrac{1}{4}\mathrm{tr}_{A}\left\{  \left(  \left\vert \downarrow
+\right\rangle _{A}{\otimes\mathds{1}}\right)  \left[  \cos^{2}\theta\left(
1+\mathrm{i}\right)  \left(  \left\vert 00\right\rangle -\mathrm{i}\left\vert
11\right\rangle \right)  \left\langle LC_{4}\right\vert +\sin^{2}\theta\left(
1+\mathrm{i}\right)  \left(  \left\vert 00\right\rangle -\mathrm{i}\left\vert
11\right\rangle \right)  \left\langle LC_{4}^{\prime}\right\vert \right]
\right\}  \\
&  =\dfrac{1}{4}\left(  \left\langle \downarrow+\right\vert _{A}%
{\otimes\mathds{1}}\right)  \left(  \left\vert \downarrow+\right\rangle
_{A}{\otimes\mathds{1}}\right)  \left[  \cos^{2}\theta\left(  1+\mathrm{i}%
\right)  \left(  \left\vert 00\right\rangle -\mathrm{i}\left\vert
11\right\rangle \right)  \left\langle LC_{4}\right\vert +\sin^{2}\theta\left(
1+\mathrm{i}\right)  \left(  \left\vert 00\right\rangle -\mathrm{i}\left\vert
11\right\rangle \right)  \left\langle LC_{4}^{\prime}\right\vert \right]
\left(  \left\vert \downarrow+\right\rangle _{A}{\otimes\mathds{1}}\right)
,\\
&  =\dfrac{1}{16}\left(  \left\langle \downarrow+\right\vert _{A}%
{\otimes\mathds{1}}\right)  \left(  \left\vert \downarrow+\right\rangle
_{A}{\otimes\mathds{1}}\right)  \cos^{2}\theta\left(  1+\mathrm{i}\right)
\left(  \left\vert 00\right\rangle -\mathrm{i}\left\vert 11\right\rangle
\right)  \left[  \left(  \left\langle \uparrow+\right\vert +\left\langle
\downarrow-\right\vert \right)  \left(  1+\mathrm{i}\right)  \left(
\left\langle 00\right\vert -\mathrm{i}\left\langle 11\right\vert
\newline\right)  \right.  \\
&  \left.  +\left(  \left\langle \left\vert \uparrow-\right\rangle \right\vert
+\left\langle \downarrow+\right\vert \right)  \left(  1-\mathrm{i}\right)
\left(  \left\langle 00\right\vert +\mathrm{i}\left\langle 11\right\vert
\right)  \right]  \newline\left(  \left\vert \downarrow+\right\rangle
_{A}{\otimes\mathds{1}}\right)  \\
&  +\left(  \left\langle \downarrow+\right\vert _{A}{\otimes\mathds{1}}%
\right)  \left(  \left\vert \downarrow+\right\rangle _{A}{\otimes
\mathds{1}}\right)  \sin^{2}\theta\left(  1+\mathrm{i}\right)  \left(
\left\vert 00\right\rangle -\mathrm{i}\left\vert 11\right\rangle \right)
\left[  \left(  \left\langle \uparrow+\right\vert -\left\langle \downarrow
-\right\vert \right)  \left(  1+\mathrm{i}\right)  \left(  \left\langle
00\right\vert -\mathrm{i}\left\langle 11\right\vert \newline\right)  \right.
\\
&  \left.  +\left(  -\left\langle \uparrow-\right\vert +\left\langle
\downarrow+\right\vert \right)  \left(  1-\mathrm{i}\right)  \left(
\left\langle 00\right\vert +\mathrm{i}\left\langle 11\right\vert \right)
\right]  \newline\left(  \left\vert \downarrow+\right\rangle _{A}%
{\otimes\mathds{1}}\right)  ,\\
&  =\dfrac{1}{16}\cos^{2}\theta\left(  1+\mathrm{i}\right)  \left(  \left\vert
00\right\rangle -\mathrm{i}\left\vert 11\right\rangle \right)  \left(
1-\mathrm{i}\right)  \left(  \left\langle 00\right\vert +\mathrm{i}%
\left\langle 11\right\vert \right)  +\sin^{2}\theta\left(  1+\mathrm{i}%
\right)  \left(  \left\vert 00\right\rangle -\mathrm{i}\left\vert
11\right\rangle \right)  \left(  1-\mathrm{i}\right)  \left(  \left\langle
00\right\vert +\mathrm{i}\left\langle 11\right\vert \right)  \newline%
\newline,\\
&  =\dfrac{1}{16}\left(  1+\mathrm{i}\right)  \left(  \left\vert
00\right\rangle -\mathrm{i}\left\vert 11\right\rangle \right)  \left(
1-\mathrm{i}\right)  \left(  \left\langle 00\right\vert +\mathrm{i}%
\left\langle 11\right\vert \right)  \newline,\\
&  =\dfrac{1}{8}\left(  \left\vert 00\right\rangle -\mathrm{i}\left\vert
11\right\rangle \right)  \left(  \left\langle 00\right\vert +\mathrm{i}%
\left\langle 11\right\vert \right)  ,
\end{align*}

\begin{align*}
\tilde{\rho}_{11}^{\hat{n}_{2}} &  =\mathrm{tr}_{A}\left[  \left(
P_{11}^{\hat{n}_{2}}{\otimes\mathds{1}}\right)  \rho_{AB}\right]  ,\\
&  =\mathrm{tr}_{A}\left[  \left(  \left\vert \downarrow-\right\rangle
\left\langle \downarrow-\right\vert {\otimes\mathds{1}}\right)  \left(
\cos^{2}\theta\left\vert LC_{4}\right\rangle \left\langle LC_{4}\right\vert
+\sin^{2}\theta\left\vert LC_{4}^{\prime}\right\rangle \left\langle
LC_{4}^{\prime}\right\vert \right)  \right]  ,\\
&  =\dfrac{1}{4}\mathrm{tr}_{A}\left\{  \left(  \left\vert \downarrow
-\right\rangle _{A}{\otimes\mathds{1}}\right)  \left[  \cos^{2}\theta\left(
1-\mathrm{i}\right)  \left(  \left\vert 00\right\rangle +\mathrm{i}\left\vert
11\right\rangle \newline\right)  \left\langle LC_{4}\right\vert -\sin
^{2}\theta\left(  1-\mathrm{i}\right)  \left(  \left\vert 00\right\rangle
+\mathrm{i}\left\vert 11\right\rangle \newline\right)  \left\langle
LC_{4}^{\prime}\right\vert \right]  \right\}  ,\\
&  =\dfrac{1}{4}\left(  \left\langle \downarrow-\right\vert _{A}%
{\otimes\mathds{1}}\right)  \left(  \left\vert \downarrow-\right\rangle
_{A}{\otimes\mathds{1}}\right)  \left[  \cos^{2}\theta\left(  1-\mathrm{i}%
\right)  \left(  \left\vert 00\right\rangle +\mathrm{i}\left\vert
11\right\rangle \newline\right)  \left\langle LC_{4}\right\vert -\sin
^{2}\theta\left(  1-\mathrm{i}\right)  \left(  \left\vert 00\right\rangle
+\mathrm{i}\left\vert 11\right\rangle \newline\right)  \left\langle
LC_{4}^{\prime}\right\vert \right]  \left(  \left\vert \downarrow
-\right\rangle _{A}{\otimes\mathds{1}}\right)  \\
&  =\dfrac{1}{16}\left(  \left\langle \downarrow-\right\vert _{A}%
{\otimes\mathds{1}}\right)  \left(  \left\vert \downarrow-\right\rangle
_{A}{\otimes\mathds{1}}\right)  \cos^{2}\theta\left(  1-\mathrm{i}\right)
\left(  \left\vert 00\right\rangle +\mathrm{i}\left\vert 11\right\rangle
\newline\right)  \left[  \left(  \left\langle \uparrow+\right\vert
+\left\langle \downarrow-\right\vert \right)  \left(  1+\mathrm{i}\right)
\left(  \left\langle 00\right\vert -\mathrm{i}\left\langle 11\right\vert
\newline\right)  \right.  \\
&  \left.  +\left(  \left\langle \left\vert \uparrow-\right\rangle \right\vert
+\left\langle \downarrow+\right\vert \right)  \left(  1-\mathrm{i}\right)
\left(  \left\langle 00\right\vert +\mathrm{i}\left\langle 11\right\vert
\right)  \right]  \newline\left(  \left\vert \downarrow-\right\rangle
_{A}{\otimes\mathds{1}}\right)  \\
&  -\left(  \left\langle \downarrow-\right\vert _{A}{\otimes\mathds{1}}%
\right)  \left(  \left\vert \downarrow-\right\rangle _{A}{\otimes
\mathds{1}}\right)  \sin^{2}\theta\left(  1-\mathrm{i}\right)  \left(
\left\vert 00\right\rangle +\mathrm{i}\left\vert 11\right\rangle
\newline\right)  \left[  \left(  \left\langle \uparrow+\right\vert
-\left\langle \downarrow-\right\vert \right)  \left(  1+\mathrm{i}\right)
\left(  \left\langle 00\right\vert -\mathrm{i}\left\langle 11\right\vert
\newline\right)  \right.  \\
&  \left.  +\left(  -\left\langle \uparrow-\right\vert +\left\langle
\downarrow+\right\vert \right)  \left(  1-\mathrm{i}\right)  \left(
\left\langle 00\right\vert +\mathrm{i}\left\langle 11\right\vert \right)
\right]  \newline\left(  \left\vert \downarrow-\right\rangle _{A}%
{\otimes\mathds{1}}\right)  \\
&  =\dfrac{1}{16}\cos^{2}\theta\left(  1-\mathrm{i}\right)  \left(  \left\vert
00\right\rangle +\mathrm{i}\left\vert 11\right\rangle \newline\right)  \left(
1+\mathrm{i}\right)  \left(  \left\langle 00\right\vert -\mathrm{i}%
\left\langle 11\right\vert \newline\right)  +\sin^{2}\theta\left(
1-\mathrm{i}\right)  \left(  \left\vert 00\right\rangle +\mathrm{i}\left\vert
11\right\rangle \newline\right)  \left(  1+\mathrm{i}\right)  \left(
\left\langle 00\right\vert -\mathrm{i}\left\langle 11\right\vert
\newline\right)  \newline\newline\\
&  =\dfrac{1}{16}\left(  1-\mathrm{i}\right)  \left(  \left\vert
00\right\rangle +\mathrm{i}\left\vert 11\right\rangle \newline\right)  \left(
1+\mathrm{i}\right)  \left(  \left\langle 00\right\vert -\mathrm{i}%
\left\langle 11\right\vert \newline\right)  \newline\\
&  =\dfrac{1}{8}\left(  \left\vert 00\right\rangle +\mathrm{i}\left\vert
11\right\rangle \right)  \left(  \left\langle 00\right\vert -\mathrm{i}%
\left\langle 11\right\vert \right)  .
\end{align*}
In the quantum result, Bob's states are all pure. Although there are identical
results in the same set of measurements, Bob's two sets of results are
completely different, which satisfies the requirements in the theorem.

If Bob's states have a LHS description, they must satisfy Eqs. (\ref{eq2.1})
and (\ref{eq2.2}). Because there are only 4 different states in the quantum
result, it is sufficient to take $\xi$ from $1$ to $4$, one has%
\[
\left\{
\begin{array}
[c]{c}%
\tilde{\rho}_{00}^{\hat{n}_{1}}=\wp\left(  00|\hat{n}_{1},1\right)  \wp
_{1}\rho_{1}+\wp\left(  00|\hat{n}_{1},2\right)  \wp_{2}\rho_{2}+\wp\left(
00|\hat{n}_{1},3\right)  \wp_{3}\rho_{3}+\wp\left(  00|\hat{n}_{1},4\right)
\wp_{4}\rho_{4},\\
\tilde{\rho}_{01}^{\hat{n}_{1}}=\wp\left(  01|\hat{n}_{1},1\right)  \wp
_{1}\rho_{1}+\wp\left(  01|\hat{n}_{1},2\right)  \wp_{2}\rho_{2}+\wp\left(
01|\hat{n}_{1},3\right)  \wp_{3}\rho_{3}+\wp\left(  01|\hat{n}_{1},4\right)
\wp_{4}\rho_{4},\\
\tilde{\rho}_{10}^{\hat{n}_{1}}=\wp\left(  10|\hat{n}_{1},1\right)  \wp
_{1}\rho_{1}+\wp\left(  10|\hat{n}_{1},2\right)  \wp_{2}\rho_{2}+\wp\left(
10|\hat{n}_{1},3\right)  \wp_{3}\rho_{3}+\wp\left(  10|\hat{n}_{1},4\right)
\wp_{4}\rho_{4},\\
\tilde{\rho}_{11}^{\hat{n}_{1}}=\wp\left(  11|\hat{n}_{1},1\right)  \wp
_{1}\rho_{1}+\wp\left(  11|\hat{n}_{1},2\right)  \wp_{2}\rho_{2}+\wp\left(
11|\hat{n}_{1},3\right)  \wp_{3}\rho_{3}+\wp\left(  11\hat{n}_{1},4\right)
\wp_{4}\rho_{4},\\
\tilde{\rho}_{00}^{\hat{n}_{2}}=\wp\left(  00|\hat{n}_{2},1\right)  \wp
_{1}\rho_{1}+\wp\left(  00|\hat{n}_{2},2\right)  \wp_{2}\rho_{2}+\wp\left(
00|\hat{n}_{2},3\right)  \wp_{3}\rho_{3}+\wp\left(  00|\hat{n}_{2},4\right)
\wp_{4}\rho_{4},\\
\tilde{\rho}_{01}^{\hat{n}_{2}}=\wp\left(  01|\hat{n}_{2},1\right)  \wp
_{1}\rho_{1}+\wp\left(  01|\hat{n}_{2},2\right)  \wp_{2}\rho_{2}+\wp\left(
01|\hat{n}_{2},3\right)  \wp_{3}\rho_{3}+\wp\left(  01|\hat{n}_{2},4\right)
\wp_{4}\rho_{4},\\
\tilde{\rho}_{10}^{\hat{n}_{2}}=\wp\left(  10|\hat{n}_{2},1\right)  \wp
_{1}\rho_{1}+\wp\left(  10|\hat{n}_{2},2\right)  \wp_{2}\rho_{2}+\wp\left(
10|\hat{n}_{2},3\right)  \wp_{3}\rho_{3}+\wp\left(  10|\hat{n}_{2},4\right)
\wp_{4}\rho_{4},\\
\tilde{\rho}_{11}^{\hat{n}_{2}}=\wp\left(  11|\hat{n}_{2},1\right)  \wp
_{1}\rho_{1}+\wp\left(  11|\hat{n}_{2},2\right)  \wp_{2}\rho_{2}+\wp\left(
11|\hat{n}_{2},3\right)  \wp_{3}\rho_{3}+\wp\left(  11|\hat{n}_{2},4\right)
\wp_{4}\rho_{4}.
\end{array}
\right.
\]
Since the eight states of Eq. (\ref{2.09}) are pure states, and a pure state
cannot be obtained by convex combination of other pure states, one has%
\begin{equation}
\left\{
\begin{array}
[c]{c}%
\tilde{\rho}_{00}^{\hat{n}_{1}}=\wp\left(  00|\hat{n}_{1},1\right)  \wp
_{1}\rho_{1},\\
\tilde{\rho}_{01}^{\hat{n}_{1}}=\wp\left(  01|\hat{n}_{1},1\right)  \wp
_{1}\rho_{1},\\
\tilde{\rho}_{10}^{\hat{n}_{1}}=\wp\left(  10|\hat{n}_{1},2\right)  \wp
_{2}\rho_{2},\\
\tilde{\rho}_{11}^{\hat{n}_{1}}=\wp\left(  11|\hat{n}_{1},2\right)  \wp
_{2}\rho_{2},\\
\tilde{\rho}_{00}^{\hat{n}_{2}}=\wp\left(  00|\hat{n}_{2},3\right)  \wp
_{3}\rho_{3},\\
\tilde{\rho}_{01}^{\hat{n}_{2}}=\wp\left(  01|\hat{n}_{2},4\right)  \wp
_{4}\rho_{4},\\
\tilde{\rho}_{10}^{\hat{n}_{2}}=\wp\left(  10|\hat{n}_{2},4\right)  \wp
_{4}\rho_{4},\\
\tilde{\rho}_{11}^{\hat{n}_{2}}=\wp\left(  11|\hat{n}_{2},3\right)  \wp
_{3}\rho_{3}.
\end{array}
\right.  \label{eq2.10}%
\end{equation}
Since in quantum result $\tilde{\rho}_{00}^{\hat{n}_{1}}=\tilde{\rho}%
_{01}^{\hat{n}_{1}}$, $\tilde{\rho}_{10}^{\hat{n}_{1}}=\tilde{\rho}_{11}%
^{\hat{n}_{1}}$, $\tilde{\rho}_{00}^{\hat{n}_{2}}=\tilde{\rho}_{11}^{\hat
{n}_{2}}$, and $\tilde{\rho}_{01}^{\hat{n}_{2}}=\tilde{\rho}_{10}^{\hat{n}%
_{2}}$, in the LHS description we describe the identical terms by the same
hidden state. And ${\sum\limits_{a}}\wp\left(  a|\hat{n}_{1},\xi\right)  =1$
and ${\sum\limits_{a}}\wp\left(  a|\hat{n}_{2},\xi\right)  =1$, so we have
$\wp\left(  00|\hat{n}_{1},1\right)  +\wp\left(  01|\hat{n}_{1},1\right)  =1$,
$\wp\left(  10|\hat{n}_{1},2\right)  +\wp\left(  11|\hat{n}_{1},2\right)  $,
$\wp\left(  00|\hat{n}_{2},3\right)  +\wp\left(  11|\hat{n}_{2},3\right)  =1$,
and $\wp\left(  01|\hat{n}_{2},4\right)  +\wp\left(  10|\hat{n}_{2},4\right)
=1$. Then we take the sum of Eq. (\ref{eq2.10}) and take the trace, and we get
\textquotedblleft$2_{Q}=1_{C}$\textquotedblright. In the same way, in this
case we get the EPR steering paradox \textquotedblleft$2_{Q}=1_{C}%
$\textquotedblright\ and at this moment $\delta=0$.

\subsection{There are identical states in different sets of measurements}

Suppose that none of the elements in $\left\{  \left\vert \eta_{i}%
\right\rangle \right\}  $ are the same, and that none of the elements in
$\left\{  \left\vert \varepsilon_{j}\right\rangle \right\}  $ are the same,
but that an element $\left\vert \eta_{m}\right\rangle $ in $\left\{
\left\vert \eta_{i}\right\rangle \right\}  $ is the same as an element
$\left\vert \varepsilon_{n}\right\rangle $ in $\left\{  \left\vert
\varepsilon_{j}\right\rangle \right\}  $, i.e., $\left\vert \eta
_{m}\right\rangle =\left\vert \varepsilon_{n}\right\rangle $, where
$m,n=1,2,\cdots,2^{M}$. In this case, Bob's conditional states are%
\begin{equation}
\left\{
\begin{array}
[c]{l}%
\tilde{\rho}_{a_{1}}^{\hat{n}_{1}}=\sum_{\alpha}p_{\alpha}\left\vert
s_{1}^{\left(  \alpha\right)  }\right\vert ^{2}\left\vert \eta_{1}%
\right\rangle \left\langle \eta_{1}\right\vert ,\\
\cdots\\
\tilde{\rho}_{a_{m}}^{\hat{n}_{1}}=\sum_{\alpha}p_{\alpha}\left\vert
s_{m}^{\left(  \alpha\right)  }\right\vert ^{2}\left\vert \eta_{m}%
\right\rangle \left\langle \eta_{m}\right\vert ,\\
\cdots\\
\tilde{\rho}_{a_{2^{M}}}^{\hat{n}_{1}}=\sum_{\alpha}p_{\alpha}\left\vert
s_{2^{M}}^{\left(  \alpha\right)  }\right\vert ^{2}\left\vert \eta_{2^{M}%
}\right\rangle \left\langle \eta_{2^{M}}\right\vert ,\\
\tilde{\rho}_{a_{1}^{\prime}}^{\hat{n}_{2}}=\sum_{\alpha}p_{\alpha}\left\vert
t_{1}^{\left(  \alpha\right)  }\right\vert ^{2}\left\vert \varepsilon
_{1}\right\rangle \left\langle \varepsilon_{1}\right\vert ,\\
\cdots\\
\tilde{\rho}_{a_{n}^{\prime}}^{\hat{n}_{2}}=\sum_{\alpha}p_{\alpha}\left\vert
t_{n}^{\left(  \alpha\right)  }\right\vert ^{2}\left\vert \eta_{m}%
\right\rangle \left\langle \eta_{m}\right\vert ,\\
\cdots\\
\tilde{\rho}_{a_{2^{M}}^{\prime}}^{\hat{n}_{2}}=\sum_{\alpha}p_{\alpha
}\left\vert t_{2^{M}}^{\left(  \alpha\right)  }\right\vert ^{2}\left\vert
\varepsilon_{2^{M}}\right\rangle \left\langle \varepsilon_{2^{M}}\right\vert .
\end{array}
\right.  \label{c3.1}%
\end{equation}
there are $2^{M+1}-1$ different\ pure states in the quantum result Eq.
(\ref{c3.1}). We can take $\xi$ from 1 to $2^{M+1}-1$ in the LHS description,
e.g.,
\begin{equation}
\left\{
\begin{array}
[c]{l}%
\tilde{\rho}_{a_{1}}^{\hat{n}_{1}}=%
%TCIMACRO{\dsum \limits_{\xi=1}^{2^{M+1}-1}}%
%BeginExpansion
{\displaystyle\sum\limits_{\xi=1}^{2^{M+1}-1}}
%EndExpansion
\wp\left(  a_{1}|\hat{n}_{1},\xi\right)  \wp_{\xi}\rho_{\xi},\\
\cdots\\
\tilde{\rho}_{a_{m}}^{\hat{n}_{1}}=%
%TCIMACRO{\dsum \limits_{\xi=1}^{2^{M+1}-1}}%
%BeginExpansion
{\displaystyle\sum\limits_{\xi=1}^{2^{M+1}-1}}
%EndExpansion
\wp\left(  a_{m}|\hat{n}_{1},\xi\right)  \wp_{\xi}\rho_{\xi},\\
\cdots\\
\tilde{\rho}_{a_{2^{M}}}^{\hat{n}_{1}}=%
%TCIMACRO{\dsum \limits_{\xi=1}^{2^{M+1}-1}}%
%BeginExpansion
{\displaystyle\sum\limits_{\xi=1}^{2^{M+1}-1}}
%EndExpansion
\wp\left(  a_{2^{M}}|\hat{n}_{1},\xi\right)  \wp_{\xi}\rho_{\xi},\\
\tilde{\rho}_{a_{1}^{\prime}}^{\hat{n}_{2}}=%
%TCIMACRO{\dsum \limits_{\xi=1}^{2^{M+1}-1}}%
%BeginExpansion
{\displaystyle\sum\limits_{\xi=1}^{2^{M+1}-1}}
%EndExpansion
\wp\left(  a_{1}^{\prime}|\hat{n}_{2},\xi\right)  \wp_{\xi}\rho_{\xi},\\
\cdots\\
\tilde{\rho}_{a_{n}^{\prime}}^{\hat{n}_{2}}=%
%TCIMACRO{\dsum \limits_{\xi=1}^{2^{M+1}-1}}%
%BeginExpansion
{\displaystyle\sum\limits_{\xi=1}^{2^{M+1}-1}}
%EndExpansion
\wp\left(  a_{n}^{\prime}|\hat{n}_{2},\xi\right)  \wp_{\xi}\rho_{\xi},\\
\cdots\\
\tilde{\rho}_{a_{2^{M}}^{\prime}}^{\hat{n}_{2}}=%
%TCIMACRO{\dsum \limits_{\xi=1}^{2^{M+1}-1}}%
%BeginExpansion
{\displaystyle\sum\limits_{\xi=1}^{2^{M+1}-1}}
%EndExpansion
\wp\left(  a_{2^{M}}^{\prime}|\hat{n}_{2},\xi\right)  \wp_{\xi}\rho_{\xi}.
\end{array}
\right.  \label{c3.2}%
\end{equation}
Each $\tilde{\rho}_{a}^{\hat{n}_{\ell}}$ in Eq. (\ref{c3.2}) contains only one
term because Bob's states are all pure. Hence, Eq. (\ref{c3.2}) can be
expressed as:%
\begin{equation}
\left\{
\begin{array}
[c]{l}%
\tilde{\rho}_{a_{1}}^{\hat{n}_{1}}=\wp_{1}\rho_{1},\\
\cdots\\
\tilde{\rho}_{a_{m}}^{\hat{n}_{1}}=\wp\left(  a_{m}|\hat{n}_{1},m\right)
\wp_{m}\rho_{m},\\
\cdots\\
\tilde{\rho}_{a_{2^{M}}}^{\hat{n}_{1}}=\wp_{2^{M}-1}\rho_{2^{M}-1},\\
\tilde{\rho}_{a_{1}^{\prime}}^{\hat{n}_{2}}=\wp_{2^{M}}\rho_{2^{M}},\\
\cdots\\
\tilde{\rho}_{a_{n}^{\prime}}^{\hat{n}_{2}}=\wp\left(  a_{n}^{\prime}|\hat
{n}_{2},m\right)  \wp_{m}\rho_{m},\\
\cdots\\
\tilde{\rho}_{a_{2^{M}}^{\prime}}^{\hat{n}_{2}}=\wp_{2^{M+1}-1}\rho
_{2^{M+1}-1}.
\end{array}
\right.  \label{c3.3}%
\end{equation}
Since we assume that $\left\vert \eta_{m}\right\rangle =\left\vert
\varepsilon_{n}\right\rangle $, we describe $\tilde{\rho}_{a_{m}}^{\hat{n}%
_{1}}$ and $\tilde{\rho}_{a_{n}^{\prime}}^{\hat{n}_{2}}$ by the same hidden
state $\wp_{m}\rho_{m}$. And ${\sum\limits_{a}}\wp\left(  a|\hat{n}%
,\xi\right)  =1$, so we have $\wp\left(  a_{1}|\hat{n}_{1},1\right)
=\cdots=\wp\left(  a_{m-1}|\hat{n}_{1},m-1\right)  =\wp\left(  a_{m+1}|\hat
{n}_{1},m+1\right)  =\cdots=\wp\left(  a_{2^{M}}|\hat{n}_{1},2^{M}-1\right)
=\wp\left(  a_{1}^{\prime}|\hat{n}_{2},2^{M}\right)  \cdots\cdots=\wp\left(
a_{n-1}^{\prime}|\hat{n}_{2},n-1\right)  =\wp\left(  a_{n+1}^{\prime}|\hat
{n}_{2},n\right)  =\cdots=\wp\left(  a_{2^{M}}^{\prime}|\hat{n}_{2}%
,2^{M+1}-1\right)  =1$, and the others $\wp\left(  a|\hat{n},\xi\right)  =0$,
with the peculiarity that in this case $\wp\left(  a_{m}|\hat{n}_{1},m\right)
=\wp\left(  a_{n}^{\prime}|\hat{n}_{2},m\right)  =1$, e.g., Eq. (\ref{c3.2})
should be written as%
\begin{equation}
\left\{
\begin{array}
[c]{l}%
\tilde{\rho}_{a_{1}}^{\hat{n}_{1}}=\wp_{1}\rho_{1},\\
\cdots\\
\tilde{\rho}_{a_{m}}^{\hat{n}_{1}}=\wp_{m}\rho_{m},\\
\cdots\\
\tilde{\rho}_{a_{2^{M}}}^{\hat{n}_{1}}=\wp_{2^{M}-1}\rho_{2^{M}-1},\\
\tilde{\rho}_{a_{1}^{\prime}}^{\hat{n}_{2}}=\wp_{2^{M}}\rho_{2^{M}},\\
\cdots\\
\tilde{\rho}_{a_{n}^{\prime}}^{\hat{n}_{2}}=\wp_{m}\rho_{m},\\
\cdots\\
\tilde{\rho}_{a_{2^{M}}^{\prime}}^{\hat{n}_{2}}=\wp_{2^{M+1}-1}\rho
_{2^{M+1}-1}.
\end{array}
\right.  \label{c3.4}%
\end{equation}
We take the sum of Eq. (\ref{c3.4}), the quantum result on the left gives
$\tilde{\rho}_{a_{1}}^{\hat{n}_{1}}+$ $\tilde{\rho}_{a_{2}}^{\hat{n}_{1}%
}+\cdots+\tilde{\rho}_{a_{2^{M}}}^{\hat{n}_{1}}+\tilde{\rho}_{a_{1}^{\prime}%
}^{\hat{n}_{2}}+\tilde{\rho}_{a_{2}^{\prime}}^{\hat{n}_{2}}+\cdots+\tilde
{\rho}_{a_{2^{M}}^{\prime}}^{\hat{n}_{2}}=2\rho_{B}$, and the classical result
on the right gives $\wp_{1}\rho_{1}+\wp_{2}\rho_{2}+\cdots+\wp_{m}\rho
_{m}+\cdots+\wp_{m}\rho_{m}+\cdots+\wp_{2^{M+1}-1}\rho_{2^{M+1}-1}=\rho
_{B}+\wp_{m}\rho_{m}$, since $%
%TCIMACRO{\dsum \limits_{\xi=1}^{2^{M}}}%
%BeginExpansion
{\displaystyle\sum\limits_{\xi=1}^{2^{M}}}
%EndExpansion
\wp_{\xi}\rho_{\xi}=\rho_{B}$. We then take the trace to get the result
\textquotedblleft$2_{Q}=\left(  1+\wp_{m}\right)  _{C}$\textquotedblright.
Since ${\sum\limits_{\xi}}\wp_{\xi}=1$, and $\wp_{1}$, $\cdots$, $\wp_{m-1}$,
$\wp_{m+1}$, $\cdots$, $\wp_{2^{M+1}-1}\neq0$, $0<\wp_{m}<1$. In this case, we
get the EPR steering paradox \textquotedblleft$2_{Q}=\left(  1+\delta\right)
_{C}$\textquotedblright\ and $\delta=\wp_{m}$, with $0<\delta<1$.

\subsubsection{Example 3}

We consider the example of a 3-qubit pure entangled state $\rho_{AB}%
=\left\vert \Psi^{\prime}\right\rangle \left\langle \Psi^{\prime}\right\vert
$, where
\begin{equation}
\left\vert \Psi^{\prime}\right\rangle =\dfrac{1}{\sqrt{6}}\left[  \left\vert
001\right\rangle +\left\vert 01\right\rangle \left(  \left\vert 0\right\rangle
+\left\vert 1\right\rangle \right)  +\left\vert 10\right\rangle \left(
\left\vert 0\right\rangle -\left\vert 1\right\rangle \right)  -\left\vert
110\right\rangle \right]  ,\label{eq3.1}%
\end{equation}
or%
\begin{align*}
\left\vert \Psi^{\prime}\right\rangle  &  =\dfrac{1}{2\sqrt{6}}\left(
\left\vert ++\right\rangle +\left\vert +-\right\rangle +\left\vert
-+\right\rangle +\left\vert --\right\rangle \right)  \left\vert 1\right\rangle
\\
&  +\dfrac{1}{2\sqrt{6}}\left(  \left\vert ++\right\rangle -\left\vert
+-\right\rangle +\left\vert -+\right\rangle -\left\vert --\right\rangle
\right)  \left(  \left\vert 0\right\rangle +\left\vert 1\right\rangle \right)
\\
&  +\dfrac{1}{2\sqrt{6}}\left(  \left\vert ++\right\rangle +\left\vert
+-\right\rangle -\left\vert -+\right\rangle -\left\vert --\right\rangle
\right)  \left(  \left\vert 0\right\rangle -\left\vert 1\right\rangle \right)
\\
&  -\dfrac{1}{2\sqrt{6}}\left(  \left\vert ++\right\rangle -\left\vert
+-\right\rangle -\left\vert -+\right\rangle +\left\vert --\right\rangle
\right)  \left\vert 0\right\rangle \\
&  =\dfrac{1}{2\sqrt{6}}\left[  \left\vert ++\right\rangle \left(  \left\vert
0\right\rangle +\left\vert 1\right\rangle \right)  +\left\vert +-\right\rangle
\left(  \left\vert 0\right\rangle -\left\vert 1\right\rangle \right)
+\left\vert -+\right\rangle \left(  \left\vert 0\right\rangle +3\left\vert
1\right\rangle \right)  -\left\vert --\right\rangle \left(  3\left\vert
0\right\rangle -\left\vert 1\right\rangle \right)  \right]  .
\end{align*}
Alice prepares the state $\rho_{AB}$ as in Eq. (\ref{eq3.1}). She keeps
particles 1, 2, and sends particle 3 to Bob. In the 2-setting steering
protocol $\left\{  \hat{n}_{1},\hat{n}_{2}\right\}  $ $\left(  \hat{n}_{1}%
\neq\hat{n}_{2}\right)  $, with%
\begin{equation}%
\begin{array}
[c]{c}%
\hat{n}_{1}=\sigma_{z}\sigma_{z}\equiv zz,\\
\hat{n}_{2}=\sigma_{x}\sigma_{x}\equiv xx,
\end{array}
\end{equation}
i.e.,%
\begin{equation}
\left\{
\begin{array}
[c]{l}%
P_{00}^{\hat{n}_{1}}=\left\vert 00\right\rangle \left\langle 00\right\vert ,\\
P_{01}^{\hat{n}_{1}}=\left\vert 01\right\rangle \left\langle 01\right\vert ,\\
P_{10}^{\hat{n}_{1}}=\left\vert 10\right\rangle \left\langle 10\right\vert ,\\
P_{11}^{\hat{n}_{1}}=\left\vert 11\right\rangle \left\langle 11\right\vert ,\\
P_{00}^{\hat{n}_{2}}=\left\vert ++\right\rangle \left\langle ++\right\vert ,\\
P_{01}^{\hat{n}_{2}}=\left\vert +-\right\rangle \left\langle +-\right\vert ,\\
P_{10}^{\hat{n}_{2}}=\left\vert -+\right\rangle \left\langle -+\right\vert ,\\
P_{11}^{\hat{n}_{2}}=\left\vert --\right\rangle \left\langle --\right\vert ,
\end{array}
\right.
\end{equation}
After Alice's measurement, Bob's unnormalized conditional states are
\begin{equation}
\left\{
\begin{array}
[c]{l}%
\tilde{\rho}_{00}^{\hat{n}_{1}}=\dfrac{1}{6}\left\vert 1\right\rangle
\left\langle 1\right\vert ,\\[9pt]%
\tilde{\rho}_{01}^{\hat{n}_{1}}=\dfrac{1}{6}\left(  \left\vert 0\right\rangle
+\left\vert 1\right\rangle \right)  \left(  \left\langle 0\right\vert
+\left\langle 1\right\vert \right)  ,\\[9pt]%
\tilde{\rho}_{10}^{\hat{n}_{1}}=\dfrac{1}{6}\left(  \left\vert 0\right\rangle
-\left\vert 1\right\rangle \right)  \left(  \left\langle 0\right\vert
-\left\langle 1\right\vert \right)  ,\\[9pt]%
\tilde{\rho}_{11}^{\hat{n}_{1}}=\dfrac{1}{6}\left\vert 0\right\rangle
\left\langle 0\right\vert ,\\[9pt]%
\tilde{\rho}_{00}^{\hat{n}_{2}}=\dfrac{1}{24}\left(  \left\vert 0\right\rangle
+\left\vert 1\right\rangle \right)  \left(  \left\langle 0\right\vert
+\left\langle 1\right\vert \right)  ,\\[9pt]%
\tilde{\rho}_{01}^{\hat{n}_{2}}=\dfrac{1}{24}\left(  \left\vert 0\right\rangle
-\left\vert 1\right\rangle \right)  \left(  \left\langle 0\right\vert
-\left\langle 1\right\vert \right)  ,\\[9pt]%
\tilde{\rho}_{10}^{\hat{n}_{2}}=\dfrac{1}{24}\left(  \left\vert 0\right\rangle
+3\left\vert 1\right\rangle \right)  \left(  \left\langle 0\right\vert
+3\left\langle 1\right\vert \right)  ,\\[9pt]%
\tilde{\rho}_{11}^{\hat{n}_{2}}=\dfrac{1}{24}\left(  3\left\vert
0\right\rangle -\left\vert 1\right\rangle \right)  \left(  3\left\langle
0\right\vert -\left\langle 1\right\vert \right)  ,
\end{array}
\right.  \label{eq3.2}%
\end{equation}
where,%

\begin{align*}
\tilde{\rho}_{00}^{\hat{n}_{1}} &  =\mathrm{tr}_{A}\left[  \left(
P_{00}^{\hat{n}_{1}}{\otimes\mathds{1}}\right)  \rho_{AB}\right]  \\
&  =\mathrm{tr}_{A}\left[  \left(  \left\vert 00\right\rangle \left\langle
00\right\vert {\otimes\mathds{1}}\right)  \left\vert \Psi^{\prime
}\right\rangle \left\langle \Psi^{\prime}\right\vert \right]  \\
&  =\mathrm{tr}_{A}\left[  \dfrac{1}{\sqrt{6}}\left(  \left\vert
00\right\rangle _{A}{\otimes\mathds{1}}\right)  \left\vert 1\right\rangle
\left\langle \Psi^{\prime}\right\vert \right]  \\
&  =\dfrac{1}{\sqrt{6}}\left\langle 00\right\vert _{A}{\otimes\mathds{1}}%
\left(  \left\vert 00\right\rangle _{A}{\otimes\mathds{1}}\right)  \left\vert
1\right\rangle \left\langle \Psi^{\prime}\right\vert \left(  \left\vert
00\right\rangle _{A}{\otimes\mathds{1}}\right)  \\
&  =\dfrac{1}{6}\left\langle 00\right\vert _{A}{\otimes\mathds{1}}\left(
\left\vert 00\right\rangle _{A}{\otimes\mathds{1}}\right)  \left\vert
1\right\rangle \left[  \left\langle 00\right\vert \left\langle 1\right\vert
+\left\langle 01\right\vert \left(  \left\langle 0\right\vert +\left\langle
1\right\vert \right)  +\left\langle 10\right\vert \left(  \left\langle
0\right\vert -\left\langle 1\right\vert \right)  -\left\langle 11\right\vert
\left\langle 0\right\vert \right]  \left(  \left\vert 00\right\rangle
_{A}{\otimes\mathds{1}}\right)  \\
&  =\dfrac{1}{6}\left\vert 1\right\rangle \left\langle 1\right\vert ,
\end{align*}

\begin{align*}
\tilde{\rho}_{01}^{\hat{n}_{1}} &  =\mathrm{tr}_{A}\left[  \left(
P_{01}^{\hat{n}_{1}}{\otimes\mathds{1}}\right)  \rho_{AB}\right]  \\
&  =\mathrm{tr}_{A}\left[  \left(  \left\vert 01\right\rangle \left\langle
01\right\vert {\otimes\mathds{1}}\right)  \left\vert \Psi^{\prime
}\right\rangle \left\langle \Psi^{\prime}\right\vert \right]  \\
&  =\mathrm{tr}_{A}\left[  \dfrac{1}{\sqrt{6}}\left(  \left\vert
01\right\rangle _{A}{\otimes\mathds{1}}\right)  \left(  \left\vert
0\right\rangle +\left\vert 1\right\rangle \right)  \left\langle \Psi^{\prime
}\right\vert \right]  \\
&  =\dfrac{1}{\sqrt{6}}\left(  \left\langle 01\right\vert _{A}{\otimes
\mathds{1}}\right)  \left(  \left\vert 01\right\rangle _{A}{\otimes
\mathds{1}}\right)  \left(  \left\vert 0\right\rangle +\left\vert
1\right\rangle \right)  \left\langle \Psi^{\prime}\right\vert \left(
\left\vert 01\right\rangle _{A}{\otimes\mathds{1}}\right)  \\
&  =\dfrac{1}{6}\left(  \left\langle 01\right\vert _{A}{\otimes\mathds{1}}%
\right)  \left(  \left\vert 01\right\rangle _{A}{\otimes\mathds{1}}\right)
\left(  \left\vert 0\right\rangle +\left\vert 1\right\rangle \right)  \left[
\left\langle 00\right\vert \left\langle 1\right\vert +\left\langle
01\right\vert \left(  \left\langle 0\right\vert +\left\langle 1\right\vert
\right)  \right.  \\
&  \left.  +\left\langle 10\right\vert \left(  \left\langle 0\right\vert
-\left\langle 1\right\vert \right)  -\left\langle 11\right\vert \left\langle
0\right\vert \right]  \left(  \left\vert 01\right\rangle _{A}{\otimes
\mathds{1}}\right)  \\
&  =\dfrac{1}{6}\left(  \left\vert 0\right\rangle +\left\vert 1\right\rangle
\right)  \left(  \left\langle 0\right\vert +\left\langle 1\right\vert \right)
,
\end{align*}

\begin{align*}
\tilde{\rho}_{10}^{\hat{n}_{1}} &  =\mathrm{tr}_{A}\left[  \left(
P_{10}^{\hat{n}_{1}}{\otimes\mathds{1}}\right)  \rho_{AB}\right]  \\
&  =\mathrm{tr}_{A}\left[  \left(  \left\vert 10\right\rangle \left\langle
10\right\vert {\otimes\mathds{1}}\right)  \left\vert \Psi^{\prime
}\right\rangle \left\langle \Psi^{\prime}\right\vert \right]  \\
&  =\mathrm{tr}_{A}\left[  \dfrac{1}{\sqrt{6}}\left(  \left\vert
10\right\rangle _{A}{\otimes\mathds{1}}\right)  \left(  \left\vert
0\right\rangle -\left\vert 1\right\rangle \right)  \left\langle \Psi^{\prime
}\right\vert \right]  \\
&  =\dfrac{1}{\sqrt{6}}\left(  \left\langle 10\right\vert _{A}{\otimes
\mathds{1}}\right)  \left(  \left\vert 10\right\rangle _{A}{\otimes
\mathds{1}}\right)  \left(  \left\vert 0\right\rangle -\left\vert
1\right\rangle \right)  \left(  \left\langle \Psi^{\prime}\right\vert
\left\vert 10\right\rangle _{A}{\otimes\mathds{1}}\right)  \\
&  =\dfrac{1}{6}\left(  \left\langle 10\right\vert _{A}{\otimes\mathds{1}}%
\right)  \left(  \left\vert 10\right\rangle _{A}{\otimes\mathds{1}}\right)
\left(  \left\vert 0\right\rangle -\left\vert 1\right\rangle \right)  \left[
\left\langle 00\right\vert \left\langle 1\right\vert +\left\langle
01\right\vert \left(  \left\langle 0\right\vert +\left\langle 1\right\vert
\right)  \right.  \\
&  \left.  +\left\langle 10\right\vert \left(  \left\langle 0\right\vert
-\left\langle 1\right\vert \right)  -\left\langle 11\right\vert \left\langle
0\right\vert \right]  \left(  \left\vert 10\right\rangle _{A}{\otimes
\mathds{1}}\right)  \\
&  =\dfrac{1}{6}\left(  \left\vert 0\right\rangle -\left\vert 1\right\rangle
\right)  \left(  \left\langle 0\right\vert -\left\langle 1\right\vert \right)
,
\end{align*}

\begin{align*}
\tilde{\rho}_{11}^{\hat{n}_{1}} &  =\mathrm{tr}_{A}\left[  \left(
P_{11}^{\hat{n}_{1}}{\otimes\mathds{1}}\right)  \rho_{AB}\right]  \\
&  =\mathrm{tr}_{A}\left[  \left(  \left\vert 11\right\rangle \left\langle
11\right\vert {\otimes\mathds{1}}\right)  \left\vert \Psi^{\prime
}\right\rangle \left\langle \Psi^{\prime}\right\vert \right]  \\
&  =\mathrm{tr}_{A}\left[  -\dfrac{1}{\sqrt{6}}\left(  \left\vert
11\right\rangle _{A}{\otimes\mathds{1}}\right)  \left\vert 0\right\rangle
\left\langle \Psi^{\prime}\right\vert \right]  \\
&  =-\dfrac{1}{\sqrt{6}}\left(  \left\langle 11\right\vert _{A}{\otimes
\mathds{1}}\right)  \left(  \left\vert 11\right\rangle _{A}{\otimes
\mathds{1}}\right)  \left\vert 0\right\rangle \left\langle \Psi^{\prime
}\right\vert \left(  \left\vert 11\right\rangle _{A}{\otimes\mathds{1}}%
\right)  \\
&  =-\dfrac{1}{6}\left(  \left\langle 11\right\vert _{A}{\otimes
\mathds{1}}\right)  \left(  \left\vert 11\right\rangle _{A}{\otimes
\mathds{1}}\right)  \left\vert 0\right\rangle \left[  \left\langle
00\right\vert \left\langle 1\right\vert +\left\langle 01\right\vert \left(
\left\langle 0\right\vert +\left\langle 1\right\vert \right)  +\left\langle
10\right\vert \left(  \left\langle 0\right\vert -\left\langle 1\right\vert
\right)  -\left\langle 11\right\vert \left\langle 0\right\vert \right]
\left(  \left\vert 11\right\rangle _{A}{\otimes\mathds{1}}\right)  \\
&  =\dfrac{1}{6}\left\vert 0\right\rangle \left\langle 0\right\vert ,
\end{align*}

\begin{align*}
\tilde{\rho}_{00}^{\hat{n}_{2}} &  =\mathrm{tr}_{A}\left[  \left(
P_{00}^{\hat{n}_{2}}{\otimes\mathds{1}}\right)  \rho_{AB}\right]  \\
&  =\mathrm{tr}_{A}\left[  \left(  \left\vert ++\right\rangle \left\langle
++\right\vert {\otimes\mathds{1}}\right)  \left\vert \Psi^{\prime
}\right\rangle \left\langle \Psi^{\prime}\right\vert \right]  \\
&  =\mathrm{tr}_{A}\left[  \dfrac{1}{2\sqrt{6}}\left(  \left\vert
++\right\rangle _{A}{\otimes\mathds{1}}\right)  \left(  \left\vert
0\right\rangle +\left\vert 1\right\rangle \right)  \left\langle \Psi^{\prime
}\right\vert \right]  \\
&  =\dfrac{1}{2\sqrt{6}}\left(  \left\langle ++\right\vert _{A}{\otimes
\mathds{1}}\right)  \left(  \left\vert ++\right\rangle _{A}{\otimes
\mathds{1}}\right)  \left(  \left\vert 0\right\rangle +\left\vert
1\right\rangle \right)  \left\langle \Psi^{\prime}\right\vert \left(
\left\vert ++\right\rangle _{A}{\otimes\mathds{1}}\right)  \\
&  =\dfrac{1}{24}\left(  \left\langle ++\right\vert _{A}{\otimes
\mathds{1}}\right)  \left(  \left\vert ++\right\rangle _{A}{\otimes
\mathds{1}}\right)  \left(  \left\vert 0\right\rangle +\left\vert
1\right\rangle \right)  \left[  \left\langle ++\right\vert \left(
\left\langle 0\right\vert +\left\langle 1\right\vert \right)  +\left\langle
+-\right\vert \left(  \left\langle 0\right\vert -\left\langle 1\right\vert
\right)  \right.  \\
&  \left.  +\left\langle -+\right\vert \left(  \left\langle 0\right\vert
+3\left\langle 1\right\vert \right)  -\left\langle --\right\vert \left(
3\left\langle 0\right\vert -\left\langle 1\right\vert \right)  \right]
\left(  \left\vert ++\right\rangle _{A}{\otimes\mathds{1}}\right)  \\
&  =\dfrac{1}{24}\left(  \left\vert 0\right\rangle +\left\vert 1\right\rangle
\right)  \left(  \left\langle 0\right\vert +\left\langle 1\right\vert \right)
,
\end{align*}

\begin{align*}
\tilde{\rho}_{01}^{\hat{n}_{2}} &  =\mathrm{tr}_{A}\left[  \left(
P_{01}^{\hat{n}_{2}}{\otimes\mathds{1}}\right)  \rho_{AB}\right]  \\
&  =\mathrm{tr}_{A}\left[  \left(  \left\vert +-\right\rangle \left\langle
+-\right\vert {\otimes\mathds{1}}\right)  \left\vert \Psi^{\prime
}\right\rangle \left\langle \Psi^{\prime}\right\vert \right]  \\
&  =\mathrm{tr}_{A}\left[  \dfrac{1}{2\sqrt{6}}\left(  \left\vert
+-\right\rangle _{A}{\otimes\mathds{1}}\right)  \left(  \left\vert
0\right\rangle -\left\vert 1\right\rangle \right)  \left\langle \Psi^{\prime
}\right\vert \right]  \\
&  =\dfrac{1}{2\sqrt{6}}\left(  \left\langle +-\right\vert _{A}{\otimes
\mathds{1}}\right)  \left(  \left\vert +-\right\rangle _{A}{\otimes
\mathds{1}}\right)  \left(  \left\vert 0\right\rangle -\left\vert
1\right\rangle \right)  \left\langle \Psi^{\prime}\right\vert \left(
\left\vert +-\right\rangle _{A}{\otimes\mathds{1}}\right)  \\
&  =\dfrac{1}{24}\left(  \left\langle +-\right\vert _{A}{\otimes
\mathds{1}}\right)  \left(  \left\vert +-\right\rangle _{A}{\otimes
\mathds{1}}\right)  \left(  \left\vert 0\right\rangle -\left\vert
1\right\rangle \right)  \left[  \left\langle ++\right\vert \left(
\left\langle 0\right\vert +\left\langle 1\right\vert \right)  +\left\langle
+-\right\vert \left(  \left\langle 0\right\vert -\left\langle 1\right\vert
\right)  \right.  \\
&  \left.  +\left\langle -+\right\vert \left(  \left\langle 0\right\vert
+3\left\langle 1\right\vert \right)  -\left\langle --\right\vert \left(
3\left\langle 0\right\vert -\left\langle 1\right\vert \right)  \right]
\left(  \left\vert +-\right\rangle _{A}{\otimes\mathds{1}}\right)  \\
&  =\dfrac{1}{24}\left(  \left\vert 0\right\rangle -\left\vert 1\right\rangle
\right)  \left(  \left\langle 0\right\vert -\left\langle 1\right\vert \right)
,
\end{align*}

\begin{align*}
\tilde{\rho}_{10}^{\hat{n}_{2}} &  =\mathrm{tr}_{A}\left[  \left(
P_{10}^{\hat{n}_{2}}{\otimes\mathds{1}}\right)  \rho_{AB}\right]  \\
&  =\mathrm{tr}_{A}\left[  \left(  \left\vert -+\right\rangle \left\langle
-+\right\vert {\otimes\mathds{1}}\right)  \left\vert \Psi^{\prime
}\right\rangle \left\langle \Psi^{\prime}\right\vert \right]  \\
&  =\mathrm{tr}_{A}\left[  \dfrac{1}{2\sqrt{6}}\left(  \left\vert
-+\right\rangle _{A}{\otimes\mathds{1}}\right)  \left(  \left\vert
0\right\rangle +3\left\vert 1\right\rangle \right)  \left\langle \Psi^{\prime
}\right\vert \right]  \\
&  =\dfrac{1}{2\sqrt{6}}\left(  \left\langle -+\right\vert _{A}{\otimes
\mathds{1}}\right)  \left(  \left\vert -+\right\rangle _{A}{\otimes
\mathds{1}}\right)  \left(  \left\vert 0\right\rangle +3\left\vert
1\right\rangle \right)  \left\langle \Psi^{\prime}\right\vert \left(
\left\vert -+\right\rangle _{A}{\otimes\mathds{1}}\right)  \\
&  =\dfrac{1}{24}\left(  \left\langle -+\right\vert _{A}{\otimes
\mathds{1}}\right)  \left(  \left\vert -+\right\rangle _{A}{\otimes
\mathds{1}}\right)  \left(  \left\vert 0\right\rangle +3\left\vert
1\right\rangle \right)  \left[  \left\langle ++\right\vert \left(
\left\langle 0\right\vert +\left\langle 1\right\vert \right)  +\left\langle
+-\right\vert \left(  \left\langle 0\right\vert -\left\langle 1\right\vert
\right)  \right.  \\
&  \left.  +\left\langle -+\right\vert \left(  \left\langle 0\right\vert
+3\left\langle 1\right\vert \right)  -\left\langle --\right\vert \left(
3\left\langle 0\right\vert -\left\langle 1\right\vert \right)  \right]
\left(  \left\vert -+\right\rangle _{A}{\otimes\mathds{1}}\right)  \\
&  =\dfrac{1}{24}\left(  \left\vert 0\right\rangle +3\left\vert 1\right\rangle
\right)  \left(  \left\langle 0\right\vert +3\left\langle 1\right\vert
\right)  ,
\end{align*}

\begin{align*}
\tilde{\rho}_{11}^{\hat{n}_{2}} &  =\mathrm{tr}_{A}\left[  \left(
P_{11}^{\hat{n}_{2}}{\otimes\mathds{1}}\right)  \rho_{AB}\right]  \\
&  =\mathrm{tr}_{A}\left[  \left(  \left\vert --\right\rangle \left\langle
--\right\vert {\otimes\mathds{1}}\right)  \left\vert \Psi^{\prime
}\right\rangle \left\langle \Psi^{\prime}\right\vert \right]  \\
&  =\mathrm{tr}_{A}\left[  -\dfrac{1}{2\sqrt{6}}\left(  \left\vert
--\right\rangle _{A}{\otimes\mathds{1}}\right)  \left(  3\left\vert
0\right\rangle -\left\vert 1\right\rangle \right)  \left\langle \Psi^{\prime
}\right\vert \right]  \\
&  =-\dfrac{1}{2\sqrt{6}}\left(  \left\langle --\right\vert _{A}%
{\otimes\mathds{1}}\right)  \left(  \left\vert --\right\rangle _{A}%
{\otimes\mathds{1}}\right)  \left(  3\left\vert 0\right\rangle -\left\vert
1\right\rangle \right)  \left\langle \Psi^{\prime}\right\vert \left(
\left\vert --\right\rangle _{A}{\otimes\mathds{1}}\right)  \\
&  =-\dfrac{1}{24}\left(  \left\langle --\right\vert _{A}{\otimes
\mathds{1}}\right)  \left(  \left\vert --\right\rangle _{A}{\otimes
\mathds{1}}\right)  \left(  3\left\vert 0\right\rangle -\left\vert
1\right\rangle \right)  \left[  \left\langle ++\right\vert \left(
\left\langle 0\right\vert +\left\langle 1\right\vert \right)  +\left\langle
+-\right\vert \left(  \left\langle 0\right\vert -\left\langle 1\right\vert
\right)  \right.  \\
&  \left.  -\left\langle -+\right\vert \left(  \left\langle 0\right\vert
-3\left\langle 1\right\vert \right)  -\left\langle --\right\vert \left(
3\left\langle 0\right\vert -\left\langle 1\right\vert \right)  \right]
\left(  \left\vert --\right\rangle _{A}{\otimes\mathds{1}}\right)  \\
&  =\dfrac{1}{24}\left(  3\left\vert 0\right\rangle -\left\vert 1\right\rangle
\right)  \left(  3\left\langle 0\right\vert -\left\langle 1\right\vert
\right)  .
\end{align*}
In the quantum result, all of Bob's states are pure. While there are identical
results in different sets of measurements, Bob's two sets of results are not
identical, satisfying the theorem's requirements.

Suppose Bob's states have a LHS description, they must satisfy Eqs.
(\ref{eq2.1}) and (\ref{eq2.2}). Because there are only 6 different states in
the quantum result, it is sufficient to take $\xi$ from $1$ to $6$, one has%
\[
\left\{
\begin{array}
[c]{c}%
\tilde{\rho}_{00}^{\hat{n}_{1}}=\wp\left(  00|\hat{n}_{1},1\right)  \wp
_{1}\rho_{1}+\wp\left(  00|\hat{n}_{1},2\right)  \wp_{2}\rho_{2}+\wp\left(
00|\hat{n}_{1},3\right)  \wp_{3}\rho_{3}+\wp\left(  00|\hat{n}_{1},4\right)
\wp_{4}\rho_{4}+\wp\left(  00|\hat{n}_{1},5\right)  \wp_{5}\rho_{5}+\wp\left(
00|\hat{n}_{1},6\right)  \wp_{6}\rho_{6},\\
\tilde{\rho}_{01}^{\hat{n}_{1}}=\wp\left(  01|\hat{n}_{1},1\right)  \wp
_{1}\rho_{1}+\wp\left(  01|\hat{n}_{1},2\right)  \wp_{2}\rho_{2}+\wp\left(
01|\hat{n}_{1},3\right)  \wp_{3}\rho_{3}+\wp\left(  01|\hat{n}_{1},4\right)
\wp_{4}\rho_{4}+\wp\left(  01|\hat{n}_{1},5\right)  \wp_{5}\rho_{5}+\wp\left(
01|\hat{n}_{1},6\right)  \wp_{6}\rho_{6},\\
\tilde{\rho}_{10}^{\hat{n}_{1}}=\wp\left(  10|\hat{n}_{1},1\right)  \wp
_{1}\rho_{1}+\wp\left(  10|\hat{n}_{1},2\right)  \wp_{2}\rho_{2}+\wp\left(
10|\hat{n}_{1},3\right)  \wp_{3}\rho_{3}+\wp\left(  10|\hat{n}_{1},4\right)
\wp_{4}\rho_{4}+\wp\left(  10|\hat{n}_{1},5\right)  \wp_{5}\rho_{5}+\wp\left(
10|\hat{n}_{1},6\right)  \wp_{6}\rho_{6},\\
\tilde{\rho}_{11}^{\hat{n}_{1}}=\wp\left(  11|\hat{n}_{1},1\right)  \wp
_{1}\rho_{1}+\wp\left(  11|\hat{n}_{1},2\right)  \wp_{2}\rho_{2}+\wp\left(
11|\hat{n}_{1},3\right)  \wp_{3}\rho_{3}+\wp\left(  11\hat{n}_{1},4\right)
\wp_{4}\rho_{4}+\wp\left(  11|\hat{n}_{1},5\right)  \wp_{5}\rho_{5}+\wp\left(
11|\hat{n}_{1},6\right)  \wp_{6}\rho_{6},\\
\tilde{\rho}_{00}^{\hat{n}_{2}}=\wp\left(  00|\hat{n}_{2},1\right)  \wp
_{1}\rho_{1}+\wp\left(  00|\hat{n}_{2},2\right)  \wp_{2}\rho_{2}+\wp\left(
00|\hat{n}_{2},3\right)  \wp_{3}\rho_{3}+\wp\left(  00|\hat{n}_{2},4\right)
\wp_{4}\rho_{4}+\wp\left(  00|\hat{n}_{2},5\right)  \wp_{5}\rho_{5}+\wp\left(
00|\hat{n}_{2},6\right)  \wp_{6}\rho_{6},\\
\tilde{\rho}_{01}^{\hat{n}_{2}}=\wp\left(  01|\hat{n}_{2},1\right)  \wp
_{1}\rho_{1}+\wp\left(  01|\hat{n}_{2},2\right)  \wp_{2}\rho_{2}+\wp\left(
01|\hat{n}_{2},3\right)  \wp_{3}\rho_{3}+\wp\left(  01|\hat{n}_{2},4\right)
\wp_{4}\rho_{4}+\wp\left(  01|\hat{n}_{2},5\right)  \wp_{5}\rho5+\wp\left(
01|\hat{n}_{2},6\right)  \wp_{6}\rho_{6},\\
\tilde{\rho}_{10}^{\hat{n}_{2}}=\wp\left(  10|\hat{n}_{2},1\right)  \wp
_{1}\rho_{1}+\wp\left(  10|\hat{n}_{2},2\right)  \wp_{2}\rho_{2}+\wp\left(
10|\hat{n}_{2},3\right)  \wp_{3}\rho_{3}+\wp\left(  10|\hat{n}_{2},4\right)
\wp_{4}\rho_{4}+\wp\left(  10|\hat{n}_{2},5\right)  \wp_{5}\rho_{5}+\wp\left(
10|\hat{n}_{2},6\right)  \wp_{6}\rho_{6},\\
\tilde{\rho}_{11}^{\hat{n}_{2}}=\wp\left(  11|\hat{n}_{2},1\right)  \wp
_{1}\rho_{1}+\wp\left(  11|\hat{n}_{2},2\right)  \wp_{2}\rho_{2}+\wp\left(
11|\hat{n}_{2},3\right)  \wp_{3}\rho_{3}+\wp\left(  11|\hat{n}_{2},4\right)
\wp_{4}\rho_{4}+\wp\left(  11|\hat{n}_{2},5\right)  \wp_{5}\rho_{5}+\wp\left(
11|\hat{n}_{2},6\right)  \wp_{6}\rho_{6}.
\end{array}
\right.
\]
Since the 8 states of Eq. (\ref{eq3.2}) are pure states, and a pure state
cannot be obtained by convex combination of other pure states, one has%
\[
\left\{
\begin{array}
[c]{l}%
\tilde{\rho}_{00}^{\hat{n}_{1}}=\wp\left(  00|\hat{n}_{1},1\right)  \wp
_{1}\rho_{1},\\
\tilde{\rho}_{01}^{\hat{n}_{1}}=\wp\left(  01|\hat{n}_{1},2\right)  \wp
_{2}\rho_{2},\\
\tilde{\rho}_{10}^{\hat{n}_{1}}=\wp\left(  10|\hat{n}_{1},3\right)  \wp
_{3}\rho_{3},\\
\tilde{\rho}_{11}^{\hat{n}_{1}}=\wp\left(  11\hat{n}_{1},4\right)  \wp_{4}%
\rho_{4},\\
\tilde{\rho}_{00}^{\hat{n}_{2}}=\wp\left(  00|\hat{n}_{2},2\right)  \wp
_{2}\rho_{2},\\
\tilde{\rho}_{01}^{\hat{n}_{2}}=\wp\left(  01|\hat{n}_{2},3\right)  \wp
_{3}\rho_{3},\\
\tilde{\rho}_{10}^{\hat{n}_{2}}=\wp\left(  10|\hat{n}_{2},5\right)  \wp
_{5}\rho_{5},\\
\tilde{\rho}_{11}^{\hat{n}_{2}}=\wp\left(  11|\hat{n}_{2},6\right)  \wp
_{6}\rho_{6},
\end{array}
\right.
\]
where the others $\wp\left(  a|\hat{n},\xi\right)  =0$, and we describe the
identical terms by the same hidden state. Because, for $\hat{n}_{1}$ and
$\hat{n}_{2}$, ${\sum\limits_{a}}\wp\left(  a|\hat{n}_{1},\xi\right)  =1$ and
${\sum\limits_{a}}\wp\left(  a|\hat{n}_{2},\xi\right)  =1$, we have
$\wp\left(  00|\hat{n}_{1},1\right)  =\wp\left(  01|\hat{n}_{1},2\right)
=\wp\left(  10|\hat{n}_{1},3\right)  =\wp\left(  11\hat{n}_{1},4\right)
=\wp\left(  00|\hat{n}_{2},2\right)  =\wp\left(  01|\hat{n}_{2},3\right)
=\wp\left(  10|\hat{n}_{2},5\right)  =\wp\left(  11|\hat{n}_{2},6\right)  =1$.
Then we have%

\begin{equation}
\left\{
\begin{array}
[c]{c}%
\tilde{\rho}_{00}^{\hat{n}_{1}}=\wp_{1}\rho_{1},\\
\tilde{\rho}_{01}^{\hat{n}_{1}}=\wp_{2}\rho_{2},\\
\tilde{\rho}_{10}^{\hat{n}_{1}}=\wp_{3}\rho_{3},\\
\tilde{\rho}_{11}^{\hat{n}_{1}}=\wp_{4}\rho_{4},\\
\tilde{\rho}_{00}^{\hat{n}_{2}}=\wp_{2}\rho_{2},\\
\tilde{\rho}_{01}^{\hat{n}_{2}}=\wp_{3}\rho_{3},\\
\tilde{\rho}_{10}^{\hat{n}_{2}}=\wp_{5}\rho_{5},\\
\tilde{\rho}_{11}^{\hat{n}_{2}}=\wp_{6}\rho_{6}.
\end{array}
\right.  \label{eq3.3}%
\end{equation}
Summing Eq. (\ref{eq3.3}), the left quantum result is $2\rho_{B}$, and the
right classical result is $\rho_{B}+\wp_{2}\rho_{2}+\wp_{3}\rho_{3}$. Then
taking the trace, we obtain \textquotedblleft$2_{Q}=\left(  1+\wp_{2}+\wp
_{3}\right)  _{C}$\textquotedblright. Because $%
%TCIMACRO{\dsum \limits_{\xi=1}^{6}}%
%BeginExpansion
{\displaystyle\sum\limits_{\xi=1}^{6}}
%EndExpansion
\wp_{\xi}=1$ and $\wp_{1}$, $\wp_{2}$, $\wp_{3}$, $\wp_{4}$, $\wp_{5}$, and
$\wp_{6}$ are non-zero, $0<$ $\wp_{2}+\wp_{3}<1$. In this case, we get the EPR
steering paradox \textquotedblleft$2_{Q}=\left(  1+\delta\right)  _{C}%
$\textquotedblright\ and $\delta=\wp_{2}+\wp_{3}$, with $0<\delta<1$.

\subsection{The same states appear in both identical and distinct sets of
measurements}

Here, we consider the case where $\left\{  \left\vert \eta_{i}\right\rangle
\right\}  $ is a subset of $\left\{  \left\vert \varepsilon_{j}\right\rangle
\right\}  $, but $\left\{  \left\vert \eta_{i}\right\rangle \right\}  $ and
$\left\{  \left\vert \varepsilon_{j}\right\rangle \right\}  $ are
nonidentical, i.e., combining cases 2 and 3. We consider the more extreme case
by assuming that the first $h$ terms in $\left\{  \left\vert \eta
_{i}\right\rangle \right\}  $ correspond equally to the first $h$ terms in
$\left\{  \left\vert \varepsilon_{j}\right\rangle \right\}  $ and that the
remaining $2^{M}-h$ terms in $\left\{  \left\vert \eta_{i}\right\rangle
\right\}  $ are all the same and equal to the $\left(  h+1\right)  $st term in
$\left\{  \left\vert \varepsilon_{j}\right\rangle \right\}  $, but that none
of the remaining $2^{M}-h$ terms in $\left\{  \left\vert \varepsilon
_{j}\right\rangle \right\}  $ are the same. In this scenario, Bob's
conditional states are%
\begin{equation}
\left\{
\begin{array}
[c]{l}%
\tilde{\rho}_{a_{1}}^{\hat{n}_{1}}=\sum_{\alpha}p_{\alpha}\left\vert
s_{1}^{\left(  \alpha\right)  }\right\vert ^{2}\left\vert \eta_{1}%
\right\rangle \left\langle \eta_{1}\right\vert ,\\
\tilde{\rho}_{a_{2}}^{\hat{n}_{1}}=\sum_{\alpha}p_{\alpha}\left\vert
s_{2}^{\left(  \alpha\right)  }\right\vert ^{2}\left\vert \eta_{2}%
\right\rangle \left\langle \eta_{2}\right\vert ,\\
\cdots\\
\tilde{\rho}_{a_{h}}^{\hat{n}_{1}}=\sum_{\alpha}p_{\alpha}\left\vert
s_{h}^{\left(  \alpha\right)  }\right\vert ^{2}\left\vert \eta_{h}%
\right\rangle \left\langle \eta_{h}\right\vert ,\\
\tilde{\rho}_{a_{h+1}}^{\hat{n}_{1}}=\sum_{\alpha}p_{\alpha}\left\vert
s_{h+1}^{\left(  \alpha\right)  }\right\vert ^{2}\left\vert \eta
_{h+1}\right\rangle \left\langle \eta_{h+1}\right\vert ,\\
\tilde{\rho}_{a_{h+2}}^{\hat{n}_{1}}=\sum_{\alpha}p_{\alpha}\left\vert
s_{h+2}^{\left(  \alpha\right)  }\right\vert ^{2}\left\vert \eta
_{h+1}\right\rangle \left\langle \eta_{h+1}\right\vert ,\\
\cdots\\
\tilde{\rho}_{a_{2^{M}}}^{\hat{n}_{1}}=\sum_{\alpha}p_{\alpha}\left\vert
s_{2^{M}}^{\left(  \alpha\right)  }\right\vert ^{2}\left\vert \eta
_{h+1}\right\rangle \left\langle \eta_{h+1}\right\vert ,
\end{array}
\right.  \left\{
\begin{array}
[c]{l}%
\tilde{\rho}_{a_{1}^{\prime}}^{\hat{n}_{2}}=\sum_{\alpha}p_{\alpha}\left\vert
t_{1}^{\left(  \alpha\right)  }\right\vert ^{2}\left\vert \eta_{1}%
\right\rangle \left\langle \eta_{1}\right\vert ,\\
\tilde{\rho}_{a_{2}^{\prime}}^{\hat{n}_{2}}=\sum_{\alpha}p_{\alpha}\left\vert
t_{2}^{\left(  \alpha\right)  }\right\vert ^{2}\left\vert \eta_{2}%
\right\rangle \left\langle \eta_{2}\right\vert ,\\
\cdots\\
\tilde{\rho}_{a_{h}^{\prime}}^{\hat{n}_{2}}=\sum_{\alpha}p_{\alpha}\left\vert
t_{h}^{\left(  \alpha\right)  }\right\vert ^{2}\left\vert \eta_{h}%
\right\rangle \left\langle \eta_{h}\right\vert ,\\
\tilde{\rho}_{a_{h+1}^{\prime}}^{\hat{n}_{2}}=\sum_{\alpha}p_{\alpha
}\left\vert t_{h+1}^{\left(  \alpha\right)  }\right\vert ^{2}\left\vert
\eta_{h+1}\right\rangle \left\langle \eta_{h+1}\right\vert ,\\
\tilde{\rho}_{a_{h+2}^{\prime}}^{\hat{n}_{2}}=\sum_{\alpha}p_{\alpha
}\left\vert t_{h+2}^{\left(  \alpha\right)  }\right\vert ^{2}\left\vert
\varepsilon_{h+2}\right\rangle \left\langle \varepsilon_{h+2}\right\vert ,\\
\cdots\\
\tilde{\rho}_{a_{2^{M}}^{\prime}}^{\hat{n}_{2}}=\sum_{\alpha}p_{\alpha
}\left\vert t_{2^{M}}^{\left(  \alpha\right)  }\right\vert ^{2}\left\vert
\varepsilon_{2^{M}}\right\rangle \left\langle \varepsilon_{2^{M}}\right\vert .
\end{array}
\right.  \label{c4.1}%
\end{equation}
In equation Eq. (\ref{c4.1}), there are $2^{M}$ different pure states in the
quantum result. For the LHS description, values of $\xi$ from $1$ to $2^{M}$
are sufficient. That is, the LHS model describes Bob's state as%
\begin{equation}
\left\{
\begin{array}
[c]{l}%
\tilde{\rho}_{a_{1}}^{\hat{n}_{1}}=%
%TCIMACRO{\dsum \limits_{\xi=1}^{2^{M}}}%
%BeginExpansion
{\displaystyle\sum\limits_{\xi=1}^{2^{M}}}
%EndExpansion
\wp\left(  a_{1}|\hat{n}_{1},\xi\right)  \wp_{\xi}\rho_{\xi},\\
\tilde{\rho}_{a_{2}}^{\hat{n}_{1}}=%
%TCIMACRO{\dsum \limits_{\xi=1}^{2^{M}}}%
%BeginExpansion
{\displaystyle\sum\limits_{\xi=1}^{2^{M}}}
%EndExpansion
\wp\left(  a_{2}|\hat{n}_{1},\xi\right)  \wp_{\xi}\rho_{\xi},\\
\cdots\\
\tilde{\rho}_{a_{h}}^{\hat{n}_{1}}=%
%TCIMACRO{\dsum \limits_{\xi=1}^{2^{M}}}%
%BeginExpansion
{\displaystyle\sum\limits_{\xi=1}^{2^{M}}}
%EndExpansion
\wp\left(  a_{h}|\hat{n}_{1},\xi\right)  \wp_{\xi}\rho_{\xi},\\
\tilde{\rho}_{a_{h+1}}^{\hat{n}_{1}}=%
%TCIMACRO{\dsum \limits_{\xi=1}^{2^{M}}}%
%BeginExpansion
{\displaystyle\sum\limits_{\xi=1}^{2^{M}}}
%EndExpansion
\wp\left(  a_{h+1}|\hat{n}_{1},\xi\right)  \wp_{\xi}\rho_{\xi},\\
\tilde{\rho}_{a_{h+2}}^{\hat{n}_{1}}=%
%TCIMACRO{\dsum \limits_{\xi=1}^{2^{M}}}%
%BeginExpansion
{\displaystyle\sum\limits_{\xi=1}^{2^{M}}}
%EndExpansion
\wp\left(  a_{h+2}|\hat{n}_{1},\xi\right)  \wp_{\xi}\rho_{\xi},\\
\cdots\\
\tilde{\rho}_{a_{2^{M}}}^{\hat{n}_{1}}=%
%TCIMACRO{\dsum \limits_{\xi=1}^{2^{M}}}%
%BeginExpansion
{\displaystyle\sum\limits_{\xi=1}^{2^{M}}}
%EndExpansion
\wp\left(  a_{2^{M}}|\hat{n}_{1},\xi\right)  \wp_{\xi}\rho_{\xi},
\end{array}
\right.  \left\{
\begin{array}
[c]{l}%
\tilde{\rho}_{a_{1}^{\prime}}^{\hat{n}_{2}}=%
%TCIMACRO{\dsum \limits_{\xi=1}^{2^{M}}}%
%BeginExpansion
{\displaystyle\sum\limits_{\xi=1}^{2^{M}}}
%EndExpansion
\wp\left(  a_{1}^{\prime}|\hat{n}_{2},\xi\right)  \wp_{\xi}\rho_{\xi},\\
\tilde{\rho}_{a_{2}^{\prime}}^{\hat{n}_{2}}=%
%TCIMACRO{\dsum \limits_{\xi=1}^{2^{M}}}%
%BeginExpansion
{\displaystyle\sum\limits_{\xi=1}^{2^{M}}}
%EndExpansion
\wp\left(  a_{2}^{\prime}|\hat{n}_{2},\xi\right)  \wp_{\xi}\rho_{\xi},\\
\cdots\\
\tilde{\rho}_{a_{h}^{\prime}}^{\hat{n}_{2}}=%
%TCIMACRO{\dsum \limits_{\xi=1}^{2^{M}}}%
%BeginExpansion
{\displaystyle\sum\limits_{\xi=1}^{2^{M}}}
%EndExpansion
\wp\left(  a_{h}^{\prime}|\hat{n}_{2},\xi\right)  \wp_{\xi}\rho_{\xi},\\
\tilde{\rho}_{a_{h+1}^{\prime}}^{\hat{n}_{2}}=%
%TCIMACRO{\dsum \limits_{\xi=1}^{2^{M}}}%
%BeginExpansion
{\displaystyle\sum\limits_{\xi=1}^{2^{M}}}
%EndExpansion
\wp\left(  a_{h+1}^{\prime}|\hat{n}_{2},\xi\right)  \wp_{\xi}\rho_{\xi},\\
\tilde{\rho}_{a_{h+2}^{\prime}}^{\hat{n}_{2}}=%
%TCIMACRO{\dsum \limits_{\xi=1}^{2^{M}}}%
%BeginExpansion
{\displaystyle\sum\limits_{\xi=1}^{2^{M}}}
%EndExpansion
\wp\left(  a_{h+2}^{\prime}|\hat{n}_{2},\xi\right)  \wp_{\xi}\rho_{\xi},\\
\cdots\\
\tilde{\rho}_{a_{2^{M}}^{\prime}}^{\hat{n}_{2}}=%
%TCIMACRO{\dsum \limits_{\xi=1}^{2^{M}}}%
%BeginExpansion
{\displaystyle\sum\limits_{\xi=1}^{2^{M}}}
%EndExpansion
\wp\left(  a_{2^{M}}^{\prime}|\hat{n}_{2},\xi\right)  \wp_{\xi}\rho_{\xi}.
\end{array}
\right.  \label{c4.2}%
\end{equation}
Eq. (\ref{c4.2}) can be expressed as each $\tilde{\rho}_{a}^{\hat{n}_{\ell}}$
containing only one term due to Bob's states being pure, i.e.,%
\begin{equation}
\left\{
\begin{array}
[c]{l}%
\tilde{\rho}_{a_{1}}^{\hat{n}_{1}}=\wp_{1}\rho_{1},\\
\tilde{\rho}_{a_{2}}^{\hat{n}_{1}}=\wp_{2}\rho_{2},\\
\cdots\\
\tilde{\rho}_{a_{h}}^{\hat{n}_{1}}=\wp_{h}\rho_{h},\\
\tilde{\rho}_{a_{h+1}}^{\hat{n}_{1}}=\wp\left(  a_{h+1}|\hat{n}_{1}%
,h+1\right)  \wp_{h+1}\rho_{h+1},\\
\tilde{\rho}_{a_{h+2}}^{\hat{n}_{1}}=\wp\left(  a_{h+2}|\hat{n}_{1}%
,h+1\right)  \wp_{h+1}\rho_{h+1},\\
\cdots\\
\tilde{\rho}_{a_{2^{M}}}^{\hat{n}_{1}}=\wp\left(  a_{2^{M}}|\hat{n}%
_{1},h+1\right)  \wp_{h+1}\rho_{h+1},
\end{array}
\right.  \left\{
\begin{array}
[c]{l}%
\tilde{\rho}_{a_{1}^{\prime}}^{\hat{n}_{2}}=\wp_{1}\rho_{1},\\
\tilde{\rho}_{a_{2}^{\prime}}^{\hat{n}_{2}}=\wp_{2}\rho_{2},\\
\cdots\\
\tilde{\rho}_{a_{h}^{\prime}}^{\hat{n}_{2}}=\wp_{h}\rho_{h},\\
\tilde{\rho}_{a_{h+1}^{\prime}}^{\hat{n}_{2}}=\wp_{h+1}\rho_{h+1},\\
\tilde{\rho}_{a_{h+2}^{\prime}}^{\hat{n}_{2}}=\wp_{h+2}\rho_{h+2},\\
\cdots\\
\tilde{\rho}_{a_{2^{M}}^{\prime}}^{\hat{n}_{2}}=\wp_{2^{M}}\rho_{2^{M}}.
\end{array}
\right.  \label{c4.3}%
\end{equation}
We assume that the first $h$ terms in $\left\{  \left\vert \eta_{i}%
\right\rangle \right\}  $ correspond equally to the first $h$ terms in
$\left\{  \left\vert \varepsilon_{j}\right\rangle \right\}  $, therefore we
describe the identical terms in $\left\{  \left\vert \eta_{i}\right\rangle
\right\}  $ and $\left\{  \left\vert \varepsilon_{j}\right\rangle \right\}  $
by the same hidden state. The remaining $2^{M}-h$ terms in $\left\{
\left\vert \eta_{i}\right\rangle \right\}  $ are all the same and equal to the
$\left(  h+1\right)  $st term in $\left\{  \left\vert \varepsilon
_{j}\right\rangle \right\}  $, hence describe them with $\wp_{h+1}\rho_{h+1}$.
And ${\sum\limits_{a}}\wp\left(  a|\hat{n},\xi\right)  =1$, so we have
$\wp\left(  a_{1}|\hat{n}_{1},1\right)  =\wp\left(  a_{2}|\hat{n}%
_{2},2\right)  =\cdots=\wp\left(  a_{h}|\hat{n}_{2},h\right)  =\wp\left(
a_{1}^{\prime}|\hat{n}_{2},1\right)  =\wp\left(  a_{2}^{\prime}|\hat{n}%
_{2},2\right)  =\cdots=\wp\left(  a_{h}^{\prime}|\hat{n}_{2},h\right)
=\wp\left(  a_{h+1}^{\prime}|\hat{n}_{2},h+1\right)  =\wp\left(
a_{h+2}^{\prime}|\hat{n}_{2},h+2\right)  =\cdots=\wp\left(  a_{2^{M}}^{\prime
}|\hat{n}_{2},2^{M}\right)  =1$, $\wp\left(  a_{h+1}|\hat{n}_{1},h+1\right)
+\wp\left(  a_{h+2}|\hat{n}_{1},h+1\right)  +\cdots+\wp\left(  a_{2^{M}}%
|\hat{n}_{1},h+1\right)  =1$ and the others $\wp\left(  a|\hat{n},\xi\right)
=0$. Summing Eq. (\ref{c4.3}), the left quantum result is $2\rho_{B}$, and the
right classical result is $\rho_{B}+%
%TCIMACRO{\dsum \limits_{\xi=1}^{h+1}}%
%BeginExpansion
{\displaystyle\sum\limits_{\xi=1}^{h+1}}
%EndExpansion
\wp_{\xi}\rho_{\xi}$, here $%
%TCIMACRO{\dsum \limits_{\xi=1}^{2^{M}}}%
%BeginExpansion
{\displaystyle\sum\limits_{\xi=1}^{2^{M}}}
%EndExpansion
\wp_{\xi}\rho_{\xi}=\rho_{B}$. Then taking the trace, we obtain
\textquotedblleft$2_{Q}=\left(  1+%
%TCIMACRO{\dsum \limits_{\xi=1}^{h+1}}%
%BeginExpansion
{\displaystyle\sum\limits_{\xi=1}^{h+1}}
%EndExpansion
\wp_{\xi}\right)  _{C}$\textquotedblright. Because $%
%TCIMACRO{\dsum \limits_{\xi=1}^{2^{M}}}%
%BeginExpansion
{\displaystyle\sum\limits_{\xi=1}^{2^{M}}}
%EndExpansion
\wp_{\xi}=1$ and there are $\left(  2^{M}-h\right)  $ $\wp_{\xi}$ non-zero,
$0<$ $%
%TCIMACRO{\dsum \limits_{\xi=1}^{h+1}}%
%BeginExpansion
{\displaystyle\sum\limits_{\xi=1}^{h+1}}
%EndExpansion
\wp_{\xi}<1$. In this case, we get the EPR steering paradox \textquotedblleft%
$2_{Q}=\left(  1+\delta\right)  _{C}$\textquotedblright\ and $\delta=%
%TCIMACRO{\dsum \limits_{\xi=1}^{h+1}}%
%BeginExpansion
{\displaystyle\sum\limits_{\xi=1}^{h+1}}
%EndExpansion
\wp_{\xi}$, with $0<\delta<1$.

\subsubsection{Example 4}

We consider the W state shared by Alice and Bob. Suppose Alice prepares the W
state $\rho_{AB}=\left\vert W\right\rangle \left\langle W\right\vert $, where%
\[
\left\vert W\right\rangle =\frac{1}{\sqrt{3}}\left[  \left\vert
100\right\rangle +\left\vert 010\right\rangle +\left\vert 001\right\rangle
\right]  ,
\]
or%
\begin{align*}
\left\vert W\right\rangle  &  =\frac{1}{2\sqrt{3}}\left(  \left\vert
++\right\rangle +\left\vert +-\right\rangle -\left\vert -+\right\rangle
-\left\vert --\right\rangle \right)  \left\vert 0\right\rangle \\
&  +\frac{1}{2\sqrt{3}}\left(  \left\vert ++\right\rangle -\left\vert
+-\right\rangle +\left\vert -+\right\rangle -\left\vert --\right\rangle
\right)  \left\vert 0\right\rangle \\
&  +\frac{1}{2\sqrt{3}}\left(  \left\vert ++\right\rangle +\left\vert
+-\right\rangle +\left\vert -+\right\rangle +\left\vert --\right\rangle
\right)  \left\vert 1\right\rangle \\
&  =\frac{1}{2\sqrt{3}}\left[  \left\vert ++\right\rangle \left(  2\left\vert
0\right\rangle +\left\vert 1\right\rangle \right)  +\left\vert +-\right\rangle
\left\vert 1\right\rangle +\left\vert -+\right\rangle \left\vert
1\right\rangle -\left\vert --\right\rangle \left(  2\left\vert 0\right\rangle
-\left\vert 1\right\rangle \right)  \right]  .
\end{align*}
She keeps particles 1, 2, and sends the other to Bob. In the two-setting
steering protocol $\left\{  \hat{n}_{1},\hat{n}_{2}\right\}  $ with%
\begin{equation}%
\begin{array}
[c]{c}%
\hat{n}_{1}=\sigma_{z}\sigma_{z}\equiv zz,\\
\hat{n}_{2}=\sigma_{x}\sigma_{x}\equiv xx.
\end{array}
\end{equation}
In the protocol, Bob asks Alice to carry out either one of two possible
projective measurements on her qubits, i.e.,%
\begin{equation}
\left\{
\begin{array}
[c]{l}%
\hat{P}_{00}^{zz}=\left\vert 00\right\rangle \left\langle 00\right\vert ,\\
\hat{P}_{01}^{zz}=\left\vert 01\right\rangle \left\langle 01\right\vert ,\\
\hat{P}_{10}^{zz}=\left\vert 10\right\rangle \left\langle 10\right\vert ,\\
\hat{P}_{11}^{zz}=\left\vert 11\right\rangle \left\langle 11\right\vert ,\\
\hat{P}_{00}^{xx}=\left\vert ++\right\rangle \left\langle ++\right\vert ,\\
\hat{P}_{01}^{xx}=\left\vert +-\right\rangle \left\langle +-\right\vert ,\\
\hat{P}_{10}^{xx}=\left\vert -+\right\rangle \left\langle -+\right\vert ,\\
\hat{P}_{11}^{xx}=\left\vert --\right\rangle \left\langle --\right\vert ,
\end{array}
\right.
\end{equation}
where $|\pm\rangle=(1/\sqrt{2})(\left\vert 0\right\rangle \pm\left\vert
1\right\rangle )$. After Alice's measurement, Bob's unnormalized conditional
states are%
\begin{equation}
\left\{
\begin{array}
[c]{l}%
\tilde{\rho}_{00}^{zz}=\frac{1}{3}\left\vert 1\right\rangle \left\langle
1\right\vert ,\\
\tilde{\rho}_{01}^{zz}=\frac{1}{3}\left\vert 0\right\rangle \left\langle
0\right\vert ,\\
\tilde{\rho}_{10}^{zz}=\frac{1}{3}\left\vert 0\right\rangle \left\langle
0\right\vert ,\\
\tilde{\rho}_{11}^{zz}=0,\\
\tilde{\rho}_{00}^{xx}=\frac{1}{12}\left(  2\left\vert 0\right\rangle
+\left\vert 1\right\rangle \right)  \left(  2\left\langle 0\right\vert
+\left\langle 1\right\vert \right)  ,\\
\tilde{\rho}_{01}^{xx}=\frac{1}{12}\left\vert 1\right\rangle \left\langle
1\right\vert ,\\
\tilde{\rho}_{10}^{xx}=\frac{1}{12}\left\vert 1\right\rangle \left\langle
1\right\vert ,\\
\tilde{\rho}_{11}^{xx}=\frac{1}{12}\left(  2\left\vert 0\right\rangle
-\left\vert 1\right\rangle \right)  \left(  2\left\langle 0\right\vert
-\left\langle 1\right\vert \right)  ,
\end{array}
\right.  \label{d1}%
\end{equation}
where%

\begin{align*}
\tilde{\rho}_{00}^{\hat{n}_{1}} &  =\mathrm{tr}_{A}\left[  \left(
P_{00}^{\hat{n}_{1}}{\otimes\mathds{1}}\right)  \rho_{AB}\right]  \\
&  =\mathrm{tr}_{A}\left[  \left(  \left\vert 00\right\rangle \left\langle
00\right\vert {\otimes\mathds{1}}\right)  \left\vert W\right\rangle
\left\langle W\right\vert \right]  \\
&  =\frac{1}{\sqrt{3}}\mathrm{tr}_{A}\left[  \left(  \left\vert
00\right\rangle _{A}{\otimes\mathds{1}}\right)  \left\vert 1\right\rangle
\left\langle W\right\vert \right]  \\
&  =\frac{1}{\sqrt{3}}\left(  \left\langle 00\right\vert _{A}{\otimes
\mathds{1}}\right)  \left(  \left\vert 00\right\rangle _{A}{\otimes
\mathds{1}}\right)  \left\vert 1\right\rangle \left\langle W\right\vert
\left(  \left\vert 00\right\rangle _{A}{\otimes\mathds{1}}\right)  \\
&  =\frac{1}{3}\left(  \left\langle 00\right\vert _{A}{\otimes\mathds{1}}%
\right)  \left(  \left\vert 00\right\rangle _{A}{\otimes\mathds{1}}\right)
\left\vert 1\right\rangle \left[  \left\langle 100\right\vert +\left\langle
010\right\vert +\left\langle 001\right\vert \right]  \left(  \left\vert
00\right\rangle _{A}{\otimes\mathds{1}}\right)  \\
&  =\frac{1}{3}\left\vert 1\right\rangle \left\langle 1\right\vert ,
\end{align*}

\begin{align*}
\tilde{\rho}_{01}^{\hat{n}_{1}} &  =\mathrm{tr}_{A}\left[  \left(
P_{01}^{\hat{n}_{1}}{\otimes\mathds{1}}\right)  \rho_{AB}\right]  \\
&  =\mathrm{tr}_{A}\left[  \left(  \left\vert 01\right\rangle \left\langle
01\right\vert {\otimes\mathds{1}}\right)  \left\vert W\right\rangle
\left\langle W\right\vert \right]  \\
&  =\frac{1}{\sqrt{3}}\mathrm{tr}_{A}\left[  \left(  \left\vert
01\right\rangle _{A}{\otimes\mathds{1}}\right)  \left\vert 0\right\rangle
\left\langle W\right\vert \right]  \\
&  =\frac{1}{\sqrt{3}}\left(  \left\langle 01\right\vert _{A}{\otimes
\mathds{1}}\right)  \left(  \left\vert 01\right\rangle _{A}{\otimes
\mathds{1}}\right)  \left\vert 0\right\rangle \left\langle W\right\vert
\left(  \left\vert 01\right\rangle _{A}{\otimes\mathds{1}}\right)  \\
&  =\frac{1}{3}\left(  \left\langle 01\right\vert _{A}{\otimes\mathds{1}}%
\right)  \left(  \left\vert 01\right\rangle _{A}{\otimes\mathds{1}}\right)
\left\vert 0\right\rangle \left[  \left\langle 100\right\vert +\left\langle
010\right\vert +\left\langle 001\right\vert \right]  \left(  \left\vert
01\right\rangle _{A}{\otimes\mathds{1}}\right)  \\
&  =\frac{1}{3}\left\vert 0\right\rangle \left\langle 0\right\vert ,
\end{align*}

\begin{align*}
\tilde{\rho}_{10}^{\hat{n}_{1}} &  =\mathrm{tr}_{A}\left[  \left(
P_{10}^{\hat{n}_{1}}{\otimes\mathds{1}}\right)  \rho_{AB}\right]  \\
&  =\mathrm{tr}_{A}\left[  \left(  \left\vert 10\right\rangle \left\langle
10\right\vert {\otimes\mathds{1}}\right)  \left\vert W\right\rangle
\left\langle W\right\vert \right]  \\
&  =\frac{1}{\sqrt{3}}\mathrm{tr}_{A}\left[  \left(  \left\vert
10\right\rangle _{A}{\otimes\mathds{1}}\right)  \left\vert 0\right\rangle
\left\langle W\right\vert \right]  \\
&  =\frac{1}{\sqrt{3}}\left(  \left\langle 10\right\vert _{A}{\otimes
\mathds{1}}\right)  \left(  \left\vert 10\right\rangle _{A}{\otimes
\mathds{1}}\right)  \left\vert 0\right\rangle \left\langle W\right\vert
\left(  \left\vert 10\right\rangle _{A}{\otimes\mathds{1}}\right)  \\
&  =\frac{1}{3}\left(  \left\langle 10\right\vert _{A}{\otimes\mathds{1}}%
\right)  \left(  \left\vert 10\right\rangle _{A}{\otimes\mathds{1}}\right)
\left\vert 0\right\rangle \left[  \left\langle 100\right\vert +\left\langle
010\right\vert +\left\langle 001\right\vert \right]  \left(  \left\vert
10\right\rangle _{A}{\otimes\mathds{1}}\right)  \\
&  =\frac{1}{3}\left\vert 0\right\rangle \left\langle 0\right\vert ,
\end{align*}

\begin{align*}
\tilde{\rho}_{11}^{\hat{n}_{1}} &  =\mathrm{tr}_{A}\left[  \left(
P_{11}^{\hat{n}_{1}}{\otimes\mathds{1}}\right)  \rho_{AB}\right]  \\
&  =\mathrm{tr}_{A}\left[  \left(  \left\vert 11\right\rangle \left\langle
11\right\vert {\otimes\mathds{1}}\right)  \left\vert W\right\rangle
\left\langle W\right\vert \right]  \\
&  =0,
\end{align*}

\begin{align*}
\tilde{\rho}_{00}^{\hat{n}_{2}} &  =\mathrm{tr}_{A}\left[  \left(
P_{00}^{\hat{n}_{2}}{\otimes\mathds{1}}\right)  \rho_{AB}\right]  \\
&  =\mathrm{tr}_{A}\left[  \left(  \left\vert ++\right\rangle \left\langle
++\right\vert {\otimes\mathds{1}}\right)  \left\vert W\right\rangle
\left\langle W\right\vert \right]  \\
&  =\frac{1}{2\sqrt{3}}\mathrm{tr}_{A}\left[  \left(  \left\vert
++\right\rangle _{A}{\otimes\mathds{1}}\right)  \left(  2\left\vert
0\right\rangle +\left\vert 1\right\rangle \right)  \left\langle W\right\vert
\right]  \\
&  =\frac{1}{2\sqrt{3}}\left(  \left\langle ++\right\vert _{A}{\otimes
\mathds{1}}\right)  \left(  \left\vert ++\right\rangle _{A}{\otimes
\mathds{1}}\right)  \left(  2\left\vert 0\right\rangle +\left\vert
1\right\rangle \right)  \left\langle W\right\vert \left(  \left\vert
++\right\rangle _{A}{\otimes\mathds{1}}\right)  \\
&  =\frac{1}{12}\left(  \left\langle ++\right\vert _{A}{\otimes\mathds{1}}%
\right)  \left(  \left\vert ++\right\rangle _{A}{\otimes\mathds{1}}\right)
\left(  2\left\vert 0\right\rangle +\left\vert 1\right\rangle \right)  \left[
\left\langle ++\right\vert \left(  2\left\langle 0\right\vert +\left\langle
1\right\vert \right)  +\left\langle +-\right\vert \left\langle 1\right\vert
\right.  \\
&  \left.  +\left\langle -+\right\vert \left\langle 1\right\vert -\left\langle
--\right\vert \left(  2\left\langle 0\right\vert -\left\langle 1\right\vert
\right)  \right]  \left(  \left\vert ++\right\rangle _{A}{\otimes
\mathds{1}}\right)  \\
&  =\frac{1}{12}\left(  2\left\vert 0\right\rangle +\left\vert 1\right\rangle
\right)  \left(  2\left\langle 0\right\vert +\left\langle 1\right\vert
\right)  ,
\end{align*}

\begin{align*}
\tilde{\rho}_{01}^{\hat{n}_{2}} &  =\mathrm{tr}_{A}\left[  \left(
P_{01}^{\hat{n}_{2}}{\otimes\mathds{1}}\right)  \rho_{AB}\right]  \\
&  =\mathrm{tr}_{A}\left[  \left(  \left\vert +-\right\rangle \left\langle
+-\right\vert {\otimes\mathds{1}}\right)  \left\vert W\right\rangle
\left\langle W\right\vert \right]  \\
&  =\frac{1}{2\sqrt{3}}\mathrm{tr}_{A}\left[  \left(  \left\vert
+-\right\rangle _{A}{\otimes\mathds{1}}\right)  \left\vert 1\right\rangle
\left\langle W\right\vert \right]  \\
&  =\frac{1}{2\sqrt{3}}\left(  \left\langle +-\right\vert _{A}{\otimes
\mathds{1}}\right)  \left(  \left\vert +-\right\rangle _{A}{\otimes
\mathds{1}}\right)  \left\vert 1\right\rangle \left\langle W\right\vert
\left(  \left\vert +-\right\rangle _{A}{\otimes\mathds{1}}\right)  \\
&  =\frac{1}{12}\left(  \left\langle +-\right\vert _{A}{\otimes\mathds{1}}%
\right)  \left(  \left\vert +-\right\rangle _{A}{\otimes\mathds{1}}\right)
\left\vert 1\right\rangle \left[  \left\langle ++\right\vert \left(
2\left\langle 0\right\vert +\left\langle 1\right\vert \right)  +\left\langle
+-\right\vert \left\langle 1\right\vert \right.  \\
&  \left.  +\left\langle -+\right\vert \left\langle 1\right\vert -\left\langle
--\right\vert \left(  2\left\langle 0\right\vert -\left\langle 1\right\vert
\right)  \right]  \left(  \left\vert +-\right\rangle _{A}{\otimes
\mathds{1}}\right)  \\
&  =\frac{1}{12}\left\vert 1\right\rangle \left\langle 1\right\vert ,
\end{align*}

\begin{align*}
\tilde{\rho}_{10}^{\hat{n}_{2}} &  =\mathrm{tr}_{A}\left[  \left(
P_{10}^{\hat{n}_{2}}{\otimes\mathds{1}}\right)  \rho_{AB}\right]  \\
&  =\mathrm{tr}_{A}\left[  \left(  \left\vert -+\right\rangle \left\langle
-+\right\vert {\otimes\mathds{1}}\right)  \left\vert W\right\rangle
\left\langle W\right\vert \right]  \\
&  =\frac{1}{2\sqrt{3}}\mathrm{tr}_{A}\left[  \left(  \left\vert
-+\right\rangle _{A}{\otimes\mathds{1}}\right)  \left\vert 1\right\rangle
\left\langle W\right\vert \right]  \\
&  =\frac{1}{2\sqrt{3}}\left(  \left\langle -+\right\vert _{A}{\otimes
\mathds{1}}\right)  \left(  \left\vert -+\right\rangle _{A}{\otimes
\mathds{1}}\right)  \left\vert 1\right\rangle \left\langle W\right\vert
\left(  \left\vert -+\right\rangle _{A}{\otimes\mathds{1}}\right)  \\
&  =\frac{1}{12}\left(  \left\langle -+\right\vert _{A}{\otimes\mathds{1}}%
\right)  \left(  \left\vert -+\right\rangle _{A}{\otimes\mathds{1}}\right)
\left\vert 1\right\rangle \left[  \left\langle ++\right\vert \left(
2\left\langle 0\right\vert +\left\langle 1\right\vert \right)  +\left\langle
+-\right\vert \left\langle 1\right\vert \right.  \\
&  \left.  +\left\langle -+\right\vert \left\langle 1\right\vert -\left\langle
--\right\vert \left(  2\left\langle 0\right\vert -\left\langle 1\right\vert
\right)  \right]  \left(  \left\vert -+\right\rangle _{A}{\otimes
\mathds{1}}\right)  \\
&  =\frac{1}{12}\left\vert 1\right\rangle \left\langle 1\right\vert ,
\end{align*}

\begin{align*}
\tilde{\rho}_{11}^{\hat{n}_{2}} &  =\mathrm{tr}_{A}\left[  \left(
P_{11}^{\hat{n}_{2}}{\otimes\mathds{1}}\right)  \rho_{AB}\right]  \\
&  =\mathrm{tr}_{A}\left[  \left(  \left\vert --\right\rangle \left\langle
--\right\vert {\otimes\mathds{1}}\right)  \left\vert W\right\rangle
\left\langle W\right\vert \right]  \\
&  =-\frac{1}{2\sqrt{3}}\mathrm{tr}_{A}\left[  \left(  \left\vert
--\right\rangle _{A}{\otimes\mathds{1}}\right)  \left(  2\left\vert
0\right\rangle -\left\vert 1\right\rangle \right)  \left\langle W\right\vert
\right]  \\
&  =-\frac{1}{2\sqrt{3}}\left(  \left\langle --\right\vert _{A}{\otimes
\mathds{1}}\right)  \left(  \left\vert --\right\rangle _{A}{\otimes
\mathds{1}}\right)  \left(  2\left\vert 0\right\rangle -\left\vert
1\right\rangle \right)  \left\langle W\right\vert \left(  \left\vert
--\right\rangle _{A}{\otimes\mathds{1}}\right)  \\
&  =-\frac{1}{12}\left(  \left\langle --\right\vert _{A}{\otimes
\mathds{1}}\right)  \left(  \left\vert --\right\rangle _{A}{\otimes
\mathds{1}}\right)  \left(  2\left\vert 0\right\rangle -\left\vert
1\right\rangle \right)  \left[  \left\langle ++\right\vert \left(
2\left\langle 0\right\vert +\left\langle 1\right\vert \right)  +\left\langle
+-\right\vert \left\langle 1\right\vert \right.  \\
&  \left.  +\left\langle -+\right\vert \left\langle 1\right\vert -\left\langle
--\right\vert \left(  2\left\langle 0\right\vert -\left\langle 1\right\vert
\right)  \right]  \left(  \left\vert --\right\rangle _{A}{\otimes
\mathds{1}}\right)  \\
&  =\frac{1}{12}\left(  2\left\vert 0\right\rangle -\left\vert 1\right\rangle
\right)  \left(  2\left\langle 0\right\vert -\left\langle 1\right\vert
\right)  .
\end{align*}
In the quantum result, Bob obtains seven pure states. Summing the Eq.
(\ref{d1}) and taking the trace, we get \emph{tr}$\left(  2\rho_{B}\right)
=2$.

If Bob's states have a LHS description, they satisfy Eqs. (\ref{eq2.1}) and
(\ref{eq2.2}). In the quantum result, there are $4$ pure states in Eq.
(\ref{d1}). It is sufficient to take $\xi$ from $1$ to $4$. Bob's states can
be write as
\begin{equation}
\left\{
\begin{array}
[c]{l}%
\tilde{\rho}_{00}^{zz}=\wp\left(  00|\hat{z}\hat{z},1\right)  \wp_{1}\rho
_{1}+\wp\left(  00|\hat{z}\hat{z},2\right)  \wp_{2}\rho_{2}+\wp\left(
00|\hat{z}\hat{z},3\right)  \wp_{3}\rho_{3}+\wp\left(  00|\hat{z}\hat
{z},4\right)  \wp_{4}\rho_{4},\\
\tilde{\rho}_{01}^{zz}=\wp\left(  01|\hat{z}\hat{z},1\right)  \wp_{1}\rho
_{1}+\wp\left(  01|\hat{z}\hat{z},2\right)  \wp_{2}\rho_{2}+\wp\left(
01|\hat{z}\hat{z},3\right)  \wp_{3}\rho_{3}+\wp\left(  01|\hat{z}\hat
{z},4\right)  \wp_{4}\rho_{4},\\
\tilde{\rho}_{10}^{zz}=\wp\left(  10|\hat{z}\hat{z},1\right)  \wp_{1}\rho
_{1}+\wp\left(  10|\hat{z}\hat{z},2\right)  \wp_{2}\rho_{2}+\wp\left(
10|\hat{z}\hat{z},3\right)  \wp_{3}\rho_{3}+\wp\left(  10|\hat{z}\hat
{z},4\right)  \wp_{4}\rho_{4},\\
\tilde{\rho}_{11}^{zz}=0,\\
\tilde{\rho}_{00}^{xx}=\wp\left(  00|\hat{x}\hat{x},1\right)  \wp_{1}\rho
_{1}+\wp\left(  00|\hat{x}\hat{x},2\right)  \wp_{2}\rho_{2}+\wp\left(
00|\hat{x}\hat{x},3\right)  \wp_{3}\rho_{3}+\wp\left(  00|\hat{x}\hat
{x},4\right)  \wp_{4}\rho_{4},\\
\tilde{\rho}_{01}^{xx}=\wp\left(  01|\hat{x}\hat{x},1\right)  \wp_{1}\rho
_{1}+\wp\left(  01|\hat{x}\hat{x},2\right)  \wp_{2}\rho_{2}+\wp\left(
01|\hat{x}\hat{x},3\right)  \wp_{3}\rho_{3}+\wp\left(  01|\hat{x}\hat
{x},4\right)  \wp_{4}\rho_{4},\\
\tilde{\rho}_{10}^{xx}=\wp\left(  10|\hat{x}\hat{x},1\right)  \wp_{1}\rho
_{1}+\wp\left(  10|\hat{x}\hat{x},2\right)  \wp_{2}\rho_{2}+\wp\left(
10|\hat{x}\hat{x},3\right)  \wp_{3}\rho_{3}+\wp\left(  10|\hat{x}\hat
{x},4\right)  \wp_{4}\rho_{4},\\
\tilde{\rho}_{11}^{xx}=\wp\left(  11|\hat{x}\hat{x},1\right)  \wp_{1}\rho
_{1}+\wp\left(  11|\hat{x}\hat{x},2\right)  \wp_{2}\rho_{2}+\wp\left(
11|\hat{x}\hat{x},3\right)  \wp_{3}\rho_{3}+\wp\left(  11|\hat{x}\hat
{x},4\right)  \wp_{4}\rho_{4}.
\end{array}
\right.  \label{eq2.7}%
\end{equation}

Because the eight states on the left-hand side of Eq. (\ref{eq2.7}) are all
pure states, and a pure state can only be expanded by itself, then each
$\tilde{\rho}_{a}^{\hat{n}_{\ell}}$ in Eq. (\ref{eq2.7}) contains only one
term, e.g.,%
\begin{equation}
\left\{
\begin{array}
[c]{l}%
\tilde{\rho}_{00}^{zz}=\wp\left(  00|\hat{z}\hat{z},1\right)  \wp_{1}\rho
_{1},\\
\tilde{\rho}_{01}^{zz}=\wp\left(  01|\hat{z}\hat{z},2\right)  \wp_{2}\rho
_{2},\\
\tilde{\rho}_{10}^{zz}=\wp\left(  10|\hat{z}\hat{z},2\right)  \wp_{2}\rho
_{2},\\
\tilde{\rho}_{11}^{zz}=0,\\
\tilde{\rho}_{00}^{xx}=\wp\left(  00|\hat{x}\hat{x},3\right)  \wp_{3}\rho
_{3},\\
\tilde{\rho}_{01}^{xx}=\wp\left(  01|\hat{x}\hat{x},1\right)  \wp_{1}\rho
_{1},\\
\tilde{\rho}_{10}^{xx}=\wp\left(  10|\hat{x}\hat{x},1\right)  \wp_{1}\rho
_{1},\\
\tilde{\rho}_{11}^{xx}=\wp\left(  11|\hat{x}\hat{x},4\right)  \wp_{4}\rho_{4}.
\end{array}
\right.  \label{d3}%
\end{equation}
The other $\wp\left(  a|\hat{n},\xi\right)  $ are $0$. Because in the quantum
result Eq. (\ref{d1}), $\tilde{\rho}_{00}^{zz}=\tilde{\rho}_{01}^{xx}%
=\tilde{\rho}_{10}^{xx}$, we take the same hide state $\wp_{1}\rho_{1}$ to
describe $\tilde{\rho}_{00}^{zz}$, $\tilde{\rho}_{01}^{xx}$, and $\tilde{\rho
}_{10}^{xx}$. Similarly, $\tilde{\rho}_{01}^{zz}=\tilde{\rho}_{10}^{zz}$, we
take $\wp_{2}\rho_{2}$ to describe them. Then because $%
%TCIMACRO{\dsum \limits_{a}}%
%BeginExpansion
{\displaystyle\sum\limits_{a}}
%EndExpansion
\wp\left(  a|\hat{n},\xi\right)  =1$, for example, for $\tilde{\rho}_{01}%
^{zz}$ and $\tilde{\rho}_{10}^{zz}$, $\wp\left(  00|\hat{z}\hat{z},2\right)
=0$, we have $\wp\left(  01|\hat{z}\hat{z},2\right)  +\wp\left(  10|\hat
{z}\hat{z},2\right)  =1$. Similarly, we can get $\wp\left(  01|\hat{x}\hat
{x},1\right)  +\wp\left(  10|\hat{x}\hat{x},1\right)  =1$, and $\wp\left(
00|\hat{z}\hat{z},1\right)  =\wp\left(  00|\hat{x}\hat{x},3\right)
=\wp\left(  11|\hat{x}\hat{x},4\right)  =1$. Then Eq. (\ref{d3}) can be
written as%
\begin{equation}
\left\{
\begin{array}
[c]{l}%
\tilde{\rho}_{00}^{zz}=\wp_{1}\rho_{1},\\
\tilde{\rho}_{01}^{zz}=\wp\left(  01|\hat{z}\hat{z},2\right)  \wp_{2}\rho
_{2},\\
\tilde{\rho}_{10}^{zz}=\wp\left(  10|\hat{z}\hat{z},2\right)  \wp_{2}\rho
_{2},\\
\tilde{\rho}_{11}^{zz}=0,\\
\tilde{\rho}_{00}^{xx}=\wp_{3}\rho_{3},\\
\tilde{\rho}_{01}^{xx}=\wp\left(  01|\hat{x}\hat{x},1\right)  \wp_{1}\rho
_{1},\\
\tilde{\rho}_{10}^{xx}=\wp\left(  10|\hat{x}\hat{x},1\right)  \wp_{1}\rho
_{1},\\
\tilde{\rho}_{11}^{xx}=\wp_{4}\rho_{4}.
\end{array}
\right.  \label{d4}%
\end{equation}

Because of $%
%TCIMACRO{\dsum \limits_{\xi}}%
%BeginExpansion
{\displaystyle\sum\limits_{\xi}}
%EndExpansion
\wp_{\xi}\rho_{\xi}=\rho_{B}$, we sum up Eq. (\ref{d4}) and take the trace,
and the right-hand side gives%
\begin{equation}
tr\left[  2\wp_{1}\rho_{1}+\wp_{2}\rho_{2}+\wp_{3}\rho_{3}+\wp_{4}\rho
_{4}\right]  =tr\left[  \rho_{B}+\wp_{1}\rho_{1}\right]  =1+\wp_{1}.
\label{eq2.15}%
\end{equation}
Since $\tilde{\rho}_{00}^{zz}$, $\tilde{\rho}_{01}^{zz}$, $\tilde{\rho}%
_{10}^{zz}$, $\tilde{\rho}_{00}^{xx}$, $\tilde{\rho}_{01}^{xx}$, $\tilde{\rho
}_{10}^{xx}$, and $\tilde{\rho}_{11}^{xx}$ are nonzero in the quantum result
Eq. (\ref{d1}), $\wp_{1}$, $\wp_{2}$, $\wp_{3}$, and $\wp_{4}$ are all
nonzero. And because $%
%TCIMACRO{\dsum \limits_{\xi}}%
%BeginExpansion
{\displaystyle\sum\limits_{\xi}}
%EndExpansion
\wp_{\xi}=1$, $0<\wp_{1}<1$. However, the result of the left-hand side is 2.
We obtain the paradox \textquotedblleft$2_{Q}=\left(  1+\wp_{1}\right)  _{C}%
$\textquotedblright. In the case of the W state, we get the contradiction
\textquotedblleft$2_{Q}=\left(  1+\delta\right)  _{C}$\textquotedblright,
where $0<\delta<1$.

\subsection{Bob's two sets of results are identical}

Finally, we consider what the result would be if $\left\{  \left\vert \eta
_{i}\right\rangle \right\}  =\left\{  \left\vert \varepsilon_{j}\right\rangle
\right\}  $. Suppose that there are $2^{M}$ distinct elements in $\left\{
\left\vert \eta_{i}\right\rangle \right\}  $ and that they are the same as the
corresponding $2^{M}$ elements in $\left\{  \left\vert \varepsilon
_{j}\right\rangle \right\}  $. Bob's conditional states are%
\begin{equation}
\left\{
\begin{array}
[c]{l}%
\tilde{\rho}_{a_{1}}^{\hat{n}_{1}}=\sum_{\alpha}p_{\alpha}\left\vert
s_{1}^{\left(  \alpha\right)  }\right\vert ^{2}\left\vert \eta_{1}%
\right\rangle \left\langle \eta_{1}\right\vert ,\\
\cdots\\
\tilde{\rho}_{a_{2^{M}}}^{\hat{n}_{1}}=\sum_{\alpha}p_{\alpha}\left\vert
s_{2^{M}}^{\left(  \alpha\right)  }\right\vert ^{2}\left\vert \eta_{2^{M}%
}\right\rangle \left\langle \eta_{2^{M}}\right\vert ,\\
\tilde{\rho}_{a_{1}^{\prime}}^{\hat{n}_{2}}=\sum_{\alpha}p_{\alpha}\left\vert
t_{1}^{\left(  \alpha\right)  }\right\vert ^{2}\left\vert \eta_{1}%
\right\rangle \left\langle \eta_{1}\right\vert ,\\
\cdots\\
\tilde{\rho}_{a_{2^{M}}^{\prime}}^{\hat{n}_{2}}=\sum_{\alpha}p_{\alpha
}\left\vert t_{2^{M}}^{\left(  \alpha\right)  }\right\vert ^{2}\left\vert
\eta_{2^{M}}\right\rangle \left\langle \eta_{2^{M}}\right\vert .
\end{array}
\right.  \label{c5.1}%
\end{equation}
In this case, there are only $2^{M}$ different\ pure states in the quantum
result Eq. (\ref{c5.1}). It is sufficient to take $\xi$ from 1 to $2^{M}$,
e.g.,
\begin{equation}
\left\{
\begin{array}
[c]{l}%
\tilde{\rho}_{a_{1}}^{\hat{n}_{1}}=%
%TCIMACRO{\dsum \limits_{\xi=1}^{2^{M}}}%
%BeginExpansion
{\displaystyle\sum\limits_{\xi=1}^{2^{M}}}
%EndExpansion
\wp\left(  a_{1}|\hat{n}_{1},\xi\right)  \wp_{\xi}\rho_{\xi},\\
\cdots\\
\tilde{\rho}_{a_{2^{M}}}^{\hat{n}_{1}}=%
%TCIMACRO{\dsum \limits_{\xi=1}^{2^{M}}}%
%BeginExpansion
{\displaystyle\sum\limits_{\xi=1}^{2^{M}}}
%EndExpansion
\wp\left(  a_{2^{M}}|\hat{n}_{1},\xi\right)  \wp_{\xi}\rho_{\xi},\\
\tilde{\rho}_{a_{1}^{\prime}}^{\hat{n}_{2}}=%
%TCIMACRO{\dsum \limits_{\xi=1}^{2^{M}}}%
%BeginExpansion
{\displaystyle\sum\limits_{\xi=1}^{2^{M}}}
%EndExpansion
\wp\left(  a_{1}^{\prime}|\hat{n}_{2},\xi\right)  \wp_{\xi}\rho_{\xi},\\
\cdots\\
\tilde{\rho}_{a_{2^{M}}^{\prime}}^{\hat{n}_{2}}=%
%TCIMACRO{\dsum \limits_{\xi=1}^{2^{M}}}%
%BeginExpansion
{\displaystyle\sum\limits_{\xi=1}^{2^{M}}}
%EndExpansion
\wp\left(  a_{2^{M}}^{\prime}|\hat{n}_{2},\xi\right)  \wp_{\xi}\rho_{\xi}.
\end{array}
\right.  \label{c5.2}%
\end{equation}
Since Bob's states are all pure, each $\tilde{\rho}_{a}^{\hat{n}_{\ell}}$ in
Eq. (\ref{c5.2}) contains only one term. Eq. (\ref{c5.2}) can be written as%
\begin{equation}
\left\{
\begin{array}
[c]{l}%
\tilde{\rho}_{a_{1}}^{\hat{n}_{1}}=\wp_{1}\rho_{1},\\
\cdots\\
\tilde{\rho}_{a_{2^{M}}}^{\hat{n}_{1}}=\wp_{2^{M}}\rho_{2^{M}},\\
\tilde{\rho}_{a_{1}^{\prime}}^{\hat{n}_{2}}=\wp_{1}\rho_{1},\\
\cdots\\
\tilde{\rho}_{a_{2^{M}}^{\prime}}^{\hat{n}_{2}}=\wp_{2^{M}}\rho_{2^{M}}.
\end{array}
\right.  \label{c5.3}%
\end{equation}
Since in quantum result we assume that $\left\{  \left\vert \eta
_{i}\right\rangle \right\}  =\left\{  \left\vert \varepsilon_{j}\right\rangle
\right\}  $, in the LHS description we describe the same term with the same
hidden state. And ${\sum\limits_{a}}\wp\left(  a|\hat{n},\xi\right)  =1$, so
we have $\wp\left(  a_{1}|\hat{n}_{1},1\right)  =\cdots=\wp\left(  a_{2^{M}%
}|\hat{n}_{1},2^{M}\right)  =\wp\left(  a_{1}^{\prime}|\hat{n}_{2},1\right)
\cdots=\wp\left(  a_{2^{M}}^{\prime}|\hat{n}_{2},2^{M}\right)  =1$, and the
others $\wp\left(  a|\hat{n},\xi\right)  =0$. Then we take the sum of Eq.
(\ref{c5.3}) and take the trace, and we get \textquotedblleft$2_{Q}=2_{C}%
$\textquotedblright. That means that, in this case, there is no contradiction,
and we cannot conclude whether Alice can steer Bob.

\subsubsection{Example 5}

Suppose Alice prepares a 3-qubit state $\rho_{AB}=\left\vert \psi\right\rangle
\left\langle \psi\right\vert $, where%
\[
\left\vert \psi\right\rangle =\frac{1}{2}\left[  \left\vert 00\right\rangle
+\left\vert 01\right\rangle +\left\vert 10\right\rangle +\left\vert
11\right\rangle \right]  \otimes\left\vert 0\right\rangle .
\]
She keeps particles 1, 2, and sends the other to Bob. In the two-setting
steering protocol $\left\{  \hat{n}_{1},\hat{n}_{2}\right\}  $ with%
\begin{equation}%
\begin{array}
[c]{c}%
\hat{n}_{1}=\sigma_{z}\sigma_{z}\equiv zz,\\
\hat{n}_{2}=\sigma_{x}\sigma_{x}\equiv xx.
\end{array}
\label{e1}%
\end{equation}
In the protocol, Bob asks Alice to carry out either one of two possible
projective measurements on her qubits, i.e.,%
\begin{equation}
\left\{
\begin{array}
[c]{l}%
\hat{P}_{00}^{zz}=\left\vert 00\right\rangle \left\langle 00\right\vert ,\\
\hat{P}_{01}^{zz}=\left\vert 01\right\rangle \left\langle 01\right\vert ,\\
\hat{P}_{10}^{zz}=\left\vert 10\right\rangle \left\langle 10\right\vert ,\\
\hat{P}_{11}^{zz}=\left\vert 11\right\rangle \left\langle 11\right\vert ,\\
\hat{P}_{00}^{xx}=\left\vert ++\right\rangle \left\langle ++\right\vert ,\\
\hat{P}_{01}^{xx}=\left\vert +-\right\rangle \left\langle +-\right\vert ,\\
\hat{P}_{10}^{xx}=\left\vert -+\right\rangle \left\langle -+\right\vert ,\\
\hat{P}_{11}^{xx}=\left\vert --\right\rangle \left\langle --\right\vert ,
\end{array}
\right.  \label{e2}%
\end{equation}
where $|\pm\rangle=(1/\sqrt{2})(\left\vert 0\right\rangle \pm\left\vert
1\right\rangle )$. After Alice's measurement, Bob's unnormalized conditional
states are%
\begin{equation}
\left\{
\begin{array}
[c]{l}%
\tilde{\rho}_{00}^{zz}=\frac{1}{4}\left\vert 0\right\rangle \left\langle
0\right\vert ,\\
\tilde{\rho}_{01}^{zz}=\frac{1}{4}\left\vert 0\right\rangle \left\langle
0\right\vert ,\\
\tilde{\rho}_{10}^{zz}=\frac{1}{4}\left\vert 0\right\rangle \left\langle
0\right\vert ,\\
\tilde{\rho}_{11}^{zz}=\frac{1}{4}\left\vert 0\right\rangle \left\langle
0\right\vert ,\\
\tilde{\rho}_{00}^{xx}=\left\vert 0\right\rangle \left\langle 0\right\vert ,\\
\tilde{\rho}_{01}^{xx}=0,\\
\tilde{\rho}_{10}^{xx}=0,\\
\tilde{\rho}_{11}^{xx}=0.
\end{array}
\right.  \label{e3}%
\end{equation}
In the quantum result, Bob obtains 5 pure states. Summing the Eq. (\ref{e3})
and taking the trace, we get \emph{tr}$\left(  2\rho_{B}\right)  =2$.

If Bob's states have a LHS description, they satisfy Eqs. (\ref{eq2.1}) and
(\ref{eq2.2}). In the quantum result, there are $1$ pure state in Eq.
(\ref{e3}). It is sufficient to take $\xi=1$. Bob's states can be write as
\begin{equation}
\left\{
\begin{array}
[c]{l}%
\tilde{\rho}_{00}^{zz}=\wp\left(  00|\hat{z}\hat{z},1\right)  \wp_{1}\rho
_{1},\\
\tilde{\rho}_{01}^{zz}=\wp\left(  01|\hat{z}\hat{z},1\right)  \wp_{1}\rho
_{1},\\
\tilde{\rho}_{10}^{zz}=\wp\left(  10|\hat{z}\hat{z},1\right)  \wp_{1}\rho
_{1},\\
\tilde{\rho}_{11}^{zz}=\wp\left(  11|\hat{z}\hat{z},1\right)  \wp_{1}\rho
_{1},\\
\tilde{\rho}_{00}^{xx}=\wp\left(  00|\hat{x}\hat{x},1\right)  \wp_{1}\rho
_{1},\\
\tilde{\rho}_{01}^{xx}=0,\\
\tilde{\rho}_{10}^{xx}=0,\\
\tilde{\rho}_{11}^{xx}=0.
\end{array}
\right.  \label{e4}%
\end{equation}

Because in the quantum result Eq. (\ref{e3}), $\tilde{\rho}_{00}^{zz}%
=\tilde{\rho}_{01}^{zz}=\tilde{\rho}_{10}^{zz}=\tilde{\rho}_{00}^{zz}%
=\tilde{\rho}_{00}^{xx}$, we take the same hide state $\wp_{1}\rho_{1}$ to
describe them. Then because $%
%TCIMACRO{\dsum \limits_{a}}%
%BeginExpansion
{\displaystyle\sum\limits_{a}}
%EndExpansion
\wp\left(  a|\hat{n},\xi\right)  =1$, we have $\wp\left(  00|\hat{z}\hat
{z},1\right)  +\wp\left(  01|\hat{z}\hat{z},1\right)  +\wp\left(  10|\hat
{z}\hat{z},1\right)  +\wp\left(  11|\hat{z}\hat{z},1\right)  =1$, and
$\wp\left(  00|\hat{x}\hat{x},1\right)  =1$. Then Eq. (\ref{e4}) can be
written as%
\begin{equation}
\left\{
\begin{array}
[c]{l}%
\tilde{\rho}_{00}^{zz}=\wp\left(  00|\hat{z}\hat{z},1\right)  \wp_{1}\rho
_{1},\\
\tilde{\rho}_{01}^{zz}=\wp\left(  01|\hat{z}\hat{z},1\right)  \wp_{1}\rho
_{1},\\
\tilde{\rho}_{10}^{zz}=\wp\left(  10|\hat{z}\hat{z},1\right)  \wp_{1}\rho
_{1},\\
\tilde{\rho}_{11}^{zz}=\wp\left(  11|\hat{z}\hat{z},1\right)  \wp_{1}\rho
_{1},\\
\tilde{\rho}_{00}^{xx}=\wp_{1}\rho_{1},\\
\tilde{\rho}_{01}^{xx}=0,\\
\tilde{\rho}_{10}^{xx}=0,\\
\tilde{\rho}_{11}^{xx}=0.
\end{array}
\right.  \label{e5}%
\end{equation}

Because of $%
%TCIMACRO{\dsum \limits_{\xi}}%
%BeginExpansion
{\displaystyle\sum\limits_{\xi}}
%EndExpansion
\wp_{\xi}\rho_{\xi}=\rho_{B}$, e.g., $\wp_{1}\rho_{1}=\rho_{B}$, we sum up Eq.
(\ref{e5}) and take the trace, and the right hand side gives%
\begin{equation}
tr\left[  2\wp_{1}\rho_{1}\right]  =tr\left[  2\rho_{B}\right]  =2.
\end{equation}
We obtain \textquotedblleft$2_{Q}=2_{C}$\textquotedblright. There is no
contradiction in this case, so we cannot conclude whether Alice can steer Bob.

\section{$k=1+\delta_{k}$ for $N$-qubit state}

In the 2-setting protocol, we propose that the EPR steering paradox
\textquotedblleft$2_{Q}=\left(  1+\delta\right)  _{C}$\textquotedblright\ is
obtained when the N-qubit entangled state $\rho_{AB}$ shared by Alice and Bob
satisfies the measurement requirement. Here we show a more general steering
paradox \textquotedblleft$k_{Q}=\left(  1+\delta_{k}\right)  _{C}%
$\textquotedblright. In a $k$-setting steering protocol $\left\{  \hat{n}%
_{1},\hat{n}_{2},\cdots,\hat{n}_{k}\right\}  $, consider Alice and Bob share
an $N$-qubit\ entangled state
\begin{equation}
\rho_{AB}=\sum_{\alpha}p_{\alpha}\left\vert \psi_{AB}^{\left(  \alpha\right)
}\right\rangle \left\langle \psi_{AB}^{\left(  \alpha\right)  }\right\vert ,
\end{equation}
where Alice retains $M$ $\left(  M<N\right)  $ particles and sends the
remaining $\left(  N-M\right)  $ particles to Bob, and
\begin{equation}
\left\vert \psi_{AB}^{\left(  \alpha\right)  }\right\rangle =\sum_{i}\left(
s_{i}^{\hat{n}_{\ell}\left(  \alpha\right)  }\left\vert \phi_{i}^{\hat
{n}_{\ell}}\right\rangle \left\vert \eta_{i}^{\hat{n}_{\ell}\left(
\alpha\right)  }\right\rangle \right)  ,
\end{equation}
in which $i=1,2,\cdots,2^{M}$ and $\hat{n}_{\ell}$ $\left(  \ell
=1,2,\cdots,k\right)  $ the measurement direction. In the k-setting steering
protocol $\left\{  \hat{n}_{1},\hat{n}_{2},\cdots,\hat{n}_{k}\right\}  $,
Alice performs $k\cdot2^{M}$ projective measurements $\hat{P}_{a}^{\hat
{n}_{\ell}}\,$, where $a$ is the measurement result. For each $\hat{P}%
_{a}^{\hat{n}_{\ell}}$, Bob obtains the corresponding unnormalized state
$\tilde{\rho}_{a}^{\hat{n}_{\ell}}=\mathrm{tr}_{A}\left[  \left(  {{P}%
_{a}^{\hat{n}_{\ell}}\otimes\mathds{1}}_{2^{N-M}\times2^{N-M}}\right)
\rho_{AB}\right]  $, where ${\mathds{1}}_{2^{N-M}\times2^{N-M}}$ is the
$2^{N-M}\times2^{N-M}$ identity matrix.

We propose a lemma for $N$-qubit in the k-setting steering protocol $\left\{
\hat{n}_{1},\hat{n}_{2},\cdots,\hat{n}_{k}\right\}  $.

\begin{lemma}
In the k-setting steering protocol $\left\{  \hat{n}_{1},\hat{n}_{2}%
,\cdots,\hat{n}_{k}\right\}  $, Alice and Bob share an $N$-qubit entangled
state.\ Assume that Alice performs $k\cdot2^{M}$ projective measurements
$\hat{P}_{a}^{\hat{n}_{\ell}}\,$ and Bob obtains the corresponding conditional
state $\tilde{\rho}_{a}^{\hat{n}_{\ell}}$. Considering the scenario where
Bob's conditional states $\{\tilde{\rho}_{a}^{\hat{n}_{\ell}}\}$ are all pure,
there will be a contradiction of \textquotedblleft$k_{Q}=\left(  1+\delta
_{k}\right)  _{C}$\textquotedblright\ $\left(  \text{with }0\leq\delta
_{k}<k-1\right)  $\ if Bob obtains $k$ sets of results $\{\tilde{\rho}%
_{a}^{\hat{n}_{\ell}}\}$ that are not exactly the same, i.e., there exists at
least one set of $\{\tilde{\rho}_{a}^{\hat{n}_{p}}\}$, with $p\in\left[
1,k\right]  $, different from the others.
\end{lemma}

\begin{proof}
The sets of projectors for Alice are as follows%
\[
{{P}_{a_{i}}^{\hat{n}_{\ell}}=}\left\vert \phi_{i}^{\hat{n}_{\ell}%
}\right\rangle \left\langle \phi_{i}^{\hat{n}_{\ell}}\right\vert
,i=1,2,\cdots,2^{M}.
\]
After Alice's measurement, Bob obtains
\begin{equation}%
\begin{array}
[c]{c}%
\tilde{\rho}_{a_{i}}^{\hat{n}_{\ell}}=\sum_{\alpha}p_{\alpha}|s_{i}^{\hat
{n}_{\ell}(\alpha)}|^{2}|\eta_{i}^{\hat{n}_{\ell}(\alpha)}\rangle\langle
\eta_{i}^{\hat{n}_{\ell}(\alpha)}|,
\end{array}
\label{s1}%
\end{equation}
If Bob's states have a LHS description, they satisfy Eqs. (\ref{eq2.1}) and
(\ref{eq2.2}), which lead to $k\cdot2^{M}$ equations:
\begin{equation}
\tilde{\rho}_{a_{i}}^{\hat{n}_{\ell}}=%
%TCIMACRO{\dsum \limits_{\xi}}%
%BeginExpansion
{\displaystyle\sum\limits_{\xi}}
%EndExpansion
\wp\left(  a_{i}|\hat{n}_{\ell},\xi\right)  \wp_{\xi}\rho_{\xi}.\label{s2}%
\end{equation}
The quantum value $k$ can be attained through summing Eq. (\ref{s1}) and then
taking the trace. However the maximum value of the LHS result is $k$ if and
only if
\begin{align*}
\sum_{i}\sum_{\xi}\wp\left(  a_{i}|\hat{n}_{1},\xi\right)  \wp_{\xi}\rho_{\xi}
&  =\sum_{i}\sum_{\xi}\wp\left(  a_{i}^{\prime}|\hat{n}_{2},\xi\right)
\wp_{\xi}\rho_{\xi}\\
&  =\cdots=\sum_{i}\sum_{\xi}\wp\left(  a_{i}^{\prime\prime}|\hat{n}_{k}%
,\xi\right)  \wp_{\xi}\rho_{\xi}\\
&  =\sum_{\xi}\wp_{\xi}\rho_{\xi}.
\end{align*}

Considering that the conditional states of Bob are all pure, which implies
that all $\left\vert \eta_{i}^{\hat{n}_{\ell}}\right\rangle $ are independent
of $\alpha$, i.e.,
\begin{equation}
\left\vert \psi_{AB}^{\left(  \alpha\right)  }\right\rangle =\sum_{i}\left(
s_{i}^{\hat{n}_{\ell}\left(  \alpha\right)  }\left\vert \phi_{i}^{\hat
{n}_{\ell}}\right\rangle \left\vert \eta_{i}^{\hat{n}_{\ell}}\right\rangle
\right)  ,
\end{equation}
Then the conditional states of Bob are%
\begin{equation}
\left\{
\begin{array}
[c]{l}%
\tilde{\rho}_{a_{1}}^{\hat{n}_{1}}=\sum_{\alpha}p_{\alpha}\left\vert
s_{1}^{\hat{n}_{1}\left(  \alpha\right)  }\right\vert ^{2}\left\vert \eta
_{1}^{\hat{n}_{1}}\right\rangle \left\langle \eta_{1}^{\hat{n}_{1}}\right\vert
,\\
\cdots\\
\tilde{\rho}_{a_{2^{M}}}^{\hat{n}_{1}}=\sum_{\alpha}p_{\alpha}\left\vert
s_{2^{M}}^{\hat{n}_{1}\left(  \alpha\right)  }\right\vert ^{2}\left\vert
\eta_{2^{M}}^{\hat{n}_{1}}\right\rangle \left\langle \eta_{2^{M}}^{\hat{n}%
_{1}}\right\vert ,\\
\tilde{\rho}_{a_{1}^{\prime}}^{\hat{n}_{2}}=\sum_{\alpha}p_{\alpha}\left\vert
s_{1}^{\hat{n}_{2}\left(  \alpha\right)  }\right\vert ^{2}\left\vert \eta
_{1}^{\hat{n}_{2}}\right\rangle \left\langle \eta_{1}^{\hat{n}_{2}}\right\vert
,\\
\cdots\\
\tilde{\rho}_{a_{2^{M}}^{\prime}}^{\hat{n}_{2}}=\sum_{\alpha}p_{\alpha
}\left\vert s_{2^{M}}^{\hat{n}_{2}\left(  \alpha\right)  }\right\vert
^{2}\left\vert \eta_{2^{M}}^{\hat{n}_{2}}\right\rangle \left\langle
\eta_{2^{M}}^{\hat{n}_{2}}\right\vert ,\\
\cdots\\
\tilde{\rho}_{a_{1}^{\prime\prime}}^{\hat{n}_{k}}=\sum_{\alpha}p_{\alpha
}\left\vert s_{1}^{\hat{n}_{k}\left(  \alpha\right)  }\right\vert
^{2}\left\vert \eta_{1}^{\hat{n}_{k}}\right\rangle \left\langle \eta_{1}%
^{\hat{n}_{k}}\right\vert ,\\
\cdots\\
\tilde{\rho}_{a_{2^{M}}^{\prime\prime}}^{\hat{n}_{k}}=\sum_{\alpha}p_{\alpha
}\left\vert s_{2^{M}}^{\hat{n}_{k}\left(  \alpha\right)  }\right\vert
^{2}\left\vert \eta_{2^{M}}^{\hat{n}_{k}}\right\rangle \left\langle
\eta_{2^{M}}^{\hat{n}_{k}}\right\vert .
\end{array}
\right.
\end{equation}
\emph{Sufficiency}---\textquotedblleft$k_{Q}=k_{C}$\textquotedblright%
$\Longrightarrow\{\tilde{\rho}_{a}^{\hat{n}_{1}}\}=\{\tilde{\rho}_{a^{\prime}%
}^{\hat{n}_{2}}\}=\cdots=\{\tilde{\rho}_{a^{\prime\prime}}^{\hat{n}_{k}}\}$.

Since the density matrix of a pure state can only be expanded by itself, that
is each $\tilde{\rho}_{a_{i}}^{\hat{n}_{\ell}}$ in Eq. (\ref{s2}) can only be
described by a definite hidden state. Without loss of generality, suppose that
Bob's conditional states are completely different in the set involving the
same measurement direction $\tilde{\rho}_{a}^{\hat{n}_{\ell}}$. Further Eq.
(\ref{s2}) can be written as
\begin{equation}
\left\{
\begin{array}
[c]{l}%
\tilde{\rho}_{a_{1}}^{\hat{n}_{1}}=\wp_{1}\rho_{1},\\
\cdots\\
\tilde{\rho}_{a_{2^{M}}}^{\hat{n}_{1}}=\wp_{2^{M}}\rho_{2^{M}},\\
\tilde{\rho}_{a_{1}^{\prime}}^{\hat{n}_{2}}=\wp_{1}^{\prime}\rho_{1}^{\prime
},\\
\cdots\\
\tilde{\rho}_{a_{2^{M}}^{\prime}}^{\hat{n}_{2}}=\wp_{2^{M}}^{\prime}%
\rho_{2^{M}}^{\prime},\\
\cdots\\
\tilde{\rho}_{a_{1}^{\prime\prime}}^{\hat{n}_{k}}=\wp_{1}^{\prime\prime}%
\rho_{1}^{\prime\prime},\\
\cdots\\
\tilde{\rho}_{a_{2^{M}}^{\prime\prime}}^{\hat{n}_{k}}=\wp_{2^{M}}%
^{\prime\prime}\rho_{2^{M}}^{\prime\prime}.
\end{array}
\right.  \label{s3}%
\end{equation}
Summing the Eq. (\ref{s3}), we have
\[
\tilde{\rho}_{a}^{\hat{n}_{1}}+\tilde{\rho}_{a^{\prime}}^{\hat{n}_{2}}%
+\cdots+\tilde{\rho}_{a^{\prime\prime}}^{\hat{n}_{k}}=\sum_{i=1}^{2^{M}}%
\wp_{i}\rho_{i}+\sum_{i=1}^{2^{M}}\wp_{i}^{\prime}\rho_{i}^{\prime}%
+\cdots+\sum_{i=1}^{2^{M}}\wp_{i}^{\prime\prime}\rho_{i}^{\prime\prime}.
\]
The equation \textquotedblleft$k_{Q}=k_{C}$\textquotedblright\ requests
\begin{equation}
\sum_{i=1}^{2^{M}}\wp_{i}\rho_{i}=\sum_{i=1}^{2^{M}}\wp_{i}^{\prime}\rho
_{i}^{\prime}=\cdots=\sum_{i=1}^{2^{M}}\wp_{i}^{\prime\prime}\rho_{i}%
^{\prime\prime}=\sum_{\xi}\wp_{\xi}\rho_{\xi}.
\end{equation}
It means that $\{\tilde{\rho}_{a}^{\hat{n}_{1}}\}=\{\tilde{\rho}_{a^{\prime}%
}^{\hat{n}_{2}}\}=\cdots=\{\tilde{\rho}_{a^{\prime\prime}}^{\hat{n}_{k}}\}$,
i.e., Bob's $k$ sets of results are identical. Note that there is also no
contradiction appearing if Bob has Bob has $k$ sets of identical results.
Therefore, its contrapositive also holds, and if at least one set of results
$\{\tilde{\rho}_{a}^{\hat{n}_{p}}\}$, with $p\in\left[  1,k\right]  $ is not
identical to the others, the paradoxical equality \textquotedblleft%
$k_{Q}=\left(  1+\delta_{k}\right)  _{C}$\textquotedblright\ can be obtained.

\emph{Necessity}---$\{\tilde{\rho}_{a}^{\hat{n}_{1}}\}=\{\tilde{\rho
}_{a^{\prime}}^{\hat{n}_{2}}\}=\cdots=\{\tilde{\rho}_{a^{\prime\prime}}%
^{\hat{n}_{k}}\}\Longrightarrow$\textquotedblleft$k_{Q}=k_{C}$%
\textquotedblright.

Suppose that there are $2^{M}$ distinct states in $\{\tilde{\rho}_{a}^{\hat
{n}_{1}}\}$ and that they are the same as the corresponding $\left(
k-1\right)  \cdot2^{M}$ elements in $\{\tilde{\rho}_{a}^{\hat{n}_{h}}\}$, with
$h=2,3,\cdots,k$. Bob's conditional states are
\begin{equation}
\left\{
\begin{array}
[c]{l}%
\tilde{\rho}_{a_{1}}^{\hat{n}_{1}}=\sum_{\alpha}p_{\alpha}\left\vert
s_{1}^{\hat{n}_{1}\left(  \alpha\right)  }\right\vert ^{2}\left\vert \eta
_{1}\right\rangle \left\langle \eta_{1}\right\vert ,\\
\cdots\\
\tilde{\rho}_{a_{2^{M}}}^{\hat{n}_{1}}=\sum_{\alpha}p_{\alpha}\left\vert
s_{2^{M}}^{\hat{n}_{1}\left(  \alpha\right)  }\right\vert ^{2}\left\vert
\eta_{2^{M}}\right\rangle \left\langle \eta_{2^{M}}\right\vert ,\\
\tilde{\rho}_{a_{1}^{\prime}}^{\hat{n}_{2}}=\sum_{\alpha}p_{\alpha}\left\vert
s_{1}^{\hat{n}_{2}\left(  \alpha\right)  }\right\vert ^{2}\left\vert \eta
_{1}\right\rangle \left\langle \eta_{1}\right\vert ,\\
\cdots\\
\tilde{\rho}_{a_{2^{M}}^{\prime}}^{\hat{n}_{2}}=\sum_{\alpha}p_{\alpha
}\left\vert s_{2^{M}}^{\hat{n}_{2}\left(  \alpha\right)  }\right\vert
^{2}\left\vert \eta_{2^{M}}\right\rangle \left\langle \eta_{2^{M}}\right\vert
,\\
\cdots\\
\tilde{\rho}_{a_{1}^{\prime\prime}}^{\hat{n}_{k}}=\sum_{\alpha}p_{\alpha
}\left\vert s_{1}^{\hat{n}_{k}\left(  \alpha\right)  }\right\vert
^{2}\left\vert \eta_{1}\right\rangle \left\langle \eta_{1}\right\vert ,\\
\cdots\\
\tilde{\rho}_{a_{2^{M}}^{\prime\prime}}^{\hat{n}_{k}}=\sum_{\alpha}p_{\alpha
}\left\vert s_{2^{M}}^{\hat{n}_{k}\left(  \alpha\right)  }\right\vert
^{2}\left\vert \eta_{2^{M}}\right\rangle \left\langle \eta_{2^{M}}\right\vert
.
\end{array}
\right.  \label{n1}%
\end{equation}
There are only $2^{M}$ different\ pure states in the quantum result Eq.
(\ref{n1}). It is sufficient to take $\xi$ from $1$ to $2^{M}$ in the LHS
description, e.g.,
\begin{equation}
\left\{
\begin{array}
[c]{l}%
\tilde{\rho}_{a}^{\hat{n}_{1}}=\sum\limits_{\xi=1}^{2^{M}}\wp\left(  a|\hat
{n}_{1},\xi\right)  \wp_{\xi}\rho_{\xi},\\
\tilde{\rho}_{a^{\prime}}^{\hat{n}_{2}}=\sum\limits_{\xi=1}^{2^{M}}\wp\left(
a^{\prime}|\hat{n}_{2},\xi\right)  \wp_{\xi}\rho_{\xi},\\
\cdots\\
\tilde{\rho}_{a^{\prime\prime}}^{\hat{n}_{k}}=\sum\limits_{\xi=1}^{2^{M}}%
\wp\left(  a^{\prime\prime}|\hat{n}_{k},\xi\right)  \wp_{\xi}\rho_{\xi}.
\end{array}
\right.  \label{n2}%
\end{equation}
Since Bob's states are all pure, each $\tilde{\rho}_{a}^{\hat{n}_{\ell}}$ in
Eq. (\ref{n2}) contains only one term. Eq. (\ref{n2}) can be written as
\begin{equation}
\left\{
\begin{array}
[c]{l}%
\tilde{\rho}_{a_{1}}^{\hat{n}_{1}}=\wp_{1}\rho_{1},\\
\cdots\\
\tilde{\rho}_{a_{2^{M}}}^{\hat{n}_{1}}=\wp_{2^{M}}\rho_{2^{M}},\\
\tilde{\rho}_{a_{1}^{\prime}}^{\hat{n}_{2}}=\wp_{1}\rho_{1},\\
\cdots\\
\tilde{\rho}_{a_{2^{M}}^{\prime}}^{\hat{n}_{2}}=\wp_{2^{M}}\rho_{2^{M}},\\
\cdots\\
\tilde{\rho}_{a_{1}^{\prime\prime}}^{\hat{n}_{k}}=\wp_{1}\rho_{1},\\
\cdots\\
\tilde{\rho}_{a_{2^{M}}^{\prime\prime}}^{\hat{n}_{k}}=\wp_{2^{M}}\rho_{2^{M}}.
\end{array}
\right.  \label{n3}%
\end{equation}
Here we describe the same term with the same hidden state. And ${\sum
\limits_{a}}\wp\left(  a|\hat{n},\xi\right)  =1$, so we have $\wp\left(
a_{1}|\hat{n}_{1},1\right)  =\cdots=\wp\left(  a_{2^{M}}|\hat{n}_{1}%
,2^{M}\right)  =\wp\left(  a_{1}^{\prime}|\hat{n}_{2},1\right)  \cdots
=\wp\left(  a_{2^{M}}^{\prime}|\hat{n}_{2},2^{M}\right)  =\cdots=\wp\left(
a_{1}^{\prime\prime}|\hat{n}_{k},1\right)  \cdots=\wp\left(  a_{2^{M}}%
^{\prime\prime}|\hat{n}_{k},2^{M}\right)  =1$, and the others $\wp\left(
a|\hat{n},\xi\right)  =0$. Then we take the sum of Eq. (\ref{n3}) and take the
trace, and get \textquotedblleft$k_{Q}=k_{C}$\textquotedblright. That means
that there is no contradiction if $\{\tilde{\rho}_{a}^{\hat{n}_{1}}%
\}=\{\tilde{\rho}_{a^{\prime}}^{\hat{n}_{2}}\}=\cdots=\{\tilde{\rho
}_{a^{\prime\prime}}^{\hat{n}_{k}}\}$, and we cannot conclude whether Alice
can steer Bob. Consequently, its contrapositive also holds, and if there is a
contradiction \textquotedblleft$k_{Q}=\left(  1+\delta_{k}\right)  _{C}%
$\textquotedblright, then it requests that at least one set of results
$\{\tilde{\rho}_{a}^{\hat{n}_{p}}\}$, with $p\in\left[  1,k\right]  $ is not
identical to the others.
\end{proof}

\section{$k=1+\delta_{k}$ for $N$-qudit state}

We now generalize the result to the most general case. In the k-setting
steering protocol, assume that Alice and Bob share a mixed entangled state of
$N$ $d$-dimensional systems ($N$-qudit). Then, we can obtain a similar simple
result for \textquotedblleft$k_{Q}=\left(  1+\delta_{k}\right)  _{C}%
$\textquotedblright\ $\left(  \text{with }0\leq\delta_{k}<k-1\right)  $.

We propose a lemma for $N$-qudit in the k-setting steering protocol $\left\{
\hat{n}_{1},\hat{n}_{2},\cdots,\hat{n}_{k}\right\}  $.

\begin{lemma}
In the k-setting steering protocol $\left\{  \hat{n}_{1},\hat{n}_{2}%
,\cdots,\hat{n}_{k}\right\}  $, Alice and Bob share an $N$-qudit entangled
state.\ Assume that Alice performs $k\cdot d^{M}$ projective measurements
$\hat{P}_{a}^{\hat{n}_{\ell}}\,$ and Bob obtains the corresponding conditional
state $\tilde{\rho}_{a}^{\hat{n}_{\ell}}$. Considering the scenario where
Bob's conditional states $\{\tilde{\rho}_{a}^{\hat{n}_{\ell}}\}$ are all pure,
there will be a contradiction of \textquotedblleft$k_{Q}=\left(  1+\delta
_{k}\right)  _{C}$\textquotedblright\ $\left(  \text{with }0\leq\delta
_{k}<k-1\right)  $\ if Bob obtains k sets of results $\{\tilde{\rho}_{a}%
^{\hat{n}_{\ell}}\}$ that are not exactly the same i.e., there exists at least
one set of $\{\tilde{\rho}_{a}^{\hat{n}_{p}}\}$, with $p\in\left[  1,k\right]
$, different from the others.
\end{lemma}

\begin{proof}
Let us consider that Alice prepared an $N$-qudit state
\begin{equation}
\rho_{AB}=\sum_{\alpha}p_{\alpha}\left\vert \psi_{AB}^{\left(  \alpha\right)
}\right\rangle \left\langle \psi_{AB}^{\left(  \alpha\right)  }\right\vert .
\end{equation}
Similarly, Alice keeps $M\left(  M<N\right)  $ particles and sends the
remaining $(N-M)$ particles to Bob. In the k-setting steering protocol
$\left\{  \hat{n}_{1},\hat{n}_{2},\cdots,\hat{n}_{k}\right\}  $, Alice
performs $k\cdot d^{M}$ projective measurements $\hat{P}_{a}^{\hat{n}_{\ell}}$
($\ell=1,2,\cdots,k$). Bob obtains the corresponding unnormalized states are
$\tilde{\rho}_{a}^{\hat{n}_{\ell}}=\mathrm{tr}_{A}\left[  \left(  {{P}%
_{a}^{\hat{n}_{\ell}}\otimes\mathds{1}}_{d^{N-M}\times d^{N-M}}\right)
\rho_{AB}\right]  $. The difference is that ${\mathds{1}}_{d^{N-M}\times
d^{N-M}}$ is the $d^{N-M}\times d^{N-M}$ identity matrix. $\left\vert
\psi_{AB}^{\left(  \alpha\right)  }\right\rangle $ can be written as%
\begin{equation}
\left\vert \psi_{AB}^{\left(  \alpha\right)  }\right\rangle =\sum_{i}\left(
s_{i}^{\hat{n}_{\ell}\left(  \alpha\right)  }\left\vert \phi_{i}^{\hat
{n}_{\ell}}\right\rangle \left\vert \eta_{i}^{\hat{n}_{\ell}\left(
\alpha\right)  }\right\rangle \right)  ,
\end{equation}
here $i=1,2,\cdots,d^{M}$. We require Bob's conditional states are all pure,
then $\left\vert \eta_{i}^{\hat{n}_{\ell}\left(  \alpha\right)  }\right\rangle
$ independent of $\alpha$.

Suppose the sets of projectors for Alice are as follows%
\[
{{P}_{a}^{\hat{n}_{\ell}}=}\left\vert \phi_{i}^{\hat{n}_{\ell}}\right\rangle
\left\langle \phi_{i}^{\hat{n}_{\ell}}\right\vert ,i=1,2,\cdots,d^{M}.
\]
After Alice's measurement, Bob obtains
\begin{equation}
\tilde{\rho}_{a_{i}}^{\hat{n}_{\ell}}=\sum_{\alpha}p_{\alpha}|s_{i}^{\hat
{n}_{\ell}(\alpha)}|^{2}|\eta_{i}^{\hat{n}_{\ell}}\rangle\langle\eta_{i}%
^{\hat{n}_{\ell}}|.\label{s22}%
\end{equation}
If Bob's states have a LHS description, they satisfy Eqs. (\ref{eq2.1}) and
(\ref{eq2.2}), which lead to $k\cdot d^{M}$ equations:
\begin{equation}
\tilde{\rho}_{a_{i}}^{\hat{n}_{\ell}}=%
%TCIMACRO{\dsum \limits_{\xi}}%
%BeginExpansion
{\displaystyle\sum\limits_{\xi}}
%EndExpansion
\wp\left(  a_{i}|\hat{n}_{\ell},\xi\right)  \wp_{\xi}\rho_{\xi}.\label{s21}%
\end{equation}
The quantum value $k$ can be attained through summing Eq. (\ref{s21}) and then
taking the trace. However the maximum value of the LHS result is $k$ if and
only if
\begin{align*}
\sum_{i}\sum_{\xi}\wp\left(  a_{i}|\hat{n}_{1},\xi\right)  \wp_{\xi}\rho_{\xi}
&  =\sum_{i}\sum_{\xi}\wp\left(  a_{i}^{\prime}|\hat{n}_{2},\xi\right)
\wp_{\xi}\rho_{\xi}\\
&  =\cdots=\sum_{i}\sum_{\xi}\wp\left(  a_{i}^{\prime\prime}|\hat{n}_{k}%
,\xi\right)  \wp_{\xi}\rho_{\xi}\\
&  =\sum_{\xi}\wp_{\xi}\rho_{\xi}.
\end{align*}

\emph{Sufficiency}---\textquotedblleft$k_{Q}=k_{C}$\textquotedblright%
$\Longrightarrow\{\tilde{\rho}_{a}^{\hat{n}_{1}}\}=\{\tilde{\rho}_{a^{\prime}%
}^{\hat{n}_{2}}\}=\cdots=\{\tilde{\rho}_{a^{\prime\prime}}^{\hat{n}_{k}}\}$.

Since the density matrix of a pure state can only be expanded by itself, that
is each $\tilde{\rho}_{a_{i}}^{\hat{n}_{\ell}}$ in Eq. (\ref{s22}) can only be
described by a definite hidden state. Without loss of generality, suppose that
Bob's conditional states are completely different in the set involving the
same measurement direction $\tilde{\rho}_{a}^{\hat{n}_{\ell}}$. Further Eq.
(\ref{s22}) can be written as
\begin{equation}
\left\{
\begin{array}
[c]{l}%
\tilde{\rho}_{a_{1}}^{\hat{n}_{1}}=\wp_{1}\rho_{1},\\
\cdots\\
\tilde{\rho}_{a_{d^{M}}}^{\hat{n}_{1}}=\wp_{d^{M}}\rho_{d^{M}},\\
\tilde{\rho}_{a_{1}^{\prime}}^{\hat{n}_{2}}=\wp_{1}^{\prime}\rho_{1}^{\prime
},\\
\cdots\\
\tilde{\rho}_{a_{d^{M}}^{\prime}}^{\hat{n}_{2}}=\wp_{d^{M}}^{\prime}%
\rho_{d^{M}}^{\prime},\\
\cdots\\
\tilde{\rho}_{a_{1}^{\prime\prime}}^{\hat{n}_{k}}=\wp_{1}^{\prime\prime}%
\rho_{1}^{\prime\prime},\\
\cdots\\
\tilde{\rho}_{a_{d^{M}}^{\prime\prime}}^{\hat{n}_{k}}=\wp_{d^{M}}%
^{\prime\prime}\rho_{d^{M}}^{\prime\prime}.
\end{array}
\right.  \label{s23}%
\end{equation}
Summing the Eq. (\ref{s23}), we have
\[
\tilde{\rho}_{a}^{\hat{n}_{1}}+\tilde{\rho}_{a^{\prime}}^{\hat{n}_{2}}%
+\cdots+\tilde{\rho}_{a^{\prime\prime}}^{\hat{n}_{k}}=\sum_{i=1}^{d^{M}}%
\wp_{i}\rho_{i}+\sum_{i=1}^{d^{M}}\wp_{i}^{\prime}\rho_{i}^{\prime}%
+\cdots+\sum_{i=1}^{d^{M}}\wp_{i}^{\prime\prime}\rho_{i}^{\prime\prime}.
\]
The equation \textquotedblleft$k_{Q}=k_{C}$\textquotedblright\ requests
\begin{equation}
\sum_{i=1}^{d^{M}}\wp_{i}^{\prime}\rho_{i}^{\prime}=\sum_{i=1}^{d^{M}}\wp
_{i}\rho_{i}=\cdots=\sum_{i=1}^{d^{M}}\wp_{i}^{\prime\prime}\rho_{i}%
^{\prime\prime}=\sum_{\xi}\wp_{\xi}\rho_{\xi}.
\end{equation}
It means that $\{\tilde{\rho}_{a}^{\hat{n}_{1}}\}=\{\tilde{\rho}_{a^{\prime}%
}^{\hat{n}_{2}}\}=\cdots=\{\tilde{\rho}_{a^{\prime\prime}}^{\hat{n}_{k}}\}$,
i.e., Bob's $k$ sets of results are identical. Note that there is also no
contradiction appearing if Bob has Bob has $k$ sets of identical results.
Therefore, its contrapositive also holds, and if at least one set of results
$\{\tilde{\rho}_{a}^{\hat{n}_{p}}\}$, with $p\in\left[  1,k\right]  $ is not
identical to the others, the paradoxical equality \textquotedblleft%
$k_{Q}=\left(  1+\delta_{k}\right)  _{C}$\textquotedblright\ can be obtained.

\emph{Necessity}---$\{\tilde{\rho}_{a}^{\hat{n}_{1}}\}=\{\tilde{\rho
}_{a^{\prime}}^{\hat{n}_{2}}\}=\cdots=\{\tilde{\rho}_{a^{\prime\prime}}%
^{\hat{n}_{k}}\}\Longrightarrow$\textquotedblleft$k_{Q}=k_{C}$%
\textquotedblright.

Suppose that there are $2^{M}$ distinct states in $\{\tilde{\rho}_{a}^{\hat
{n}_{1}}\}$ and that they are the same as the corresponding $\left(
k-1\right)  \cdot d^{M}$ elements in $\{\tilde{\rho}_{a}^{\hat{n}_{h}}\}$,
with $h=2,3,\cdots,k$. Bob's conditional states are
\begin{equation}
\left\{
\begin{array}
[c]{l}%
\tilde{\rho}_{a_{1}}^{\hat{n}_{1}}=\sum_{\alpha}p_{\alpha}\left\vert
s_{1}^{\hat{n}_{1}\left(  \alpha\right)  }\right\vert ^{2}\left\vert \eta
_{1}\right\rangle \left\langle \eta_{1}\right\vert ,\\
\cdots\\
\tilde{\rho}_{a_{d^{M}}}^{\hat{n}_{1}}=\sum_{\alpha}p_{\alpha}\left\vert
s_{d^{M}}^{\hat{n}_{1}\left(  \alpha\right)  }\right\vert ^{2}\left\vert
\eta_{d^{M}}\right\rangle \left\langle \eta_{d^{M}}\right\vert ,\\
\tilde{\rho}_{a_{1}^{\prime}}^{\hat{n}_{2}}=\sum_{\alpha}p_{\alpha}\left\vert
s_{1}^{\hat{n}_{2}\left(  \alpha\right)  }\right\vert ^{2}\left\vert \eta
_{1}\right\rangle \left\langle \eta_{1}\right\vert ,\\
\cdots\\
\tilde{\rho}_{a_{d^{M}}^{\prime}}^{\hat{n}_{2}}=\sum_{\alpha}p_{\alpha
}\left\vert s_{d^{M}}^{\hat{n}_{2}\left(  \alpha\right)  }\right\vert
^{2}\left\vert \eta_{d^{M}}\right\rangle \left\langle \eta_{d^{M}}\right\vert
,\\
\cdots\\
\tilde{\rho}_{a_{1}^{\prime\prime}}^{\hat{n}_{k}}=\sum_{\alpha}p_{\alpha
}\left\vert s_{1}^{\hat{n}_{k}\left(  \alpha\right)  }\right\vert
^{2}\left\vert \eta_{1}\right\rangle \left\langle \eta_{1}\right\vert ,\\
\cdots\\
\tilde{\rho}_{a_{d^{M}}^{\prime\prime}}^{\hat{n}_{k}}=\sum_{\alpha}p_{\alpha
}\left\vert s_{d^{M}}^{\hat{n}_{k}\left(  \alpha\right)  }\right\vert
^{2}\left\vert \eta_{d^{M}}\right\rangle \left\langle \eta_{d^{M}}\right\vert
.
\end{array}
\right.  \label{n21}%
\end{equation}
There are only $d^{M}$ different\ pure states in the quantum result Eq.
(\ref{n21}). It is sufficient to take $\xi$ from $1$ to $d^{M}$ in the LHS
description, e.g.,
\begin{equation}
\left\{
\begin{array}
[c]{l}%
\tilde{\rho}_{a}^{\hat{n}_{1}}=\sum\limits_{\xi=1}^{d^{M}}\wp\left(  a|\hat
{n}_{1},\xi\right)  \wp_{\xi}\rho_{\xi},\\
\tilde{\rho}_{a^{\prime}}^{\hat{n}_{2}}=\sum\limits_{\xi=1}^{d^{M}}\wp\left(
a^{\prime}|\hat{n}_{2},\xi\right)  \wp_{\xi}\rho_{\xi},\\
\cdots\\
\tilde{\rho}_{a^{\prime\prime}}^{\hat{n}_{k}}=\sum\limits_{\xi=1}^{d^{M}}%
\wp\left(  a^{\prime\prime}|\hat{n}_{k},\xi\right)  \wp_{\xi}\rho_{\xi}.
\end{array}
\right.  \label{n22}%
\end{equation}
Since Bob's states are all pure, each $\tilde{\rho}_{a}^{\hat{n}_{\ell}}$ in
Eq. (\ref{n22}) contains only one term. Eq. (\ref{n22}) can be written as
\begin{equation}
\left\{
\begin{array}
[c]{l}%
\tilde{\rho}_{a_{1}}^{\hat{n}_{1}}=\wp_{1}\rho_{1},\\
\cdots\\
\tilde{\rho}_{a_{d^{M}}}^{\hat{n}_{1}}=\wp_{d^{M}}\rho_{d^{M}},\\
\tilde{\rho}_{a_{1}^{\prime}}^{\hat{n}_{2}}=\wp_{1}\rho_{1},\\
\cdots\\
\tilde{\rho}_{a_{d^{M}}^{\prime}}^{\hat{n}_{2}}=\wp_{d^{M}}\rho_{d^{M}},\\
\cdots\\
\tilde{\rho}_{a_{1}^{\prime\prime}}^{\hat{n}_{k}}=\wp_{1}\rho_{1},\\
\cdots\\
\tilde{\rho}_{a_{d^{M}}^{\prime\prime}}^{\hat{n}_{k}}=\wp_{d^{M}}\rho_{d^{M}}.
\end{array}
\right.  \label{n23}%
\end{equation}
Here we describe the same term with the same hidden state. And ${\sum
\limits_{a}}\wp\left(  a|\hat{n},\xi\right)  =1$, so we have $\wp\left(
a_{1}|\hat{n}_{1},1\right)  =\cdots=\wp\left(  a_{d^{M}}|\hat{n}_{1}%
,d^{M}\right)  =\wp\left(  a_{1}^{\prime}|\hat{n}_{2},1\right)  \cdots
=\wp\left(  a_{d^{M}}^{\prime}|\hat{n}_{2},d^{M}\right)  =\cdots=\wp\left(
a_{1}^{\prime\prime}|\hat{n}_{k},1\right)  \cdots=\wp\left(  a_{d^{M}}%
^{\prime\prime}|\hat{n}_{k},d^{M}\right)  =1$, and the others $\wp\left(
a|\hat{n},\xi\right)  =0$. Then we take the sum of Eq. (\ref{n23}) and take
the trace, and get \textquotedblleft$k_{Q}=k_{C}$\textquotedblright. That
means that there is no contradiction if $\{\tilde{\rho}_{a}^{\hat{n}_{1}%
}\}=\{\tilde{\rho}_{a^{\prime}}^{\hat{n}_{2}}\}=\cdots=\{\tilde{\rho
}_{a^{\prime\prime}}^{\hat{n}_{k}}\}$, and we cannot conclude whether Alice
can steer Bob. Consequently, its contrapositive also holds, and if there is a
contradiction \textquotedblleft$k_{Q}=\left(  1+\delta_{k}\right)  _{C}%
$\textquotedblright, then it requests that at least one set of results
$\{\tilde{\rho}_{a}^{\hat{n}_{p}}\}$, with $p\in\left[  1,k\right]  $ is not
identical to the others.
\end{proof}